\documentclass[runningheads]{llncs}
\usepackage{upquote,textcomp}
\usepackage{amsmath, amssymb, thmtools}
\usepackage[misc]{ifsym}
\usepackage{fancyvrb}
\usepackage{xspace}
\usepackage[square,sort,comma,numbers]{natbib}
\usepackage[usenames]{xcolor}

\usepackage{array}
\usepackage{txfonts}
\usepackage{ebproof}
\usepackage{float}
\usepackage{mathpartir}
\usepackage{color}
\usepackage{listings}
\usepackage{enumerate}

\lstdefinelanguage{TTsec}{
  morekeywords={
    if,then,else
  },
  morekeywords=[2]{pure , >>=, label, unlabel, plug, new, read, write},
  morekeywords=[3]{Int, String, Type, Bool, DIO, Labeled, Ref, Label},
  morekeywords=[4]{True, False, L, H},
  otherkeywords={:, =, ., ==},
  literate={
    {->}{{$\to$}}{2}
    {lambda}{{$\lambda$}}{1}
    {->}{{$\to$}}{2}
    {forall}{{$\forall$}}{1}
    {DIO}{DIO$_{\tau}$}{1}
    {Labeled}{Labeled$_{\tau}$}{1}
  },
  basicstyle={\small\ttfamily},
  keywordstyle={\texttt},
  keywordstyle=[2]{\color{blue}\texttt},
  keywordstyle=[3]{\color{darkgray}\texttt},
  keywordstyle=[4]{\color{red}\texttt},
  keepspaces,
  tabsize = 2,
  columns={flexible},
  mathescape = true
}[keywords]
\usepackage{scrextend}
\usepackage{capt-of}
\allowdisplaybreaks

\definecolor{red}{RGB}{199,71,77}
\definecolor{blue}{RGB}{104,72,186}
\definecolor{purple}{RGB}{36,93,190}
\definecolor{green}{RGB}{67,131,67}
\definecolor{gray}{RGB}{134,140,148}
\definecolor{orange}{RGB}{234, 151, 51}

\newcommand{\IdrisData}[1]{\textcolor{orange}{#1}}
\newcommand{\IdrisType}[1]{\textcolor{purple}{#1}}
\newcommand{\IdrisBound}[1]{\textcolor{green}{#1}}
\newcommand{\IdrisFunction}[1]{\textcolor{blue}{#1}}
\newcommand{\IdrisKeyword}[1]{{\textcolor{red}{#1}}}

\fvset{
  commandchars=\\\{\},
  formatcom=\color{gray},
  fontsize=\small
}

\newcommand\DepSec{\textsc{DepSec}\xspace}
\newcommand\MAC{\textbf{MAC}\xspace}

\newcommand\Idris{\textsc{Idris}\xspace}
\newcommand\Labeled[2]{\Verb|\IdrisType{Labeled}|~#1~#2\xspace}
\newcommand\DIO[2]{\Verb|\IdrisType{DIO}|~#1~#2\xspace}

\newcommand\IO[1]{\Verb|\IdrisType{IO}|~#1\xspace}
\newcommand\join{\Verb|\IdrisFunction{join}|\xspace}

\newcommand\MkLabeled{\Verb|\IdrisData{MkLabeled}|\xspace}
\newcommand\MkDIO{\Verb|\IdrisData{MkDIO}|\xspace}
\newcommand\slabel{\Verb|\IdrisFunction{label}|\xspace}
\newcommand\unlabel{\Verb|\IdrisFunction{unlabel}|\xspace}
\newcommand\readFile{\Verb|\IdrisFunction{readFile}|\xspace}
\newcommand\plug{\Verb|\IdrisFunction{plug}|\xspace}
\newcommand\ST{\Verb|\IdrisFunction{ST}|\xspace}

\newcommand\Bool{\texttt{Bool}}
\newcommand\Int{\texttt{Int}}
\newcommand\String{\texttt{String}}
\newcommand\True{\texttt{True}\xspace}
\newcommand\False{\texttt{False}}
\newcommand\ifThenElse[3]{\ensuremath{\texttt{if}\ #1\ \texttt{then}\ #2\ \texttt{else}\ #3}}
\newcommand\DIOValue[1]{\ensuremath{\texttt{DIO}_{v}\ \,#1}}
\newcommand\LabeledValue[1]{\ensuremath{\texttt{Labeled}_{v}\ \,#1}}
\newcommand\LabeledType[2]{\ensuremath{\texttt{Labeled}_{\tau}\ \,#1\ \,#2}}

\newcommand\pure[1]{\ensuremath{\texttt{pure}\ #1}}
\newcommand\bind[2]{\ensuremath{#1 >\!\!>\!\!= #2}}
\newcommand\binds{\ensuremath {>\!\!>\!\!=}}
\newcommand\xlabel[1]{\ensuremath{\texttt{label}\ \,#1}\xspace}
\newcommand\ylabel[1]{\ensuremath{\texttt{label}\ \,#1}\xspace}
\newcommand\evalsto{\ensuremath{\longrightarrow}}
\newcommand\evalstostar{\ensuremath{\longrightarrow^{\ast}}}
\newcommand\bigstepto{\ensuremath{\Downarrow}}
\newcommand\conf[2]{\ensuremath{\langle #1, #2\rangle}}
\newcommand\app[2]{\ensuremath{#1\ #2}}
\newcommand\flowsto{\ensuremath{\sqsubseteq}\xspace}
\newcommand\RefType[2]{\ensuremath{\texttt{Ref}_{\tau}\ \,#1\ #2}}
\newcommand\RefValue[2]{\ensuremath{\texttt{Ref}_{v}^{#1}\ #2}}
\newcommand\newRef[2]{\ensuremath{\texttt{new}^{#1}\ #2}}
\newcommand\readRef[1]{\ensuremath{\texttt{read}\ #1}}
\newcommand\writeRef[2]{\ensuremath{\texttt{write}\ #1\ #2}}
\newcommand\hole{\ensuremath{\bullet}}

\newcommand\plugHole[1]{\ensuremath{\texttt{plug}_{\hole}\ #1}}
\newcommand\Type{\texttt{Type}}
\newcommand\Unit{\texttt{()}}
\newcommand\Term{\textsf{Term}}

\newcommand\Label{\texttt{Label}}

\newcommand\Conf{\textsf{Conf}}
\newcommand\List{\text{List}}
\newcommand\erasure{\ensuremath{\varepsilon_{\ell_{A}}}\xspace}
\newcommand\eps[1]{\erasure\big(#1\big)}
\newcommand\aDIO[2]{\ensuremath{\texttt{DIO}_{\tau}\ \,#1\ \, #2}}
\newcommand\aplug[1]{\ensuremath{\texttt{plug}\ #1}}
\newcommand\aunlabel[1]{\ensuremath{\texttt{unlabel}\ \,#1}}
\newcommand\TBind{\textsc{T-Bind}\xspace}

\newcommand\TLam{\textsc{T-Lam}\xspace}

\newcommand\TApp{\textsc{T-App}\xspace}
\newcommand\TVar{\textsc{T-Var}\xspace}
\newcommand\TForall{\textsc{T-Forall}\xspace}

\newcommand\TT{\textsf{TT}\xspace}
\newcommand\TTsec{\ensuremath{\TT_{sec}}\xspace}
\newcommand\TTsecbullet{\ensuremath{\TT_{sec\bullet}}\xspace}
\newcommand\TUnlabel{\textsc{T-Unlabel}\xspace}
\newcommand\concat{\ensuremath{+\!\!\!+}}

\newcommand\gitlink{\url{https://github.com/simongregersen/DepSec}}
\newcommand\User[1]{\ensuremath{\mathit{U}(\mathit{#1})}\xspace}
\newcommand\Author[2]{\ensuremath{\mathit{A}(\mathit{#1},\mathit{#2})}\xspace}
\newcommand\PC[2]{\ensuremath{\mathit{PC}(\mathit{#1},\mathit{#2})}\xspace}

\usepackage{etoolbox}

\newtoggle{Full}
\toggletrue{Full} 

\begin{document}

\title{A Dependently Typed Library for Static Information-Flow Control in \Idris \iftoggle{Full}{\thanks{This is an extended version of a paper of the same title presented at POST 2019.}}{}}
\author{Simon Gregersen \and S\o ren Eller Thomsen \and Aslan Askarov}
\authorrunning{Gregersen et al.}
\institute{Aarhus University, Aarhus, Denmark \\
\email{\{gregersen,sethomsen,askarov\}@cs.au.dk}}

\maketitle

\begin{abstract}
  Safely integrating third-party code in applications while protecting the confidentiality of information is a long-standing problem.
  Pure functional programming languages, like Haskell, make it possible to enforce lightweight information-flow control through libraries like \MAC by \citeauthor{russo15}.
  This work presents \DepSec, a \MAC inspired, dependently typed library for static information-flow control in \Idris.
  We showcase how adding dependent types increases the expressiveness of state-of-the-art static information-flow control libraries and how \DepSec matches a special-purpose dependent information-flow type system on a key example.
  Finally, we show novel and powerful means of specifying statically enforced declassification policies using dependent types.

  \keywords{Information-Flow Control \and Dependent Types \and Idris.}
\end{abstract}

\section{Introduction}
Modern software applications are increasingly built using libraries and code from multiple third parties.
At the same time, protecting confidentiality of information manipulated by such applications is a growing, yet long-standing problem.
Third-party libraries could in general have been written by anyone and they are usually run with the same privileges as the main application.
While powerful, such privileges open up for abuse.

Traditionally, access control~\cite{bell76} and encryption have been the main means for preventing data dissemination and leakage, however, such mechanisms fall short when third-party code needs access to sensitive information to provide its functionality.
The key observation is that these mechanisms only place restrictions on the access to information but not its propagation.
Once information is accessed, the accessor is free to improperly transmit or leak the information in some form, either by intention or error.

Language-based Information-Flow Control~\cite{sabelfeld03} is a promising technique for enforcing information security.
Traditional enforcement techniques analyze how information at different security levels flows within a program ensuring that information flows only to appropriate places, suppressing illegal flows.
To achieve this, most information-flow control tools require the design of new languages, compilers, or interpreters (e.g.
\cite{myers99,simonet03,liu09,lourenco15,jif,hedin14,chapman04}).
Despite a large, growing body of work on language-based information-flow security, there has been little adoption of the proposed techniques.
For information-flow policies to be enforced in such systems, the whole system has to be written in new languages -- an inherently expensive and time-consuming process for large software systems.
Moreover, in practice, it might very well be that only small parts of an application are governed by information-flow policies.

Pure functional programming languages, like Haskell, have something to offer with respect to information security as they strictly separate side-effect free and side-effectful code. This makes it possible to enforce lightweight information-flow control through libraries~\cite{li06,russo09,stefan11,buiras15,russo15} by constructing an embedded domain-specific security sub-language.
Such libraries enforce a secure-by-construction programming model as any program written against the library interface is not capable of leaking secrets.
This construction forces the programmer to write security-critical code in the sub-language but otherwise allows them to freely interact and integrate with non-security critical code written in the full language.
In particular, static enforcement libraries like \MAC~\cite{russo15} are appealing as no run-time checks are needed and code that exhibits illegal flows is rejected by the type checker at compile-time.
Naturally, the expressiveness of Haskell's type system sets the limitation on which programs can be deemed secure and which information flow policies can be guaranteed.

Dependent type theories \cite{lof84, nordstrom90} are implemented in many programming languages such as Coq~\cite{coq88}, Agda~\cite{norell07}, \Idris \cite{brady11}, and F$^{*}$~\cite{swamy11}.
Programming languages that implement such theories allow types to dependent on values.
This enables programmers to give programs a very precise type and increased confidence in its correctness.

In this paper, we show that dependent types provide a direct and natural way of expressing precise data-dependent security policies.
Dependent types can be used to represent rich security policies in environments like databases and data-centric web applications where, for example, new classes of users and new kinds of data are encountered at run-time and the security level depends on the manipulated data itself \cite{lourenco15}.
Such dependencies are not expressible in less expressive systems like \MAC.
Among other things, with dependent types, we can construct functions where the security level of the output depends on an argument:
\begin{Verbatim}
\IdrisFunction{getPassword} : (\IdrisBound{u} : \IdrisType{Username}) -> \IdrisType{Labeled} \IdrisBound{u} \IdrisType{String}
\end{Verbatim}
Given a user name \Verb|\IdrisBound{u}|, \Verb|\IdrisFunction{getPassword}| retrieves the corresponding password and classifies it at the security level of \Verb|\IdrisBound{u}|.
As such, we can express much more precise security policies that can depend on the manipulated data.

\Idris is a general-purpose functional programming language with full-spectrum dependent types, that is, there is no restrictions on which values may appear in types.
The language is strongly influenced by Haskell and has, among others, inherited its strict encapsulation of side-effects.
\Idris essentially asks the question: ``What if Haskell had full dependent types?''~\cite{brady13}.
This work, essentially, asks:
\begin{quote}
``What if \MAC had full dependent types?''
\end{quote}
We address this question using \Idris because of its positioning as a general-purpose language rather than a proof assistant.
All ideas should be portable to equally expressive systems with full dependent types and strict monadic encapsulation of side-effects.

In summary, the contributions of this paper are as follows.
\begin{itemize}
\item We present \DepSec, a \MAC inspired statically enforced dependently typed information-flow control library for \Idris.
\item We show how adding dependent types strictly increases the expressiveness of state-of-the-art static information-flow control libraries and how \DepSec matches the expressiveness of a special-purpose dependent information-flow type system on a key example.
\item We show how \DepSec enables and aids the construction of policy-parameterized functions that abstract over the security policy.
\item We show novel and powerful means to specify statically-ensured declassification using dependent types for a wide variety of policies.
\item We show progress-insensitive noninterference \cite{askarov08} for the core library in a sequential setting.
\end{itemize}

\paragraph{Outline}
The rest of the paper proceeds through a presentation of the \DepSec library (Section~\ref{sec:depsec}); a conference manager case study (Section~\ref{sec:conference-system}) and the introduction of policy-parameterized functions (Section~\ref{sec:generic-functions}) both showcasing the expressiveness of \DepSec; means to specify statically-ensured declassification policies (Section~\ref{sec:declassification}); soundness of the core library (Section~\ref{sec:soundness}); and related work (Section \ref{sec:related-work}).

All code snippets presented in the following are extracts from the source code. All source code is implemented in \Idris 1.3.1. and available at
\begin{quote}
\gitlink.
\end{quote}
The source code of the core library is also available in Appendix~\ref{app:core-lib}.

\subsection{Assumptions and threat model}\label{sec:assumptions}
In the rest of this paper, we require that code is divided up into trusted code, written by someone we trust, and untrusted code, written by a potential attacker.
The trusted computing base (TCB) has no restrictions, but untrusted code does not have access to modules providing input/output behavior, the data constructors of the domain specific language and a few specific functions related to declassification.
In \Idris, this means that we specifically do not allow access to \Verb|\IdrisType{IO}| functions and \Verb|\IdrisFunction{unsafePerformIO}|.
In \DepSec, constructors and functions marked with a \Verb|TCB| comment are inaccessible to untrusted code.
Throughout the paper we will emphasize when this is the case.

We require that all definitions made by untrusted code are total, that is, defined for all possible inputs and are guaranteed to terminate.
This is necessary if we want to trust proofs given by untrusted code.
Otherwise, it could construct an element of the empty type from which it could prove anything:
\begin{Verbatim}
\IdrisFunction{empty} : \IdrisType{\IdrisType{Void}}
\IdrisFunction{\IdrisFunction{empty}} = \IdrisFunction{\IdrisFunction{empty}}
\end{Verbatim}
In \Idris, this can be checked using the \texttt{--total} compiler flag.
Furthermore, we do not consider concurrency nor any internal or termination covert channels.

\section{The \DepSec library}\label{sec:depsec}

In information-flow control, labels are used to model the sensitivity of data.
Such labels usually form a security lattice~\cite{denning77} where the induced partial ordering $\sqsubseteq$ specifies allowed flows of information and hence the security policy.
For example, $\ell_{1} \sqsubseteq \ell_{2}$ specifies that data with label~$\ell_{1}$ is allowed to flow to entities with label~$\ell_{2}$.
In \DepSec, labels are represented by values that form a verified join semilattice implemented as \Idris interfaces\footnote{Interfaces in \Idris are similar to type classes in Haskell.}. That is, we require proofs of the lattice properties when defining an instance of \Verb!\IdrisType{JoinSemilattice}!.

\begin{Verbatim}
\IdrisKeyword{interface} \IdrisType{\IdrisFunction{\IdrisFunction{\IdrisFunction{\IdrisFunction{\IdrisFunction{\IdrisFunction{\IdrisFunction{\IdrisFunction{\IdrisData{\IdrisData{\IdrisData{\IdrisData{\IdrisData{\IdrisData{\IdrisData{\IdrisData{\IdrisType{\IdrisType{\IdrisType{\IdrisType{\IdrisType{\IdrisType{\IdrisType{\IdrisType{\IdrisType{\IdrisType{\IdrisType{\IdrisType{\IdrisBound{\IdrisBound{\IdrisBound{\IdrisBound{\IdrisBound{\IdrisBound{\IdrisBound{\IdrisBound{\IdrisBound{\IdrisBound{\IdrisBound{\IdrisBound{\IdrisBound{\IdrisBound{\IdrisBound{\IdrisBound{\IdrisBound{\IdrisBound{\IdrisBound{\IdrisBound{\IdrisBound{\IdrisBound{\IdrisBound{\IdrisBound{\IdrisType{JoinSemilattice \IdrisType{\IdrisType{\IdrisType{\IdrisType{\IdrisType{\IdrisType{\IdrisType{\IdrisBound{a}}}}}}}}}}}}}}}}}}}}}}}}}}}}}}}}}}}}}}}}}}}}}}}}}}}}}}}}}}}}}} \IdrisKeyword{where}
  \IdrisFunction{join} : \IdrisBound{\IdrisBound{a}} -> \IdrisBound{\IdrisBound{a}} -> \IdrisBound{\IdrisBound{a}}
  \IdrisFunction{associative} :
    (\IdrisBound{x}, \IdrisBound{y}, \IdrisBound{z} : \IdrisBound{\IdrisBound{\IdrisBound{\IdrisBound{\IdrisBound{\IdrisBound{a}}}}}}) -> \IdrisBound{\IdrisBound{x}} \IdrisFunction{\IdrisFunction{\IdrisBound{`join`}}} (\IdrisBound{\IdrisBound{y}} \IdrisFunction{\IdrisFunction{\IdrisBound{`join`}}} \IdrisBound{\IdrisBound{z}}) \IdrisType{=} (\IdrisBound{\IdrisBound{x}} \IdrisFunction{\IdrisFunction{\IdrisBound{`join`}}} \IdrisBound{\IdrisBound{y}}) \IdrisFunction{\IdrisFunction{\IdrisBound{`join`}}} \IdrisBound{\IdrisBound{z}}
  \IdrisFunction{commutative} : (\IdrisBound{x}, \IdrisBound{y} : \IdrisBound{\IdrisBound{\IdrisBound{\IdrisBound{a}}}})    -> \IdrisBound{\IdrisBound{x}} \IdrisFunction{\IdrisFunction{\IdrisBound{`join`}}} \IdrisBound{\IdrisBound{y}} \IdrisType{=} \IdrisBound{\IdrisBound{y}} \IdrisFunction{\IdrisFunction{\IdrisBound{`join`}}} \IdrisBound{\IdrisBound{x}}
  \IdrisFunction{idempotent}  : (\IdrisBound{x} : \IdrisBound{\IdrisBound{a}})       -> \IdrisBound{\IdrisBound{x}} \IdrisFunction{\IdrisFunction{\IdrisBound{`join`}}} \IdrisBound{\IdrisBound{x}} \IdrisType{=} \IdrisBound{\IdrisBound{x}}
\end{Verbatim}
Dependent function types (often referred to as $\Pi$ types) in \Idris can express such requirements. If \Verb|\IdrisType{A}| is a type and \Verb|\IdrisType{B}| is a type indexed by a value of type \Verb|\IdrisType{{A}}| then \Verb!(\IdrisBound{x} : \IdrisBound{\IdrisType{A}}) -> \IdrisType{B}!
is the type of functions that map arguments \Verb|\IdrisBound{x}| of type \Verb|\IdrisType{A}| to values of type \Verb|\IdrisType{B} \IdrisBound{x}|.

A lattice induces a partial ordering, which gives a direct way to express flow constraints. We introduce a verified partial ordering together with an implementation of this for \Verb|\IdrisType{JoinSemilattice}|.
That is, to define an instance of the \Verb|\IdrisType{Poset}| interface we require a concrete instance of an associated data type \Verb|\IdrisFunction{leq}| as well as proofs of necessary algebraic properties of \Verb|\IdrisFunction{leq}|.

\begin{Verbatim}
\IdrisKeyword{interface} \IdrisType{\IdrisFunction{\IdrisFunction{\IdrisFunction{\IdrisFunction{\IdrisFunction{\IdrisFunction{\IdrisFunction{\IdrisFunction{\IdrisData{\IdrisData{\IdrisData{\IdrisData{\IdrisData{\IdrisData{\IdrisData{\IdrisData{\IdrisType{\IdrisType{\IdrisType{\IdrisType{\IdrisType{\IdrisType{\IdrisType{\IdrisType{\IdrisType{\IdrisType{\IdrisType{\IdrisType{\IdrisBound{\IdrisBound{\IdrisBound{\IdrisBound{\IdrisBound{\IdrisBound{\IdrisBound{\IdrisBound{\IdrisBound{\IdrisBound{\IdrisBound{\IdrisBound{\IdrisBound{\IdrisBound{\IdrisBound{\IdrisBound{\IdrisBound{\IdrisBound{\IdrisBound{\IdrisBound{\IdrisBound{\IdrisBound{\IdrisBound{\IdrisBound{\IdrisType{Poset \IdrisType{\IdrisType{\IdrisType{\IdrisType{\IdrisType{\IdrisType{\IdrisType{\IdrisBound{a}}}}}}}}}}}}}}}}}}}}}}}}}}}}}}}}}}}}}}}}}}}}}}}}}}}}}}}}}}}}}} \IdrisKeyword{where}
  \IdrisFunction{leq} : \IdrisBound{\IdrisBound{a}} -> \IdrisBound{\IdrisBound{a}} -> \IdrisType{Type}
  \IdrisFunction{reflexive}     : (\IdrisBound{x} : \IdrisBound{\IdrisBound{a}}) -> \IdrisBound{\IdrisBound{x}} \IdrisFunction{\IdrisFunction{\IdrisBound{`leq`}}} \IdrisBound{\IdrisBound{x}}
  \IdrisFunction{antisymmetric} : (\IdrisBound{x}, \IdrisBound{y} : \IdrisBound{\IdrisBound{\IdrisBound{\IdrisBound{a}}}}) -> \IdrisBound{\IdrisBound{x}} \IdrisFunction{\IdrisFunction{\IdrisBound{`leq`}}} \IdrisBound{\IdrisBound{y}} -> \IdrisBound{\IdrisBound{y}} \IdrisFunction{\IdrisFunction{\IdrisBound{`leq`}}} \IdrisBound{\IdrisBound{x}} -> \IdrisBound{\IdrisBound{x}} \IdrisType{=} \IdrisBound{\IdrisBound{y}}
  \IdrisFunction{transitive}    : (\IdrisBound{x}, \IdrisBound{y}, \IdrisBound{z} : \IdrisBound{\IdrisBound{\IdrisBound{\IdrisBound{\IdrisBound{\IdrisBound{a}}}}}}) -> \IdrisBound{\IdrisBound{x}} \IdrisFunction{\IdrisFunction{\IdrisBound{`leq`}}} \IdrisBound{\IdrisBound{y}} -> \IdrisBound{\IdrisBound{y}} \IdrisFunction{\IdrisFunction{\IdrisBound{`leq`}}} \IdrisBound{\IdrisBound{z}} -> \IdrisBound{\IdrisBound{x}} \IdrisFunction{\IdrisFunction{\IdrisBound{`leq`}}} \IdrisBound{\IdrisBound{z}}

\IdrisKeyword{implementation} \IdrisType{\IdrisType{\IdrisType{\IdrisType{\IdrisType{\IdrisType{\IdrisType{\IdrisType{\IdrisType{\IdrisType{\IdrisFunction{\IdrisFunction{\IdrisFunction{\IdrisFunction{\IdrisFunction{\IdrisFunction{\IdrisFunction{\IdrisFunction{\IdrisFunction{\IdrisFunction{\IdrisFunction{\IdrisFunction{\IdrisFunction{\IdrisFunction{\IdrisFunction{\IdrisFunction{\IdrisFunction{\IdrisFunction{\IdrisFunction{\IdrisFunction{\IdrisFunction{\IdrisFunction{\IdrisData{\IdrisData{\IdrisBound{\IdrisBound{\IdrisBound{\IdrisBound{\IdrisBound{\IdrisBound{\IdrisBound{\IdrisBound{\IdrisBound{\IdrisBound{\IdrisBound{\IdrisBound{\IdrisBound{\IdrisBound{\IdrisBound{\IdrisBound{\IdrisBound{\IdrisBound{\IdrisBound{\IdrisBound{\IdrisBound{\IdrisBound{\IdrisBound{JoinSemilattice \IdrisBound{\IdrisBound{\IdrisBound{\IdrisBound{\IdrisBound{\IdrisBound{\IdrisBound{\IdrisBound{\IdrisBound{\IdrisBound{a}}}}}}}}}} => \IdrisType{\IdrisType{Poset}} \IdrisBound{\IdrisBound{\IdrisBound{\IdrisBound{\IdrisBound{\IdrisBound{\IdrisBound{\IdrisBound{\IdrisBound{\IdrisBound{\IdrisBound{\IdrisBound{\IdrisBound{\IdrisBound{\IdrisBound{\IdrisBound{\IdrisBound{\IdrisBound{\IdrisBound{\IdrisBound{\IdrisBound{\IdrisBound{a}}}}}}}}}}}}}}}}}}}}}}}}}}}}}}}}}}}}}}}}}}}}}}}}}}}}}}}}}}}}}}}}}}}}}}}}}}}}}}} \IdrisKeyword{where}
  \IdrisFunction{\IdrisFunction{\IdrisBound{\IdrisFunction{leq} \IdrisBound{x} \IdrisBound{y}}}} = (\IdrisBound{\IdrisBound{x}} \IdrisFunction{\IdrisFunction{`join`}} \IdrisBound{\IdrisBound{y}} \IdrisType{=} \IdrisBound{\IdrisBound{y}})
  ...
\end{Verbatim}
This definition allows for generic functions to impose as few restrictions as possible on the user while being able to exploit the algebraic structure in proofs, as will become evident in Section~\ref{sec:conference-system} and~\ref{sec:generic-functions}.
For the sake of the following case studies, we also have a definition of a \Verb!\IdrisType{BoundedJoinSemilattice}!
requiring a least element \Verb|\IdrisFunction{Bottom}| of an instance of \Verb!\IdrisType{JoinSemilattice}!
and a proof of the element being the unit.
\begin{figure}[htb!]
  \centering
\begin{Verbatim}
\IdrisKeyword{data} \IdrisType{Labeled} : \IdrisBound{label} -> \IdrisType{Type} -> \IdrisType{Type} \IdrisKeyword{where}
  \IdrisData{MkLabeled} : \IdrisBound{\IdrisBound{valueType}} -> \IdrisType{\IdrisType{Labeled}} \IdrisBound{\IdrisBound{label}} \IdrisBound{\IdrisBound{valueType}} -- TCB

\IdrisKeyword{data} \IdrisType{DIO} : \IdrisBound{\IdrisBound{l}} -> \IdrisType{Type} -> \IdrisType{Type} \IdrisKeyword{where}
  \IdrisData{MkDIO} : \IdrisFunction{\IdrisFunction{IO}} \IdrisBound{\IdrisBound{valueType}} -> \IdrisType{\IdrisType{DIO}} \IdrisBound{\IdrisBound{l}} \IdrisBound{\IdrisBound{valueType}} -- TCB

\IdrisType{\IdrisType{\IdrisFunction{\IdrisFunction{\IdrisFunction{\IdrisFunction{\IdrisFunction{\IdrisFunction{\IdrisData{\IdrisData{\IdrisBound{\IdrisBound{\IdrisBound{\IdrisBound{\IdrisBound{\IdrisBound{\IdrisBound{\IdrisBound{\IdrisBound{Monad (\IdrisType{\IdrisType{\IdrisType{\IdrisType{\IdrisType{\IdrisType{\IdrisType{\IdrisType{\IdrisType{\IdrisType{\IdrisType{\IdrisType{\IdrisType{\IdrisType{DIO}}}}}}}}}}}}}} \IdrisBound{\IdrisBound{\IdrisBound{\IdrisBound{\IdrisBound{\IdrisBound{\IdrisBound{\IdrisBound{\IdrisBound{\IdrisBound{\IdrisBound{\IdrisBound{\IdrisBound{\IdrisBound{l}}}}}}}}}}}}}})}}}}}}}}}}}}}}}}}}} \IdrisKeyword{where}
  ...

\IdrisFunction{label} : \IdrisType{\IdrisType{Poset}} \IdrisBound{\IdrisBound{label}} => \{\IdrisBound{l} : \IdrisBound{\IdrisBound{label}}\} -> \IdrisBound{\IdrisBound{a}} -> \IdrisType{\IdrisType{Labeled}} \IdrisBound{\IdrisBound{l}} \IdrisBound{\IdrisBound{a}}

\IdrisFunction{unlabel} : \IdrisType{\IdrisType{Poset}} \IdrisBound{\IdrisBound{label}} => \{\IdrisBound{l}, \IdrisBound{l'} : \IdrisBound{\IdrisBound{\IdrisBound{\IdrisBound{label}}}}\}
       -> \{\IdrisKeyword{auto} \IdrisBound{flow} : \IdrisBound{\IdrisBound{l}} \IdrisFunction{\IdrisFunction{`leq`}} \IdrisBound{\IdrisBound{l'}}\}
       -> \IdrisType{\IdrisType{Labeled}} \IdrisBound{\IdrisBound{l}} \IdrisBound{\IdrisBound{a}}
       -> \IdrisType{\IdrisType{DIO}} \IdrisBound{\IdrisBound{l'}} \IdrisBound{\IdrisBound{a}}

\IdrisFunction{plug} : \IdrisType{\IdrisType{Poset}} \IdrisBound{\IdrisBound{label}} => \{\IdrisBound{l}, \IdrisBound{l'} : \IdrisBound{\IdrisBound{\IdrisBound{\IdrisBound{label}}}}\}
    -> \IdrisType{\IdrisType{DIO}} \IdrisBound{\IdrisBound{l'}} \IdrisBound{\IdrisBound{a}}
    -> \{\IdrisKeyword{auto} \IdrisBound{flow} : \IdrisBound{\IdrisBound{l}} \IdrisFunction{\IdrisFunction{`leq`}} \IdrisBound{\IdrisBound{l'}}\}
    -> \IdrisType{\IdrisType{DIO}} \IdrisBound{\IdrisBound{l}} (\IdrisType{\IdrisType{Labeled}} \IdrisBound{\IdrisBound{l'}} \IdrisBound{\IdrisBound{a}})

\IdrisFunction{run} : \IdrisType{DIO} \IdrisBound{l} \IdrisBound{a} -> \IdrisFunction{IO} \IdrisBound{a} -- TCB

\IdrisFunction{lift} : \IdrisFunction{IO} \IdrisBound{a} -> \IdrisType{DIO} \IdrisBound{l} \IdrisBound{a} -- TCB
\end{Verbatim}
  \caption{Type signature of the core \DepSec API. }
  \label{fig:depsec}
\end{figure}
\paragraph{The Core API} Figure \ref{fig:depsec} presents the type signature of \DepSec's core API.
Notice that names beginning with a lower case letter that appear as a parameter or index in a type declaration will be automatically bound as an implicit argument in \Idris, and the \Verb|\IdrisKeyword{auto}| annotation on implicit arguments means that \Idris will attempt to fill in the implicit argument by searching the calling context for an appropriate value.

Abstract data type \Labeled{$\ell$}{$a$} denotes a value of type $a$ with sensitivity level $\ell$.
We say that \Labeled{$\ell$}{$a$} is \emph{indexed} by $\ell$ and \emph{parameterized} by $a$.
Abstract data type \DIO{$\ell$}{$a$} denotes a secure computation that handles values with sensitivity level $\ell$ and results in a value of type $a$.
It is internally represented as a wrapper around the regular \IO{\hspace{-0.35em}} monad that, similar to the one in Haskell, can be thought of as a state monad where the state is the entire world.
Notice that both data constructors \MkLabeled and \MkDIO are not available to untrusted code as this would allow pattern matching and uncontrolled unwrapping of protected entities.
As a consequence, we introduce functions \slabel and \unlabel for labeling and unlabeling values.
Like \citet{rajani18}, but unlike \MAC, the type signature of \slabel imposes no lattice constraints on the computation context.
This does not leak information as, if $l \sqsubseteq l'$ and a computation $c$ has type \DIO{$l'$}{(\Labeled{$l$}{$V$})} for any type $V$, then there is no way for the labeled return value of $c$ to escape the computation context with label~$l'$.

As in \MAC, the API contains a function \plug that safely integrates sensitive computations into less sensitive ones.
This avoids the need for nested computations and \emph{label creep}, that is, the raising of the current label to a point where the computation can no longer perform useful tasks~\cite{russo15,vassena18}.
Finally, we also have functions \Verb|\IdrisFunction{run}| and \Verb|\IdrisFunction{lift}| that are only available to trusted code for unwrapping of the \DIO{$\ell$}{\hspace{-0.2em}} monad and lifting of the \Verb|\IdrisFunction{IO}| monad into the \DIO{$\ell$}{\hspace{-0.2em}} monad.

\paragraph{Labeled resources}
Data type \Labeled{$\ell$}{$a$} is used to denote a labeled \Idris value with type $a$. This is an example of a \emph{labeled resource}~\cite{russo15}.
By itself, the core library does not allow untrusted code to perform any side effects but we can safely incorporate, for example, file access and mutable references as other labeled resources.
Figure~\ref{fig:fileIO} presents type signatures for files indexed by security levels used for secure file handling while mutable references are available in Appendix~\ref{app:core-lib}.
Abstract data type \Verb|\IdrisType{SecFile}|~$\ell$ denotes a secure file with sensitivity level~$\ell$.
As for \Labeled{$\ell$}{$a$}, the data constructor \Verb|\IdrisData{MkSecFile}| is not available to untrusted code.

The function \Verb|\IdrisFunction{readFile}| takes as input a secure file \Verb|\IdrisType{\IdrisType{SecFile}} \IdrisBound{\IdrisBound{l'}}| and returns a computation with sensitivity level \Verb|\IdrisBound{l}| that returns a labeled value with sensitivity level \Verb|\IdrisBound{l'}|.
Notice that the \Verb|\IdrisBound{\IdrisBound{l}}|~$\sqsubseteq$~\Verb|\IdrisBound{\IdrisBound{l'}}| flow constraint is required to enforce the \emph{no read-up} policy \cite{bell76}. That is, the result of the computation returned by \Verb|\IdrisFunction{readFile}| only involves data with sensitivity at most \Verb|\IdrisBound{l}|.
The function \Verb|\IdrisFunction{writeFile}| takes as input a secure file \Verb|\IdrisType{\IdrisType{SecFile}} \IdrisBound{\IdrisBound{l''}}| and a labeled value of sensitivity level \Verb|\IdrisBound{l'}|, and it returns a computation with sensitivity level \Verb|\IdrisBound{l}| that returns a labeled value with sensitivity level \Verb|\IdrisBound{l''}|.
Notice that both the \Verb|\IdrisBound{\IdrisBound{l}}|~$\sqsubseteq$~\Verb|\IdrisBound{\IdrisBound{l'}}| and \Verb|\IdrisBound{\IdrisBound{l'}}|~$\sqsubseteq$~\Verb| \IdrisBound{\IdrisBound{l''}}| flow constraints are required, essentially enforcing the \emph{no write-down} policy \cite{bell76}, that is, the file never receives data more sensitive than its sensitivity level.

Finally, notice that the standard library functions for reading and writing files in \Idris used to implement the functions in Figure~\ref{fig:fileIO} do not raise exceptions.
Rather, both functions return an instance of the sum type \Verb|\IdrisType{Either}|.
We stay consistent with \Idris' choice for this instead of adding exception handling as done in \MAC.

\begin{figure}[htb!]
  \centering
\begin{Verbatim}
\IdrisKeyword{data} \IdrisType{SecFile} : \{\IdrisBound{label} : \IdrisType{Type}\} -> (\IdrisBound{l} : \IdrisBound{\IdrisBound{label}}) -> \IdrisType{Type} \IdrisKeyword{where}
  \IdrisData{MkSecFile} : (\IdrisBound{path} : \IdrisType{String}) -> \IdrisType{\IdrisType{SecFile}} \IdrisBound{\IdrisBound{l}} -- TCB

\IdrisFunction{readFile} : \IdrisType{\IdrisType{Poset}} \IdrisBound{\IdrisBound{label}} => \{\IdrisBound{l},\IdrisBound{l'} : \IdrisBound{\IdrisBound{\IdrisBound{\IdrisBound{label}}}}\}
        -> \{\IdrisKeyword{auto} \IdrisBound{flow} : \IdrisBound{\IdrisBound{l}} \IdrisFunction{\IdrisFunction{`leq`}} \IdrisBound{\IdrisBound{l'}}\}
        -> \IdrisType{\IdrisType{SecFile}} \IdrisBound{\IdrisBound{l'}}
        -> \IdrisType{\IdrisType{DIO}} \IdrisBound{\IdrisBound{l}} (\IdrisType{\IdrisType{Labeled}} \IdrisBound{\IdrisBound{l'}} (\IdrisType{\IdrisType{Either}} \IdrisType{\IdrisType{FileError}} \IdrisType{String}))

\IdrisFunction{writeFile} : \IdrisType{\IdrisType{Poset}} \IdrisBound{\IdrisBound{label}} => \{\IdrisBound{l},\IdrisBound{l'},\IdrisBound{l''} : \IdrisBound{\IdrisBound{\IdrisBound{\IdrisBound{\IdrisBound{\IdrisBound{label}}}}}}\}
         -> \{\IdrisKeyword{auto} \IdrisBound{flow} : \IdrisBound{\IdrisBound{l}} \IdrisFunction{\IdrisFunction{`leq`}} \IdrisBound{\IdrisBound{l'}}\} -> \{\IdrisKeyword{auto} \IdrisBound{flow'} : \IdrisBound{\IdrisBound{l'}} \IdrisFunction{\IdrisFunction{`leq`}} \IdrisBound{\IdrisBound{l''}}\}
         -> \IdrisType{\IdrisType{SecFile}} \IdrisBound{\IdrisBound{l''}}
         -> \IdrisType{\IdrisType{Labeled}} \IdrisBound{\IdrisBound{l'}} \IdrisType{String}
         -> \IdrisType{\IdrisType{DIO}} \IdrisBound{\IdrisBound{l}} (\IdrisType{\IdrisType{Labeled}} \IdrisBound{\IdrisBound{l''}} (\IdrisType{\IdrisType{Either}} \IdrisType{\IdrisType{FileError}} \IdrisType{()}))
\end{Verbatim}
  \caption{Type signatures for secure file handling.}
  \label{fig:fileIO}
\end{figure}

\section[Case study: Conference manager system]{Case study: Conference manager system} \label{sec:conference-system}

This case study showcases the expressiveness of \DepSec by reimplementing a conference manager system with a fine-grained data-dependent security policy introduced by \citet{lourenco15}.
\citeauthor{lourenco15} base their development on a minimal $\lambda$-calculus with references and collections and they show how secure operations on relevant scenarios can be modelled and analysed using \emph{dependent information flow types}.
Our reimplementation demonstrates how \DepSec matches the expressiveness of such a special-purpose built dependent type system on a key example.

In this scenario, a user is either a regular user, an author user, or a program committee (PC) member.
The conference manager contains information about the users, their submissions, and submission reviews.
This data is stored in lists of references to records, and the goal is to statically ensure, by typing, the confidentiality of the data stored in the conference manager system.
As such, the security policy is:
\begin{itemize}
\item A registered user's information is not observable by other users.
\item The content of a paper can be seen by its authors as well as its reviewers.
\item Comments to the PC of a submission's review can only be seen by other members that are also reviewers of that submission.
\item The only authors that are allowed to see the grade and the review of the submission are those that authored that submission.
\end{itemize}

To achieve this security policy, \citeauthor{lourenco15} make use of indexed security labels~\cite{liu09}.
The security level $\mathit{U}$ is partitioned into a number of security compartments such that $\User{uid}$ represents the compartment of the registered user with id $\mathit{uid}$.
Similarly, the security level $\mathit{A}$ is indexed such that $\Author{uid}{sid}$ stands for the compartment of data belonging to author $\mathit{uid}$ and their submission $\mathit{sid}$, and $\mathit{PC}$ is indexed such that $\PC{uid}{sid}$ stands for data belonging to the PC member with user id $\mathit{uid}$ assigned to review the submission with id $\mathit{sid}$.
Furthermore, levels $\top$ and $\bot$ are introduced such that, for example, $\User{\ensuremath{\bot}} \sqsubseteq \User{uid} \sqsubseteq \User{\ensuremath{\top}}$.
Now, the security lattice is defined using two equations:
\begin{align}
 \forall \mathit{uid}, \mathit{sid}.\ \User{{uid}} &\sqsubseteq \Author{{uid}}{{sid}} \label{eq:UA}\\
  \forall \mathit{uid1}, \mathit{uid2}, \mathit{sid}.\ \Author{{uid1}}{{sid}} &\sqsubseteq \PC{{uid2}}{{sid}} \label{eq:APC}
\end{align}
\citeauthor{lourenco15} are able to type a list of submissions with a dependent sum type that assigns the content of the paper the security level $\Author{uid}{sid}$, where $\mathit{uid}$ and $\mathit{sid}$ are fields of the record.
For example, if a concrete submission with identifier $2$ was made by the user with identifier $1$, the content of the paper gets classified at security level $\Author{1}{2}$.
In consequence, $\Author{1}{2} \sqsubseteq \PC{\ensuremath{n}}{2}$ for any $\mathit{uid}$ $n$ and the content of the paper is only observable by its assigned reviewers.
Similar types are given for the list of user information and the list of submission reviews, enforcing the security policy described in the above.

To express this policy in \DepSec, we introduce abstract data types \Verb|\IdrisType{Id}| and \Verb|\IdrisType{Compartment}| (cf. Figure \ref{fig:conference-id-compartment}) followed by an implementation of the \Verb|\IdrisType{BoundedJoinSemilattice}| interface that satisfies equations \eqref{eq:UA} and \eqref{eq:APC}.
\begin{figure}[htb!]
  \centering
  \begin{minipage}{0.4\linewidth}
\begin{Verbatim}
\IdrisKeyword{data} \IdrisType{Id} : \IdrisType{Type} \IdrisKeyword{where}
  \IdrisData{Top} : \IdrisType{\IdrisType{Id}}
  \IdrisData{Nat} : \IdrisType{\IdrisType{Nat}} -> \IdrisType{\IdrisType{Id}}
  \IdrisData{Bot} : \IdrisType{\IdrisType{Id}}
\end{Verbatim}
  \end{minipage}
  \begin{minipage}{0.4\linewidth}
\begin{Verbatim}
\IdrisKeyword{data} \IdrisType{Compartment} : \IdrisType{Type} \IdrisKeyword{where}
  \IdrisData{U}  : \IdrisFunction{\IdrisFunction{Id}} -> \IdrisType{\IdrisType{Compartment}}
  \IdrisData{A}  : \IdrisFunction{\IdrisFunction{Id}} -> \IdrisFunction{\IdrisFunction{Id}} -> \IdrisType{\IdrisType{Compartment}}
  \IdrisData{PC} : \IdrisFunction{\IdrisFunction{Id}} -> \IdrisFunction{\IdrisFunction{Id}} -> \IdrisType{\IdrisType{Compartment}}
\end{Verbatim}
  \end{minipage}
  \caption{Abstract data types for the conference manager sample security lattice.}
  \label{fig:conference-id-compartment}
\end{figure}

Using the above, the required dependent sum types can easily be encoded with \DepSec in \Idris as presented in Figure \ref{fig:conference-records}.
\begin{figure}[htb!]
  \centering
  \begin{minipage}{0.49\linewidth}
\begin{Verbatim}
\IdrisKeyword{record} \IdrisType{User} \IdrisKeyword{where}
  \IdrisKeyword{constructor}\IdrisData{ MkUser}
  \IdrisBound{uid}   : \IdrisType{Id}
  \IdrisBound{name}  : \IdrisType{\IdrisType{\IdrisType{\IdrisType{Labeled}}}} (\IdrisData{\IdrisData{\IdrisData{\IdrisData{U}}}} \IdrisFunction{\IdrisFunction{\IdrisBound{\IdrisBound{uid}}}}) \IdrisType{\IdrisType{String}}
  \IdrisBound{univ}  : \IdrisType{\IdrisType{\IdrisType{\IdrisType{Labeled}}}} (\IdrisData{\IdrisData{\IdrisData{\IdrisData{U}}}} \IdrisFunction{\IdrisFunction{\IdrisBound{\IdrisBound{uid}}}}) \IdrisType{\IdrisType{String}}
  \IdrisBound{email} : \IdrisType{\IdrisType{\IdrisType{\IdrisType{Labeled}}}} (\IdrisData{\IdrisData{\IdrisData{\IdrisData{U}}}} \IdrisFunction{\IdrisFunction{\IdrisBound{\IdrisBound{uid}}}}) \IdrisType{\IdrisType{String}}

\IdrisKeyword{record} \IdrisType{Submission} \IdrisKeyword{where}
  \IdrisKeyword{constructor}\IdrisData{ MkSubmission}
  \IdrisBound{uid}   : \IdrisType{Id}
  \IdrisBound{sid}   : \IdrisType{Id}
  \IdrisBound{title} : \IdrisType{\IdrisType{\IdrisType{\IdrisType{Labeled}}}} (\IdrisData{\IdrisData{\IdrisData{\IdrisData{A}}}} \IdrisFunction{\IdrisFunction{\IdrisBound{\IdrisBound{uid}}}} \IdrisFunction{\IdrisFunction{\IdrisBound{\IdrisBound{sid}}}}) \IdrisType{\IdrisType{String}}
  \IdrisBound{abs}   : \IdrisType{\IdrisType{\IdrisType{\IdrisType{Labeled}}}} (\IdrisData{\IdrisData{\IdrisData{\IdrisData{A}}}} \IdrisFunction{\IdrisFunction{\IdrisBound{\IdrisBound{uid}}}} \IdrisFunction{\IdrisFunction{\IdrisBound{\IdrisBound{sid}}}}) \IdrisType{\IdrisType{String}}
  \IdrisBound{paper} : \IdrisType{\IdrisType{\IdrisType{\IdrisType{Labeled}}}} (\IdrisData{\IdrisData{\IdrisData{\IdrisData{A}}}} \IdrisFunction{\IdrisFunction{\IdrisBound{\IdrisBound{uid}}}} \IdrisFunction{\IdrisFunction{\IdrisBound{\IdrisBound{sid}}}}) \IdrisType{\IdrisType{String}}
\end{Verbatim}
  \end{minipage}
  \begin{minipage}{0.49\linewidth}
\begin{Verbatim}
\IdrisKeyword{record} \IdrisType{Review} \IdrisKeyword{where}
  \IdrisKeyword{constructor}\IdrisData{ MkReview}
  \IdrisBound{uid}     : \IdrisType{Id}
  \IdrisBound{sid}     : \IdrisType{Id}
  \IdrisBound{PC_only} : \IdrisType{\IdrisType{\IdrisType{\IdrisType{Labeled}}}} (\IdrisData{\IdrisData{\IdrisData{\IdrisData{PC}}}} \IdrisData{\IdrisData{\IdrisData{\IdrisBound{uid}}}} \IdrisFunction{\IdrisFunction{\IdrisBound{\IdrisBound{sid}}}}) \IdrisType{\IdrisType{String}}
  \IdrisBound{review}  : \IdrisType{\IdrisType{\IdrisType{\IdrisType{Labeled}}}} (\IdrisData{\IdrisData{\IdrisData{\IdrisData{A}}}} \IdrisData{\IdrisData{\IdrisData{\IdrisData{Top}}}} \IdrisFunction{\IdrisFunction{\IdrisBound{\IdrisBound{sid}}}}) \IdrisType{\IdrisType{String}}
  \IdrisBound{grade}   : \IdrisType{\IdrisType{\IdrisType{\IdrisType{Labeled}}}} (\IdrisData{\IdrisData{\IdrisData{\IdrisData{A}}}} \IdrisData{\IdrisData{\IdrisData{\IdrisData{Top}}}} \IdrisFunction{\IdrisFunction{\IdrisBound{\IdrisBound{sid}}}}) \IdrisType{\IdrisType{Integer}}
\end{Verbatim}
  \end{minipage}
  \caption{Conference manager types encoded with \DepSec.}
  \label{fig:conference-records}
\end{figure}
With these typings in place, implementing the examples from \citet{lourenco15} is straightforward.
For example, the function \Verb|\IdrisFunction{viewAuthorPapers}| takes as input a list of submissions and a user identifier \Verb|\IdrisBound{uid1}| from which it returns a computation that returns a list of submissions authored by the user with identifier \Verb|\IdrisBound{uid1}|.
Notice that \Verb|\IdrisFunction{uid}| denotes the automatically generated record projection function that retrieves the field \Verb|\IdrisBound{uid}| of the record, and that \Verb|\IdrisType{(\IdrisBound{x}: A ** B)}| is notation for a dependent pair (often referred to as a $\Sigma$ type) where \Verb|\IdrisType{A}| and \Verb|\IdrisType{B}| are types and \Verb|\IdrisType{B}| may depend on \Verb|\IdrisBound{x}|.

\begin{Verbatim}
\IdrisFunction{viewAuthorPapers} : \IdrisFunction{\IdrisFunction{Submissions}}
                -> (\IdrisBound{uid1} : \IdrisFunction{\IdrisType{Id}})
                -> \IdrisType{\IdrisType{DIO}} \IdrisFunction{\IdrisFunction{Bottom}} (\IdrisType{\IdrisType{List}} \IdrisType{(\IdrisBound{sub}} \IdrisType{:} \IdrisType{\IdrisType{Submission}} \IdrisType{**} \IdrisBound{\IdrisBound{uid1}} \IdrisType{=} (\IdrisFunction{\IdrisFunction{uid}} \IdrisBound{\IdrisBound{sub}})\IdrisType{)})
\end{Verbatim}
The \Verb|\IdrisFunction{addCommentSubmission}| operation is used by the PC members to add comments to the submissions.
The function takes as input a list of reviews, a user identifier of a PC member, a submission identifier, and a comment with label \Verb|\IdrisData{\IdrisData{A}} \IdrisBound{\IdrisBound{uid1}} \IdrisBound{\IdrisBound{sid1}}|.
It returns a computation that updates the \Verb|\IdrisBound{PC\_only}| field in the review of the paper with identifier \Verb|\IdrisBound{sid1}|.
\begin{Verbatim}
\IdrisFunction{addCommentSubmission} : \IdrisFunction{\IdrisFunction{Reviews}} -> (\IdrisBound{uid1} : \IdrisFunction{\IdrisType{Id}}) -> (\IdrisBound{sid1} : \IdrisFunction{\IdrisType{Id}})
                    -> \IdrisType{\IdrisType{Labeled}} (\IdrisData{\IdrisData{A}} \IdrisBound{\IdrisBound{uid1}} \IdrisBound{\IdrisBound{sid1}}) \IdrisType{String}
                    -> \IdrisType{\IdrisType{DIO}} \IdrisFunction{\IdrisFunction{Bottom}} \IdrisType{()}
\end{Verbatim}
Notice that to implement this specific type signature, up-classification is necessary to assign the comment with type \Verb|\IdrisType{\IdrisType{Labeled}} (\IdrisData{\IdrisData{A}} \IdrisBound{\IdrisBound{uid1}} \IdrisBound{\IdrisBound{sid1}}) \IdrisType{String}| to a field with type \Verb|\IdrisType{\IdrisType{\IdrisType{\IdrisType{Labeled}}}} (\IdrisData{\IdrisData{\IdrisData{\IdrisData{PC}}}} \IdrisData{\IdrisData{\IdrisData{\IdrisBound{uid}}}} \IdrisFunction{\IdrisFunction{\IdrisBound{\IdrisBound{sid1}}}}) \IdrisType{\IdrisType{String}}|.
This can be achieved soundly with the \Verb|\IdrisFunction{relabel}| primitive introduced by Vassena~et~al.~\cite{vassena18} as \Verb|{\IdrisData{A} \IdrisBound{uid1} \IdrisBound{sid1}}|~$\sqsubseteq$~\Verb|\IdrisData{\IdrisData{\IdrisData{\IdrisData{PC}}}} \IdrisData{\IdrisData{\IdrisData{\IdrisBound{uid}}} \IdrisFunction{\IdrisFunction{\IdrisBound{\IdrisBound{sid1}}}}}|.
We include this primitive in Appendix~\ref{app:core-lib}.

Several other examples are available in the accompanying source code. The entire case study amounts to about 300 lines of code where half of the lines implement and verify the lattice.

\section{Policy-parameterized functions}\label{sec:generic-functions}
A consequence of using a dependently typed language, and the design of \DepSec, is that functions can be defined such that they abstract over the security policy while retaining precise security levels.
This makes it possible to reuse code across different applications and write other libraries on top of \DepSec.
We can exploit the existence of a lattice \join, the induced poset, and their verified algebraic properties to write such functions.

\begin{figure}[htb!]
  \centering
\begin{Verbatim}
\IdrisFunction{readTwoFiles} : \IdrisType{BoundedJoinSemilattice} \IdrisBound{label}
            => \{\IdrisBound{l}, \IdrisBound{l'} : \IdrisBound{label}\}
            -> \IdrisType{SecFile} \IdrisBound{l}
            -> \IdrisType{SecFile} \IdrisBound{l'}
            -> \IdrisType{DIO} \IdrisFunction{Bottom} (\IdrisType{Labeled} (\IdrisFunction{join} \IdrisBound{l} \IdrisBound{l'}) (\IdrisType{Either} \IdrisType{FileError} \IdrisType{String}))
\IdrisFunction{readTwoFiles} \IdrisBound{file1} \IdrisBound{file2} \{\IdrisBound{l}\} \{\IdrisBound{l'}\} =
  \IdrisKeyword{do} \IdrisBound{file1'} <- \IdrisFunction{readFile} \{flow = \IdrisFunction{leq_bot_x} \IdrisBound{l}\} \IdrisBound{file1}
     \IdrisBound{file2'} <- \IdrisFunction{readFile} \{flow = \IdrisFunction{leq_bot_x} \IdrisBound{l'}\} \IdrisBound{file2}
     \IdrisKeyword{let} \IdrisBound{dio} : \IdrisType{DIO} (\IdrisFunction{join} \IdrisBound{l} \IdrisBound{l'}) (\IdrisType{Either} \IdrisType{FileError} \IdrisType{String})
       = \IdrisKeyword{do} \IdrisBound{c1} <- \IdrisFunction{unlabel} \{flow = \IdrisFunction{join_x_xy} \IdrisBound{l} \IdrisBound{l'}\} \IdrisBound{file1'}
            \IdrisBound{c2} <- \IdrisFunction{unlabel} \{flow = \IdrisFunction{join_y_xy} \IdrisBound{l} \IdrisBound{l'}\} \IdrisBound{file2'}
            \IdrisFunction{pure} $ \IdrisKeyword{case} \IdrisData{(\IdrisBound{c1\IdrisData{,}}} \IdrisBound{c2\IdrisData{)}} \IdrisKeyword{of}
                        \IdrisData{(\IdrisData{Right}} \IdrisBound{c1'\IdrisData{,}} \IdrisData{Right} \IdrisBound{c2'\IdrisData{)}} => \IdrisData{Right} $ \IdrisBound{c1'} \IdrisFunction{++} \IdrisBound{c2'}
                        \IdrisData{(\IdrisData{Left}} \IdrisBound{e1\IdrisData{,}} _\IdrisData{)} => \IdrisData{Left} \IdrisBound{e1}
                        \IdrisData{(_\IdrisData{,}} \IdrisData{Left} \IdrisBound{e2\IdrisData{)}} => \IdrisData{Left} \IdrisBound{e2}
     \IdrisFunction{plug} \{flow = \IdrisFunction{leq_bot_x} (\IdrisFunction{join} \IdrisBound{l} \IdrisBound{l'})\} \IdrisBound{dio}
\end{Verbatim}
  \caption{Reading two files to a string labeled with the join of the labels of the files.}
  \label{fig:readTwoFiles}
\end{figure}

Figure \ref{fig:readTwoFiles} presents the function \Verb|\IdrisFunction{readTwoFiles}| that is parameterized by a bounded join semilattice.
It takes two secure files with labels \Verb|\IdrisBound{l}| and \Verb|\IdrisBound{l'}| as input and returns a computation that concatenates the contents of the two files labeled with the join of \Verb|\IdrisBound{l}| and \Verb|\IdrisBound{l'}|.
To implement this, we make use of the \unlabel and \readFile primitives from Figure \ref{fig:depsec} and \ref{fig:fileIO}, respectively.
This computation unlabels the contents of the files and returns the concatenation of the contents if no file error occurred.
Notice that \Verb|\IdrisFunction{pure}| is the \Idris function for monadic return, corresponding to the \Verb|\IdrisFunction{return}| function in Haskell.
Finally, this computation is plugged into the surrounding computation.
Notice how the usage of \readFile and \unlabel introduces several proof obligations, namely $\bot \sqsubseteq$~\Verb|\IdrisBound{l}|, \Verb|\IdrisBound{l'}|, \Verb|\IdrisBound{l}|~$\sqcup$~\Verb|\IdrisBound{l'}| and \Verb|\IdrisBound{l}|, \Verb|\IdrisBound{l'}|~$\sqsubseteq$~\Verb|\IdrisBound{l}|~$\sqcup$~\Verb|\IdrisBound{l'}|.
When working on a concrete lattice these obligations are usually fulfilled by \Idris' automatic proof search but, currently, such proofs need to be given manually in the general case.
All obligations follow immediately from the algebraic properties of the bounded semilattice and are given in three auxiliary lemmas \Verb|\IdrisFunction{leq\_bot\_x}|, \Verb|\IdrisFunction{join\_x\_xy}|, and \Verb|\IdrisFunction{join\_y\_xy}| available in Appendix~\ref{app:core-lib} (amounting to 10 lines of code).

Writing functions operating on a fixed number of resources is limiting.
However, the function in Figure~\ref{fig:readTwoFiles} can easily be generalized to a function working on an arbitrary data structure containing files with different labels from an arbitrary lattice.
Similar to the approach taken by \citet{buiras15} that hide the label of a labeled value using a data type definition, we hide the label of a secure file with a dependent pair
\begin{Verbatim}
\IdrisFunction{GenFile} : \IdrisType{Type} -> \IdrisType{Type}
\IdrisFunction{GenFile} \IdrisBound{label} = \IdrisType{(\IdrisBound{l}} \IdrisType{:} \IdrisBound{label} \IdrisType{**} \IdrisType{SecFile} \IdrisBound{l\IdrisType{)}}
\end{Verbatim}
that abstracts away the concrete sensitivity level of the file. Moreover, we introduce a specialized join function
\begin{Verbatim}
\IdrisFunction{joinOfFiles} : \IdrisType{\IdrisType{BoundedJoinSemilattice}} \IdrisBound{\IdrisBound{label}}
           => \IdrisType{\IdrisType{List}} (\IdrisType{\IdrisType{GenFile}} \IdrisBound{\IdrisBound{label}})
           -> \IdrisBound{\IdrisBound{label}}
\end{Verbatim}
that folds the \join function over a list of file sensitivity labels. Now, it is possible to implement a function that takes as input a list of files, reads the files, and returns a computation that concatenates all their contents (if no file error occurred) where the return value is labeled with the join of all their sensitivity labels.
\begin{Verbatim}
\IdrisFunction{readFiles} : \IdrisType{BoundedJoinSemilattice} \IdrisBound{a}
         => (\IdrisBound{files}: (\IdrisType{List} (\IdrisType{GenFile} \IdrisBound{a})))
         -> \IdrisType{DIO} \IdrisFunction{Bottom} (\IdrisType{Labeled} (\IdrisFunction{joinOfFiles} \IdrisBound{files})
                         (\IdrisType{Either} (\IdrisType{List} \IdrisType{FileError}) \IdrisType{String}))
\end{Verbatim}
When implementing this, one has to satisfy non-trivial proof obligations as, for example, that $l \sqsubseteq$~\Verb|\IdrisFunction{joinOfFiles}(\IdrisBound{files})| for all secure files $f \in$~\Verb|\IdrisBound{files}| where the label of $f$ is~$l$.
While provable (in 40 lines of code in our development), if equality is decidable for elements of the concrete lattice we can postpone such proof obligations to a point in time where it can be solved by reflexivity of equality.
By defining a decidable lattice order
\begin{Verbatim}
\IdrisFunction{decLeq} : \IdrisType{\IdrisType{JoinSemilattice}} \IdrisBound{\IdrisBound{a}} => \IdrisType{\IdrisType{DecEq}} \IdrisBound{\IdrisBound{a}} => (\IdrisBound{x}, \IdrisBound{y} : \IdrisBound{\IdrisBound{\IdrisBound{\IdrisBound{a}}}}) -> \IdrisType{\IdrisType{Dec}} (\IdrisBound{\IdrisBound{x}} \IdrisFunction{\IdrisFunction{`leq`}} \IdrisBound{\IdrisBound{y}})
\IdrisFunction{\IdrisFunction{decLeq}} \IdrisBound{x} \IdrisBound{y} = \IdrisFunction{\IdrisFunction{decEq}} (\IdrisBound{\IdrisBound{x}} \IdrisFunction{\IdrisFunction{`join`}} \IdrisBound{\IdrisBound{y}}) \IdrisBound{\IdrisBound{y}}
\end{Verbatim}
we can get such a proof ``for free'' by inserting a dynamic check of whether the flow is allowed.
With this, a \Verb|\IdrisFunction{readFiles'}| function with the exact same functionality as the original \Verb|\IdrisFunction{readFiles}| function can be implemented with minimum effort.
In the below, \Verb|\IdrisBound{prf}| is the proof that the label \Verb|\IdrisFunction{l}| of \Verb|\IdrisBound{file}| may flow to \Verb|\IdrisFunction{joinOfFiles} \IdrisBound{files} |.
\begin{Verbatim}
\IdrisFunction{readFiles'} : \IdrisType{BoundedJoinSemilattice} \IdrisBound{a} => \IdrisType{DecEq} \IdrisBound{a}
          => (\IdrisBound{files}: (\IdrisType{List} (\IdrisType{GenFile} \IdrisBound{a})))
          -> \IdrisType{DIO} \IdrisFunction{Bottom} (\IdrisType{Labeled} (\IdrisFunction{joinOfFiles} \IdrisBound{files})
                          (\IdrisType{Either} (\IdrisType{List} \IdrisType{FileError}) \IdrisType{String}))
\IdrisFunction{\IdrisFunction{readFiles'}} \IdrisBound{files} =
  ...
  \IdrisKeyword{case} \IdrisFunction{decLeq}{ \IdrisFunction{l} (\IdrisFunction{\IdrisFunction{joinOfFiles}} \IdrisBound{\IdrisBound{files}})} \IdrisKeyword{of}
    \IdrisData{Yes} \IdrisBound{prf} => ...
    \IdrisData{No} _ => ...
\end{Verbatim}
The downside of this is the introduction of a negative case, the \Verb|\IdrisData{No}|-case, that needs handling even though it will never occur if \Verb|\IdrisFunction{joinOfFiles}| is implemented correctly.

In combination with \Verb|\IdrisType{GenFile}|, \Verb|\IdrisFunction{decLeq}| can be used to implement several other interesting examples.
For instance, a function that reads all files with a sensitivity label below a certain label to a string labeled with that label.
The accompanying source code showcases multiple such examples that exploit decidable equality.

\section{Declassification}\label{sec:declassification}
Realistic applications often release some secret information as part of their intended behavior; this action is known as \emph{declassification}.

In \DepSec, trusted code may declassify secret information without adhering to any security policy as trusted code has access to both the \DIO{$\ell$}{$a$} and \Labeled{$\ell$}{$a$} data constructors.
However, only giving trusted code the power of declassification is limiting as we want to allow the use of third-party code as much as possible.
The main challenge we address is how to grant untrusted code the right amount of power such that declassification is only possible in the intended way.

\citet*{sabelfeld05} identify four dimensions of declassification: \emph{what}, \emph{who}, \emph{where}, and \emph{when}.
In this section, we present novel and powerful means for static declassification with respect to three of the four dimensions and illustrate these with several examples.
To statically enforce different declassification policies we take the approach of \citet{sabelfeld03b} and use escape hatches, a special kind of functions.
In particular, we introduce the notion of a \emph{hatch builder}; a function that creates an escape hatch for a particular resource and which can only be used when a certain condition is met.
Such an escape hatch can therefore be used freely by untrusted code.

\subsection{The \emph{what} dimension}
Declassification policies related to the \emph{what} dimension place restrictions on exactly ``what'' and ``how much'' information is released.
It is in general difficult to statically predict how data to be declassified is manipulated or changed by programs \cite{russo09} but exploiting dependent types can get us one step closer.

To control what information is released, we introduce the notion of a \emph{predicate hatch builder} only available to trusted code for producing hatches for untrusted code.
\begin{Verbatim}
\IdrisFunction{predicateHatchBuilder} : \IdrisType{\IdrisType{Poset}} \IdrisBound{\IdrisBound{lt}} => \{\IdrisBound{l}, \IdrisBound{l'} : \IdrisBound{\IdrisBound{\IdrisBound{\IdrisBound{lt}}}}\} -> \{\IdrisBound{D}, \IdrisBound{E} : \IdrisType{\IdrisType{Type}}\}
                     -> (\IdrisBound{d} : \IdrisBound{\IdrisBound{D}})
                     -> (\IdrisBound{P} : \IdrisBound{\IdrisBound{D}} -> \IdrisBound{\IdrisBound{E}} -> \IdrisType{Type})
                     -> \IdrisType{(\IdrisBound{d}} \IdrisType{:} \IdrisBound{\IdrisBound{D}} \IdrisType{**} \IdrisType{\IdrisType{Labeled}} \IdrisBound{\IdrisBound{l}} \IdrisType{(\IdrisBound{e}} \IdrisType{:} \IdrisBound{\IdrisBound{E}} \IdrisType{**} \IdrisBound{\IdrisBound{P}} \IdrisBound{\IdrisBound{d}} \IdrisBound{\IdrisBound{e\IdrisType{)}}}
                                  -> \IdrisType{\IdrisType{Labeled}} \IdrisBound{\IdrisBound{l'}} \IdrisBound{\IdrisBound{E\IdrisType{)}}} -- TCB
\end{Verbatim}
Intuitively, the hatch builder takes as input a data structure \Verb|\IdrisBound{d}| of type \Verb|\IdrisBound{D}| followed by a predicate \Verb|\IdrisBound{P}| upon \Verb|\IdrisBound{d}| and something of type \Verb|\IdrisBound{E}|.
It returns a dependent pair of the initial data structure and a declassification function from sensitivity level \Verb|\IdrisBound{l}| to \Verb|\IdrisBound{l'}|.
To actually declassify a labeled value \Verb|\IdrisBound{e}| of type \Verb|\IdrisBound{E}| one has to provide a proof that \Verb|\IdrisBound{P d e}| holds.
Notice that this proof may be constructed in the context of the sensitivity level \Verb|\IdrisBound{l}| that we are declassifying from.

The reason for parameterizing the predicate \Verb|\IdrisBound{P}| by a data structure of type \Verb|\IdrisBound{D}| is to allow declassification to be restricted to a specific context or data structure.
This is used in the following example of an auction system, in which only the highest bid of a specific list of bids can be declassified.

\paragraph{Example}
Consider a two point lattice where \Verb|\IdrisData{L}|~$\sqsubseteq$~\Verb|\IdrisData{H}|, \Verb|\IdrisData{H}|~$\not\sqsubseteq$~\Verb|\IdrisData{L}| and an auction system where participants place bids secretly.
All bids are labeled \Verb|\IdrisData{H}| and are put into a data structure \Verb|\IdrisType{BidLog}|.
In the end, we want only the winning bid to be released and hence declassified to label \Verb|\IdrisData{L}|.
To achieve this, we define a declassification predicate \Verb|\IdrisFunction{HighestBid}|.
\begin{Verbatim}
\IdrisFunction{HighestBid} : \IdrisFunction{\IdrisType{BidLog}} -> \IdrisFunction{\IdrisFunction{Bid}} -> \IdrisType{Type}
\IdrisFunction{\IdrisFunction{HighestBid}} = \textbackslash{}\IdrisBound{log}, \IdrisBound{b} => \IdrisType{\IdrisType{(\IdrisType{\IdrisType{Elem}}}} (\IdrisFunction{\IdrisFunction{label}} \IdrisBound{\IdrisBound{b}}) \IdrisBound{\IdrisBound{log\IdrisType{,}}} \IdrisType{\IdrisType{MaxBid}} \IdrisBound{\IdrisBound{b}} \IdrisBound{\IdrisBound{log\IdrisType{)}}}
\end{Verbatim}
Informally, given a log \Verb|\IdrisBound{log}| of labeled bids and a bid \Verb|\IdrisBound{b}|, the predicate states that the bid is in the log, \Verb|\IdrisType{Elem} (\IdrisFunction{label} \IdrisBound{b}) \IdrisBound{log}|, and that it is the maximum bid, \Verb|\IdrisType{MaxBid} \IdrisBound{b} \IdrisBound{log}|.
We apply \Verb|\IdrisFunction{predicateHatchBuilder}| to a log of bids and the \Verb|\IdrisFunction{HighestBid}| predicate to obtain a specialized escape hatch of type \Verb|\IdrisFunction{BidHatch}| that enforces the declassification policy defined by the predicate.

\begin{Verbatim}
\IdrisFunction{BidHatch} : \IdrisType{Type}
\IdrisFunction{\IdrisFunction{BidHatch}} = \IdrisType{(\IdrisBound{log}} \IdrisType{:} \IdrisFunction{\IdrisFunction{BidLog}} \IdrisType{**} \IdrisType{\IdrisType{Labeled}} \IdrisData{\IdrisData{H}} \IdrisType{(\IdrisBound{b}} \IdrisType{:} \IdrisFunction{\IdrisFunction{Bid}} \IdrisType{**} \IdrisFunction{\IdrisFunction{HighestBid}} \IdrisBound{\IdrisBound{log}} \IdrisBound{\IdrisBound{b\IdrisType{)}}}
                            -> \IdrisType{\IdrisType{Labeled}} \IdrisData{\IdrisData{L}} \IdrisFunction{\IdrisFunction{Bid\IdrisType{)}}}
\end{Verbatim}
This hatch can be used freely by untrusted code when implementing the auction system.
By constructing a function
\begin{Verbatim}
\IdrisFunction{getMaxBid} : (\IdrisBound{r} : \IdrisFunction{\IdrisFunction{BidLog}}) -> \IdrisType{\IdrisType{DIO}} \IdrisData{\IdrisData{H}} \IdrisType{(\IdrisBound{b}} \IdrisType{:} \IdrisFunction{\IdrisFunction{Bid}} \IdrisType{**} \IdrisFunction{\IdrisFunction{HighestBid}} \IdrisBound{\IdrisBound{r}} \IdrisBound{\IdrisBound{b\IdrisType{)}}}
\end{Verbatim}
untrusted code can plug the resulting computation into an \Verb|\IdrisData{L}| context and declassify the result value using the argument \Verb|\IdrisBound{hatch}| function.
\begin{Verbatim}
\IdrisFunction{auction} : \IdrisFunction{\IdrisFunction{BidHatch}} -> \IdrisType{\IdrisType{DIO}} \IdrisData{\IdrisData{L}} (\IdrisType{\IdrisType{Labeled}} \IdrisData{\IdrisData{L}} \IdrisFunction{\IdrisFunction{Bid}})
\IdrisFunction{auction} \IdrisData{([]} \IdrisData{**} _\IdrisData{)} = \IdrisFunction{pure} $ \IdrisFunction{label} \IdrisData{(\IdrisBound{"no bids"\IdrisData{,}}} \IdrisData{0\IdrisData{)}}
\IdrisFunction{\IdrisFunction{auction}} \IdrisData{\IdrisData{(\IdrisBound{r}}} \IdrisData{\IdrisData{::}} \IdrisBound{rs} \IdrisData{**} \IdrisBound{hatch\IdrisData{)}} =
  \IdrisKeyword{do} \IdrisBound{max} <- \IdrisFunction{\IdrisFunction{plug}} (\IdrisFunction{getMaxBid} (\IdrisBound{r :: rs}))
    \IdrisKeyword{let} \IdrisBound{max'} : \IdrisType{\IdrisType{Labeled}} \IdrisData{\IdrisData{L}} \IdrisFunction{\IdrisFunction{Bid}} = \IdrisBound{\IdrisBound{\IdrisBound{\IdrisBound{hatch}}}} \IdrisBound{\IdrisBound{max}}
    ...
\end{Verbatim}
To show the \Verb|\IdrisFunction{HighestBid}| predicate (which in our implementation comprises 40 lines of code), untrusted code will need a generalized \Verb|\IdrisFunction{unlabel}| function that establishes the relationship between \Verb|\IdrisFunction{label}| and the output of \Verb|\IdrisFunction{unlabel}|.
The only difference is its return type: a computation that returns a value and a proof that when labeling this value we will get back the initial input.
This definition poses no risk to soundness as the proof is protected by the computation sensitivity level.
\begin{Verbatim}
\IdrisFunction{unlabel'} : \IdrisType{\IdrisType{Poset}} \IdrisBound{\IdrisBound{lt}} => \{\IdrisBound{l},\IdrisBound{l'}: \IdrisBound{\IdrisBound{\IdrisBound{\IdrisBound{lt}}}}\}
        -> \{\IdrisKeyword{auto} \IdrisBound{flow}: \IdrisBound{\IdrisBound{l}} \IdrisFunction{\IdrisFunction{`leq`}} \IdrisBound{\IdrisBound{l'}}\}
        -> (\IdrisBound{labeled}: \IdrisType{\IdrisType{Labeled}} \IdrisBound{\IdrisBound{l}} \IdrisBound{\IdrisBound{a}})
        -> \IdrisType{\IdrisType{DIO}} \IdrisBound{\IdrisBound{l'}} \IdrisType{(\IdrisBound{c}} \IdrisType{:} \IdrisBound{\IdrisBound{a}} \IdrisType{**} \IdrisFunction{\IdrisFunction{label}} \IdrisBound{\IdrisBound{c}} \IdrisType{=} \IdrisBound{\IdrisBound{labeled\IdrisType{)}}}
\end{Verbatim}

\paragraph{Limiting hatch usage}
Notice how escape hatches, generally, can be used an indefinite number of times.
The \texttt{Control.ST} library \cite{brady16} provides facilities for creating, reading, writing, and destroying state in the type of \Idris functions and, especially, allows tracking of state change in a function type.
This allows us to limit the number of times a hatch can be used.
Based on a concept of resources, a dependent type \Verb|\IdrisType{STrans}| tracks how resources change when a function is invoked.
Specifically, a value of type \Verb|\IdrisType{STrans} \IdrisBound{m returnType in\_{res} out\_{res}}| represents a sequence of actions that manipulate state where \Verb|\IdrisBound{m}| is an underlying computation context in which the actions will be executed, \Verb|\IdrisBound{returnType}| is the return type of the sequence, \Verb|\IdrisBound{in\_{res}}| is the required list of resources available before executing the sequence, and \Verb|\IdrisBound{out\_{res}}| is the list of resources available after executing the sequence.

To represent state transitions more directly, \ST is a type level function that computes an appropriate \Verb|\IdrisType{STrans}| type given a underlying computation context, a result type, and a list of \emph{actions}, which describe transitions on resources.
Actions can take multiple forms but the one we will make use of is of the form \Verb|\IdrisBound{lbl} \IdrisBound{:::} \IdrisFunction{ty\_in} \IdrisKeyword{:->} \IdrisFunction{ty\_out}| that expresses that the resource \Verb|\IdrisBound{lbl}| begins in state \Verb|\IdrisFunction{ty\_in}| and ends in state \Verb|\IdrisFunction{ty\_out}|.
By instantiating \ST with \DIO{\Verb|\IdrisBound{l}|}{\hspace{-0.2em}} as the underlying computation context:
\begin{Verbatim}
\IdrisFunction{DIO'} : \IdrisBound{\IdrisBound{l}} -> (\IdrisBound{ty} : \IdrisType{Type}) -> \IdrisType{\IdrisType{List}} (\IdrisType{\IdrisType{Action}} \IdrisBound{\IdrisBound{ty}}) -> \IdrisType{Type}
\IdrisFunction{\IdrisFunction{DIO'}} \IdrisBound{l} = \IdrisFunction{\IdrisFunction{ST}} (\IdrisType{\IdrisType{DIO}} \IdrisBound{\IdrisBound{l}})
\end{Verbatim}
and use it together with a resource \Verb|\IdrisFunction{Attempts}|, we can create a function \Verb|\IdrisFunction{limit}| that applies its first argument \Verb|\IdrisBound{f}| to its second argument \Verb|\IdrisBound{arg}| with \Verb|\IdrisFunction{\IdrisFunction{Attempts}} (\IdrisData{\IdrisData{S}} \IdrisBound{\IdrisBound{n}})| as its initial required state and \Verb|\IdrisFunction{\IdrisFunction{Attempts}} \IdrisBound{n}| as the output state.
\begin{Verbatim}
\IdrisFunction{limit} : (\IdrisBound{f} : \IdrisBound{\IdrisBound{a}} -> \IdrisBound{\IdrisBound{b}}) -> (\IdrisBound{arg} : \IdrisBound{\IdrisBound{a}})
     -> \IdrisFunction{\IdrisFunction{DIO'}} \IdrisBound{\IdrisBound{l}} \IdrisBound{\IdrisBound{b}} \IdrisData{\IdrisData{[\IdrisBound{\IdrisBound{attempts}}}} \IdrisFunction{\IdrisFunction{:::}} \IdrisFunction{\IdrisFunction{Attempts}} (\IdrisData{\IdrisData{S}} \IdrisBound{\IdrisBound{n}}) \IdrisData{\IdrisData{:->}} \IdrisFunction{\IdrisFunction{Attempts}} \IdrisBound{\IdrisBound{n\IdrisData{\IdrisData{]}}}}
\end{Verbatim}
That is, we encode that the function consumes ``an attempt.'' With the \Verb|\IdrisFunction{limit}| function it is possible to create functions where users are forced, by typing, to specify how many times it is used.

As an example, consider a variant of an example by \citet{russo09} where we construct a specialized hatch \Verb|\IdrisFunction{passwordHatch}| that declassifies the boolean comparison of a secret number with an arbitrary number.
\begin{Verbatim}
\IdrisFunction{passwordHatch} : (\IdrisBound{labeled} : \IdrisType{\IdrisType{Labeled}} \IdrisData{\IdrisData{H}} \IdrisType{Int})
             -> (\IdrisBound{guess} : \IdrisType{Int})
             -> \IdrisFunction{\IdrisFunction{DIO'}} \IdrisBound{\IdrisBound{l}} \IdrisType{\IdrisType{Bool}} \IdrisData{\IdrisData{[\IdrisBound{\IdrisBound{attempts}}}} \IdrisFunction{\IdrisFunction{:::}} \IdrisFunction{\IdrisFunction{Attempts}} (\IdrisData{\IdrisData{S}} \IdrisBound{\IdrisBound{n}}) \IdrisData{\IdrisData{:->}} \IdrisFunction{\IdrisFunction{Attempts}} \IdrisBound{\IdrisBound{n\IdrisData{\IdrisData{]}}}}
\IdrisFunction{\IdrisFunction{passwordHatch}} (\IdrisData{\IdrisData{MkLabeled}} \IdrisBound{v}) = \IdrisFunction{\IdrisFunction{limit}} (\textbackslash{}\IdrisBound{g} => \IdrisBound{\IdrisBound{g}} \IdrisFunction{\IdrisFunction{==}} \IdrisBound{\IdrisBound{v}})
\end{Verbatim}
To use this hatch, untrusted code is forced to specify how many times it is used.
\begin{Verbatim}
\IdrisFunction{pwCheck} : \IdrisType{\IdrisType{Labeled}} \IdrisData{\IdrisData{H}} \IdrisType{Int}
       -> \IdrisFunction{\IdrisFunction{DIO'}} \IdrisData{\IdrisData{L}} \IdrisType{()} \IdrisData{\IdrisData{[\IdrisBound{\IdrisBound{attempts}}}} \IdrisFunction{\IdrisFunction{:::}} \IdrisFunction{\IdrisFunction{Attempts}} (\IdrisData{3} \IdrisFunction{+} \IdrisBound{\IdrisBound{n}}) \IdrisData{\IdrisData{:->}} \IdrisFunction{\IdrisFunction{Attempts}} \IdrisBound{\IdrisBound{n\IdrisData{\IdrisData{]}}}}
\IdrisFunction{\IdrisFunction{pwCheck}} \IdrisBound{pw} =
  \IdrisKeyword{do} \IdrisBound{x1} <- \IdrisFunction{\IdrisFunction{passwordHatch}} \IdrisBound{\IdrisBound{pw}} \IdrisFunction{\IdrisFunction{\IdrisData{1}}}
     \IdrisBound{x2} <- \IdrisFunction{\IdrisFunction{passwordHatch}} \IdrisBound{\IdrisBound{pw}} \IdrisFunction{\IdrisFunction{\IdrisData{2}}}
     \IdrisBound{x3} <- \IdrisFunction{\IdrisFunction{passwordHatch}} \IdrisBound{\IdrisBound{pw}} \IdrisFunction{\IdrisFunction{\IdrisData{3}}}
     \IdrisBound{x4} <- \IdrisFunction{\IdrisFunction{passwordHatch}} \IdrisBound{\IdrisBound{pw}} \IdrisFunction{\IdrisFunction{\IdrisData{4}}} -- type error!
     ...
\end{Verbatim}

\subsection{The \emph{who} and \emph{when} dimensions}\label{sec:who}
To handle declassification policies related to \emph{who} may declassify information and \emph{when} declassification may happen we introduce the notion of a \emph{token hatch builder} only available to trusted code for producing hatches for untrusted code to use.
\begin{Verbatim}
\IdrisFunction{tokenHatchBuilder} : \IdrisType{\IdrisType{Poset}} \IdrisBound{\IdrisBound{labelType}} => \{\IdrisBound{l}, \IdrisBound{l'} : \IdrisBound{\IdrisBound{\IdrisBound{\IdrisBound{labelType}}}}\} -> \{\IdrisBound{E}, \IdrisBound{S} : \IdrisType{\IdrisType{Type}}\}
                 -> (\IdrisBound{Q} : \IdrisBound{\IdrisBound{S}} -> \IdrisType{Type})
                 -> \IdrisType{(\IdrisBound{s}} \IdrisType{:} \IdrisBound{\IdrisBound{S}} \IdrisType{**} \IdrisBound{\IdrisBound{\IdrisBound{\IdrisBound{Q}}}} \IdrisBound{\IdrisBound{s\IdrisType{)}}} -> \IdrisType{\IdrisType{Labeled}} \IdrisBound{\IdrisBound{l}} \IdrisBound{\IdrisBound{E}} -> \IdrisType{\IdrisType{Labeled}} \IdrisBound{\IdrisBound{l'}} \IdrisBound{\IdrisBound{E}} -- TCB
\end{Verbatim}
The hatch builder takes as input a predicate \Verb|\IdrisBound{Q}| on something of type \Verb|\IdrisBound{S}| and returns a declassification function from sensitivity level \Verb|\IdrisBound{l}| to \Verb|\IdrisBound{l\textquotesingle}| given that the user can prove the existence of some \Verb|\IdrisBound{s}| such that \Verb|\IdrisBound{Q s}| holds.
As such, by limiting when and how untrusted can obtain a value that satisfy predicate \Verb|\IdrisBound{Q}|, we can construct several interesting declassification policies.

The rest of this section discusses how predicate hatches can be used for time-based and authority-based control of declassification; the use of the latter is demonstrated on a case study.

\paragraph{Time-based hatches}
To illustrate the idea of token hatches for the \emph{when} dimension of declassification, consider the following example.
Let \Verb|\IdrisType{Time}| be an abstract data type with a data constructor only available to trusted code and \Verb|\IdrisFunction{tick} : \IdrisType{\IdrisType{DIO}} \IdrisBound{\IdrisBound{l}} \IdrisType{\IdrisType{Time}}| a function that returns the current system time wrapped in the \Verb|\IdrisType{Time}| data type such that this is the only way for untrusted code to construct anything of type \Verb|\IdrisType{Time}|.
Notice that this does not expose an unrestricted timer API as untrusted code can not inspect the actual value.

Now, we instantiate the token hatch builder with a predicate that demands the existence of a \Verb|\IdrisType{Time}| token that is greater than some specific value.
\begin{Verbatim}
\IdrisFunction{TimeHatch} : \IdrisType{\IdrisType{Time}} -> \IdrisType{Type}
\IdrisFunction{\IdrisFunction{TimeHatch}} \IdrisBound{t} = \IdrisType{(\IdrisBound{t'}} \IdrisType{**} \IdrisBound{\IdrisBound{t}} \IdrisFunction{\IdrisFunction{<=}} \IdrisBound{\IdrisBound{t'}} \IdrisType{=} \IdrisData{\IdrisData{True\IdrisType{)}}} -> \IdrisType{\IdrisType{Labeled}} \IdrisData{\IdrisData{H}} \IdrisType{\IdrisType{Nat}} -> \IdrisType{\IdrisType{Labeled}} \IdrisData{\IdrisData{L}} \IdrisType{\IdrisType{Nat}}
\end{Verbatim}
As such, \Verb|\IdrisFunction{TimeHatch} \IdrisBound{t}| can only be used after a specific point in time \Verb|\IdrisBound{t}| has passed as only then untrusted code will  be able to satisfy the predicate.
\begin{Verbatim}
\IdrisFunction{timer} : \IdrisType{\IdrisType{Labeled}} \IdrisData{\IdrisData{H}} \IdrisType{\IdrisType{Nat}} -> \IdrisFunction{\IdrisFunction{TimeHatch}} \IdrisBound{\IdrisBound{t}} -> \IdrisType{\IdrisType{DIO}} \IdrisData{\IdrisData{L}} \IdrisType{()}
\IdrisFunction{\IdrisFunction{timer}} \IdrisBound{secret} \{\IdrisBound{t}\} \IdrisBound{timeHatch} =
  \IdrisKeyword{do} \IdrisBound{time} <- \IdrisFunction{\IdrisFunction{tick}}
     \IdrisKeyword{\IdrisType{\IdrisType{case \IdrisFunction{\IdrisFunction{\IdrisType{\IdrisBound{\IdrisBound{\IdrisBound{\IdrisBound{\IdrisBound{\IdrisBound{\IdrisBound{decEq (\IdrisBound{\IdrisBound{t}} \IdrisFunction{\IdrisFunction{<=}} \IdrisBound{\IdrisBound{time}}) \IdrisData{\IdrisData{True}}}}}}}}}}}} \IdrisKeyword{of}
}}}
       \IdrisData{\IdrisData{Yes}} \IdrisBound{prf} =>
         \IdrisKeyword{let} \IdrisBound{declassified} : \IdrisType{\IdrisType{Labeled}} \IdrisData{\IdrisData{L}} \IdrisType{\IdrisType{Nat}} = \IdrisBound{\IdrisBound{timeHatch}} \IdrisData{\IdrisData{(\IdrisBound{time}}} \IdrisData{**} \IdrisBound{\IdrisBound{prf\IdrisData{)}}} \IdrisBound{\IdrisBound{secret}}
         ...
       \IdrisData{\IdrisData{No}} _ => ...
\end{Verbatim}

\paragraph{Authority-based hatches}
The \emph{Decentralized Labeling Model} (DLM)~\cite{myers00} marks data with a set of principals who owns the information.
While executing a program, the program is given \emph{authority}, that is, it is authorized to act on behalf of some set of principals.
Declassification simply makes a copy of the released data and marks it with the same set of principals but excludes the authorities.

Similarly to Russo et al.~\cite{russo09}, we adapt this idea such that it works on a security lattice of~\Verb|\IdrisType{Principal}|s, assign authorities with security levels from the lattice, and let authorities declassify information at that security level.

To model this, we define the abstract data type \Verb|\IdrisType{Authority}| with a data constructor available only to trusted code so that having an instance of \Verb|\IdrisType{Authority} \IdrisBound{s}| corresponds to having the authority of the principal \Verb|\IdrisBound{s}|.
Notice how assignment of authorities to pieces of code consequently is a part of the trusted code.
Now, we instantiate the token hatch builder with a predicate that demands the authority of \Verb|\IdrisBound{s}| to declassify information at that level.
\begin{Verbatim}
\IdrisFunction{authHatch} : \{ \IdrisBound{l}, \IdrisBound{l'} : \IdrisType{\IdrisType{\IdrisType{\IdrisType{Principal}}}} \}
         -> \IdrisType{(\IdrisBound{s}} \IdrisType{**} \IdrisType{\IdrisType{(\IdrisBound{\IdrisBound{l}}}} \IdrisType{=} \IdrisBound{\IdrisBound{s\IdrisType{,}}} \IdrisType{\IdrisType{Authority}} \IdrisBound{\IdrisBound{s\IdrisType{)\IdrisType{)}}}}
         -> \IdrisType{\IdrisType{Labeled}} \IdrisBound{\IdrisBound{l}} \IdrisBound{\IdrisBound{a}} -> \IdrisType{\IdrisType{Labeled}} \IdrisBound{\IdrisBound{l'}} \IdrisBound{\IdrisBound{a}}
\IdrisFunction{\IdrisFunction{authHatch}} \{\IdrisBound{l}\} = \IdrisFunction{\IdrisFunction{tokenHatchBuilder}} (\textbackslash{}\IdrisBound{s} => \IdrisType{\IdrisType{(\IdrisBound{\IdrisBound{l \IdrisType{=\IdrisBound{\IdrisBound{ s\IdrisType{,}}}}}}}} \IdrisType{\IdrisType{Authority}} \IdrisBound{\IdrisBound{s\IdrisType{)}}})
\end{Verbatim}
That is, \Verb|\IdrisFunction{authHatch}| makes it possible to declassify information at level \Verb|\IdrisBound{l}| to \Verb|\IdrisBound{l'}| given an instance of the \Verb|\IdrisType{Authority} \IdrisBound{l}| data type.

\paragraph{Example}
Consider the scenario of an online dating service that has the distinguishing feature of allowing its users to specify the visibility of their profiles at a fine-grained level.
To achieve this, the service allows users to provide a \emph{discovery agent} that controls their visibility.
Consider a user, Bob, whose implementation of the discovery agent takes as input his own profile and the profile of another user, say Alice.
The agent returns a possibly side-effectful computation that returns an option type indicating whether Bob wants to be discovered by Alice.
If that is the case, a profile is returned by the computation with the information about Bob that he wants Alice to be able to see.
When Alice searches for candidate matches, her profile is run against the discovery agents of all candidates and the result is added to her browsing queue.

To implement this dating service, we define the record type \Verb|\IdrisType{ProfileInfo} \IdrisBound{A}| that contains personal information related to principal \Verb|\IdrisBound{A}|.

\begin{Verbatim}
\IdrisKeyword{record} \IdrisType{ProfileInfo} (\IdrisBound{A} : \IdrisType{Principal}) \IdrisKeyword{where}
  \IdrisKeyword{constructor} \IdrisData{MkProfileInfo}
  \IdrisBound{name}      : \IdrisType{Labeled} \IdrisBound{A} \IdrisType{String}
  \IdrisBound{gender}    : \IdrisType{Labeled} \IdrisBound{A} \IdrisType{String}
  \IdrisBound{birthdate} : \IdrisType{Labeled} \IdrisBound{A} \IdrisType{String}
  ...
\end{Verbatim}
The interesting part of the dating service is the implementation of discovery agents.
Figure \ref{fig:sample-discoverer} presents a sample discovery agent that matches all profiles with the opposite gender and only releases information about the name and gender.
The discovery agent demands the authority of \Verb|\IdrisBound{A}| and takes as input two profiles \Verb|\IdrisBound{a} : \IdrisType{\IdrisType{ProfileInfo}} \IdrisBound{\IdrisBound{A}}| and \Verb|\IdrisBound{b} : \IdrisType{\IdrisType{ProfileInfo}} \IdrisBound{\IdrisBound{B}}|.
The resulting computation security level is \Verb|\IdrisBound{B}| so to incorporate information from \Verb|\IdrisBound{a}| into the result, declassification is needed.
This is achieved by providing \Verb|\IdrisFunction{authHatch}| with the authority proof of \Verb|\IdrisBound{A}|.
The discovery agent \Verb|\IdrisFunction{sampleDiscoverer}| in Figure~\ref{fig:sample-discoverer} unlabels \Verb|\IdrisBound{B}|'s gender, declassifies and unlabels \Verb|\IdrisBound{A}|'s gender and name, and compares the two genders.
If the genders match, a profile with type \Verb|\IdrisType{ProfileInfo} \IdrisBound{\IdrisBound{B}}| only containing the name and gender of \Verb|\IdrisBound{A}| is returned.
Otherwise, \Verb|\IdrisData{Nothing}| is returned indicating that \Verb|\IdrisBound{A}| does not want to be discovered.
Notice that \Verb|\IdrisFunction{Refl}| is the constructor for the built-in equality type in \Idris and it is used to construct the proof of equality between principals required by the hatch.
\begin{figure}[htb!]
  \centering
\begin{Verbatim}
\IdrisFunction{sampleDiscoverer} : \{\IdrisBound{A}, \IdrisBound{B} : \IdrisType{\IdrisType{\IdrisType{\IdrisType{Principal}}}}\}
                -> \IdrisType{\IdrisType{Authority}} \IdrisBound{\IdrisBound{A}}
                -> (\IdrisBound{a} : \IdrisType{\IdrisType{ProfileInfo}} \IdrisBound{\IdrisBound{A}})
                -> (\IdrisBound{b} : \IdrisType{\IdrisType{ProfileInfo}} \IdrisBound{\IdrisBound{B}})
                -> \IdrisType{\IdrisType{DIO}} \IdrisBound{\IdrisBound{B}} (\IdrisType{\IdrisType{Maybe}} (\IdrisFunction{\IdrisFunction{ProfileInfo}} \IdrisBound{\IdrisBound{B}}))
\IdrisFunction{\IdrisFunction{sampleDiscoverer}} \{\IdrisBound{A}\} \{\IdrisBound{B}\} \IdrisBound{auth} \IdrisBound{a} \IdrisBound{b} =
  \IdrisKeyword{do} \IdrisBound{bGender} <- \IdrisFunction{\IdrisFunction{unlabel}} $ \IdrisFunction{\IdrisFunction{gender}} \IdrisBound{\IdrisBound{b}}
     \IdrisBound{aGender} <- \IdrisFunction{\IdrisFunction{unlabel}} $ \IdrisFunction{\IdrisFunction{authHatch}} \IdrisData{\IdrisData{(\IdrisBound{A}}} \IdrisData{**} \IdrisData{\IdrisData{(\IdrisData{\IdrisData{Refl\IdrisData{,}}}}} \IdrisBound{\IdrisBound{auth\IdrisData{)\IdrisData{)}}}} (\IdrisFunction{\IdrisFunction{gender}} \IdrisBound{\IdrisBound{a}})
     \IdrisBound{aName} <- \IdrisFunction{\IdrisFunction{unlabel}} $ \IdrisFunction{\IdrisFunction{authHatch}} \IdrisData{\IdrisData{(\IdrisBound{A}}} \IdrisData{**} \IdrisData{\IdrisData{(\IdrisData{\IdrisData{Refl\IdrisData{,}}}}} \IdrisBound{\IdrisBound{auth\IdrisData{)\IdrisData{)}}}} (\IdrisFunction{\IdrisFunction{name}} \IdrisBound{\IdrisBound{a}})
     \IdrisKeyword{case} \IdrisFunction{\IdrisFunction{\IdrisType{\IdrisBound{\IdrisBound{decEq \IdrisBound{\IdrisBound{bGender}} \IdrisBound{\IdrisBound{aGender}}}}}}} \IdrisKeyword{of}
       \IdrisData{\IdrisData{Yes}} _ => \IdrisFunction{\IdrisFunction{pure}} \IdrisData{\IdrisData{Nothing}}
       \IdrisData{\IdrisData{No}} _  => \IdrisFunction{\IdrisFunction{pure}} (\IdrisData{\IdrisData{Just}} (\IdrisData{\IdrisData{MkProfileInfo}} \IdrisFunction{\IdrisFunction{\IdrisBound{\IdrisBound{aName}}}} \IdrisFunction{\IdrisFunction{\IdrisData{\IdrisBound{aGender}}}} \IdrisFunction{\IdrisFunction{\IdrisData{""}}} \IdrisFunction{\IdrisFunction{\IdrisData{""}}} \IdrisFunction{\IdrisFunction{\IdrisData{""}}}))
\end{Verbatim}
  \caption{A discovery agent that matches with all profiles of the opposite gender and only releases the name and gender.}
  \label{fig:sample-discoverer}
\end{figure}

\section{Soundness}\label{sec:soundness}
Recent works \cite{vassena16,vassena18} present a mechanically-verified model of \MAC and show pro\-gress-insensitive noninterference (PINI) for a sequential calculus.
We use this work as a starting point and discuss necessary modification in the following.
Notice that this work does not consider any declassification mechanisms and neither do we; we leave this as future work.

The proof relies on the \emph{two-steps erasure} technique, an extension of the \emph{term erasure} \cite{li10} technique that ensures that the same public output is produced if secrets are erased before or after program execution.
The technique relies on a type-driven erasure function $\varepsilon_{\ell_{A}}$ on terms and configurations where $\ell_{A}$ denotes the attacker security level.
A configuration consists of an $\ell$-indexed compartmentalized store $\Sigma$ and a term $t$.
A configuration $\langle \Sigma, t \rangle$ is erased by erasing $t$ and by erasing $\Sigma$ pointwise, i.e.
$\varepsilon_{\ell_{A}}(\Sigma) = \lambda\ell.
\varepsilon_{\ell_{A}}(\Sigma(\ell))$.
On terms, the function essentially rewrites data and computations above $\ell_{A}$ to a special $\bullet$ value.
The full definition of the erasure function is available in Appendix~\ref{sec:erasure}.
From this definition, the definition of low-equivalence of configurations follows.
\begin{definition}
Let $c_{1}$ and $c_{2}$ be configurations. $c_{1}$ and $c_{2}$ are said to be $\ell_{A}$-equivalent, written $c_{1} \approx_{\ell_{A}} c_{2} $, if and only if $\varepsilon_{\ell_{A}}(c_{1}) \equiv \varepsilon_{\ell_{A}}(c_{2})$.
\end{definition}
After defining the erasure function, the noninterference theorem follows from showing a \emph{single-step simulation} relationship between the erasure function and a small-step reduction relation: erasing sensitive data from a configuration and then taking a step is the same as first taking a step and then erasing sensitive data.
This is the content of the following proposition.
\begin{proposition}\label{prop:single-step}
 If $c_{1} \approx_{\ell_{A}} c_{2}$, $c_{1} \rightarrow c_{1}'$, and $c_{2} \rightarrow c_{2}'$ then $c_{1}' \approx_{\ell_{A}} c_{2}'$.
\end{proposition}
The main theorem follows by repeated applications of Proposition~\ref{prop:single-step}.
\begin{theorem}[PINI]
 If $c_{1} \approx_{\ell_{A}} c_{2}$, $c_{1} \Downarrow c_{1}'$, and $c_{2} \Downarrow c_{2}'$ then $c_{1}' \approx_{\ell_{A}} c_{2}'$.
\end{theorem}
Both the statement and the proof of noninterference for \DepSec are mostly similar to the ones for \MAC and available in Appendix~\ref{sec:results}.
Nevertheless, one has to be aware of a few subtleties.

First, one has to realize that even though dependent types in a language like \Idris may depend on data, the data itself is not a part of a value of a dependent type.
Recall the type \Verb|\IdrisType{Vect} \IdrisBound{n} \IdrisType{Nat}| of vectors of length \Verb|\IdrisBound{n}| with components of type \Verb|\IdrisType{Nat}| and consider the following program.
\begin{Verbatim}
\IdrisFunction{length} : \IdrisType{Vect} \IdrisBound{n} \IdrisBound{a} -> \IdrisType{Nat}
\IdrisFunction{length} \{n = \IdrisBound{n}\} \IdrisBound{xs} = \IdrisBound{n}
\end{Verbatim}
This example may lead one to believe that it is possible to extract data from a dependent type.
This is \emph{not} the case.
Both \Verb|\IdrisBound{n}| and \Verb|\IdrisBound{a}| are implicit arguments to the \Verb|\IdrisFunction{length}| function that the compiler is able to infer.
The actual type is
\begin{Verbatim}
\IdrisFunction{length} : \{\IdrisBound{n} : \IdrisType{Nat}\} -> \{\IdrisBound{a} : \IdrisType{Type}\} -> \IdrisType{Vect} \IdrisBound{n} \IdrisBound{a} -> \IdrisType{Nat}
\end{Verbatim}
As a high-level dependently typed functional programming language, \Idris is elaborated to a low-level core language based on dependent type theory \cite{brady13}.
In the elaboration process, such implicit arguments are made explicit when functions are defined and inferred when functions are invoked.
This means that in the underlying core language, only explicit arguments are given.
Our modeling given in Appendix~\ref{sec:syntax} reflects this fact soundly.

Second, to model the extended expressiveness of \DepSec, we extend both the semantics and the type system with compile-time pure-term reduction and higher-order dependent types.
These definitions are standard (defined for \Idris by \citet{brady13}) and available in Appendix~\ref{sec:semantics} and \ref{sec:typing}.
Moreover, as types now become first-class terms, the definition of \erasure has to be extended to cover the new kinds of terms.
As before, primitive types are unaffected by the erasure function, but dependent and indexed types, such as the type \Verb|\IdrisType{DIO}|, have to be erased homomorphically, e.g., $\varepsilon_{\ell_{A}}$ (\DIO{$\ell$}{$\tau$} : \Verb|\IdrisType{Type}|) $\triangleq$ \DIO{$\varepsilon_{\ell_{A}}(\ell)$}{$\varepsilon_{\ell_{A}}(\tau)}$.
The intuition of why this is sensible comes from the observation that indexed dependent types considered as terms may contain values that will have to be erased.
This is purely a technicality of the proof.
If defined otherwise, the erasure function would not commute with capture-avoiding substitution on terms, $\varepsilon_{\ell_{A}}(t[v/x]) = \varepsilon_{\ell_{A}}(t)[\varepsilon_{\ell_{A}}(v)/x$], which is vital for the remaining proof.

\section{Related work}\label{sec:related-work}
\paragraph{Security libraries}
The pioneering and formative work by \citet{li06} shows how \emph{arrows} \cite{hughes00}, a generalization of monads, can provide information-flow control without runtime checks as a library in Haskell.
\citet{tsai07} further extend this work to handle side-effects, concurrency, and heterogeneous labels.
\citet{russo09} eliminate the need for arrows and implement the security library \textbf{SecLib} in Haskell based solely on monads.
Rather than labeled values, this work introduces a monad which statically label side-effect free values.
Furthermore, it presents combinators to dynamically specify and enforce declassification policies that bear a resemblance to the policies that \DepSec are able to enforce statically.

The security library \textbf{LIO} \cite{stefan11,stefan12} dynamically enforces information-flow control in both sequential and concurrent settings.
\citet{stefan17} extend the security guarantees of this work to also cover exceptions.
Similar to this work, \citet{stefan11} present a simple API for implementing secure conference reviewing systems in \textbf{LIO} with support for data-dependent security policies.

Inspired by the design of \textbf{SecLib} and \textbf{LIO}, \citet{russo15} introduces the security library \MAC.
The library statically enforces information-flow control in the presence of advanced features like exceptions, concurrency, and mutable data structures by exploiting Haskell's type system to impose flow constraints.
\citet{vassena16,vassena18} show progress-insensitive noninterference for \MAC in a sequential setting and progress-sensitive noninterference in a concurrent setting, both using the two-steps erasure technique.

The flow constraints enforcing confidentiality of read and write operations in \DepSec are identical to those of \MAC.
This means that the examples from \MAC that do not involve concurrency can be ported directly to \DepSec.
To the best of our knowledge, data-dependent security policies like the one presented in Section \ref{sec:conference-system} cannot be expressed and enforced in \MAC, unlike \textbf{LIO} that allows such policies to be enforced dynamically.
\DepSec allows for such security policies to be enforced statically.
Moreover, \citet{russo15} does not consider declassification.
To address the static limitations of \MAC, \textbf{HLIO} \cite{buiras15} takes a hybrid approach by exploiting advanced features in Haskell's type-system like singleton types and constraint polymorphism.
\citet{buiras15} are able to statically enforce information-flow control while allowing selected security checks to be deferred until run-time.

\paragraph{Dependent types for security}
Several works have considered the use of dependent types to capture the nature of data-dependent security policies.
\citet{zheng04, zheng07} proposed the first dependent security type system for dealing with dynamic changes to runtime security labels in the context of Jif~\cite{jif}, a full-fledged IFC-aware compiler for Java programs, where similar to our work, operations on labels are modeled at the level of types.
\citet{zhang15} use dependent types in a similar fashion for the design of a hardware description language for timing-sensitive information-flow security.

A number of functional languages have been developed with dependent type systems and used to encode value-dependent information flow properties, e.g.
Fine \cite{swamy10}.
These approaches require the adoption of entirely new languages and compilers where \DepSec is embedded in an already existing language.
\citet{morgenstern10} encode an authorization and IFC-aware programming language in Agda.
However, their encoding does not consider side-effects.
\citet{nanevski11} use dependent types to verify information flow and access control policies in an interactive manner.

\citet{lourenco15} introduce the notion of \emph{dependent information-flow types} and propose a \emph{fine-grained} type system; every value and function have an associated security level.
Their approach is different to the \emph{coarse-grained} approach taken in our work where only some computations and values have associated security labels.
\citet{rajani18} show that both approaches are equally expressive for static IFC techniques and \citet{vassena19} show the same for dynamic IFC techniques.

\paragraph{Principles for Information Flow}
\citet{bastys2018prudent} put forward a set of informal principles for information flow security definitions and enforcement mechanisms: \emph{attacker-driven security, trust-aware enforcement, separation of policy annotations and code, language-independence, justified abstraction, and permissiveness}.

\DepSec follows the principle of trust-aware enforcement, as we make clear the boundary between the trusted and untrusted components in the program.
Additionally, the design of our declassification mechanism follows the principle of separation of policy annotations
and code.
The use of dependent types increases the permissiveness of our enforcement as we discuss throughout the paper.
While our approach is not fully language-independent, we posit that the approach may be ported to other programming languages
with general-purpose dependent types.

\paragraph{Declassification enforcement}
Our hatch builders are reminiscent of downgrading policies of~\citet{DBLP:conf/popl/LiZ05}.
For example, similar to them, \DepSec's declassification policies naturally express the idea of \emph{delimited release}~\cite{sabelfeld03}
that provides explicit characterization of the declassifying computation. Here,
\DepSec's policies can express a broad range of policies that can be expressed through predicates, an improvement over
simple expression-based enforcement mechanisms for delimited release~\cite{sabelfeld03,DBLP:conf/pldi/AskarovS07,DBLP:conf/csfw/AskarovS09}.

An interesting point in the design of declassification policies is \emph{robust declassification}~\cite{zdancewic01} that demands that untrusted components must not affect information release.
\emph{Qualified robustness} \cite{myers04,DBLP:journals/corr/abs-1107-5594} generalizes this notion by giving untrusted code a limited ability to affect information release through the introduction of an explicit endorsement operation.
Our approach is orthogonal to both notions of robustness as the intent is to let the untrusted components declassify information but only under very controlled circumstances while adhering to the security policy.

\section{Conclusion and future work}
In this paper, we have presented \DepSec~-- a library for statically enforced information-flow control in \Idris.
Through several case studies, we have showcased how the \DepSec primitives increase the expressiveness of state-of-the-art information-flow control libraries and how \DepSec matches the expressiveness of a special-purpose dependent information-flow type system on a key example.
Moreover, the library allows programmers to implement policy-parameterized functions that abstract over the security policy while retaining precise security levels.

By taking ideas from the literature and by exploiting dependent types, we have shown powerful means of specifying statically enforced declassification policies related to \emph{what}, \emph{who}, and \emph{when} information is released.
Specifically, we have introduced the notion of predicate hatch builders and token hatch builders that rely on the fulfillment of predicates and possession of tokens for declassification.
We have also shown how the \ST monad \cite{brady16} can be used to limit hatch usage statically.

Finally, we have discussed the necessary means to show progress-insensitive noninterference in a sequential setting for a dependently typed information-flow control library like \DepSec.

\paragraph{Future work}
There are several avenues for further work.
Integrity is vital in many security policies and is not considered in \MAC nor \DepSec.
It will be interesting to take integrity and the presence of concurrency into the dependently typed setting and consider internal and termination covert channels as well.
It also remains to prove our declassification mechanisms sound.
Here, attacker-centric epistemic security conditions~\cite{DBLP:conf/sp/AskarovS07,DBLP:journals/tissec/HalpernO08}
that intuitively express many declassification policies may be a good starting point.

\paragraph{Acknowledgements}
Thanks are due to Mathias Vorreiter Pedersen, Bas Spitters, Alejandro Russo, and Marco Vassena for their valuable insights and the anonymous reviewers for their comments on this paper.
This work is partially supported by DFF project~6108-00363 from The Danish Council for Independent Research for the Natural Sciences (FNU), Aarhus University Research Foundation, and the Concordium Blockchain Research Center, Aarhus University, Denmark.

\bibliographystyle{splncsnat}
\bibliography{bibliography}

\iftoggle{Full}{
  \newpage
\appendix

\section{The Calculus} \label{sec:calculus}
This section formalizes \DepSec as \TTsec, a dependently typed call-by-value $\lambda$-calculus extended with conditional expressions, references, unit, integer, and boolean values, as well as higher order dependent types and security primitives.

\subsection{Syntax}\label{sec:syntax}
Figure~\ref{fig:TT} shows the formal syntax of the pure calculus underlying \TTsec where meta variables $t$, $c$, $b$, and $\tau$ denote terms, constants, binders, and types, respectively.
The syntax closely resembles the syntax of \TT, the underlying calculus of \Idris, but with the addition of a conditional construct and base types \Int, \Bool, and \Unit.
\begin{figure}[htb!]
  \centering
  \begin{minipage}[t]{1\textwidth}
    \begin{tabular}{>{$}l<{$}>{$}r<{$}>{$}l<{$}>(l<)}
      \text{Terms, } t &\Coloneqq &c & constant \\
                       &| & x & variable \\
                       &| & b\texttt{.}t & binding \\
                       &| & t\ t & application \\
                       &| & \ifThenElse{t}{t}{t} & conditional \\
                       &| & \tau & type constructor \\
      \text{Constants, } c &\Coloneqq &i & integer literal \\
                       &| & bool & boolean literal \\
                       &| & \Unit & unit literal \\
      \text{Binders, } b &\Coloneqq & \lambda x : t &  abstraction \\
                         &|      & \forall x : t & function space \\
      \text{Types, } \tau &\Coloneqq &\Type &type of types\\
                         &| & \Int  & integer type \\
                         &| & \Bool & boolean type\\
                         &| & \Unit & unit type
    \end{tabular}
  \end{minipage}
  \caption{Syntax of the core calculus underlying \TTsec.}
  \label{fig:TT}
\end{figure}
We extend this standard calculus with the security primitives of \DepSec.
Figure~\ref{fig:TTsec} presents the extensions to Figure~\ref{fig:TT} that forms the formal syntax of \TTsec.
We introduce the security monad $\DIOValue{t}$ as well as type $\aDIO{\ell}{t}$ and monadic operators \pure{t}, \bind{t}{t}, and \aplug{t}.
We introduce a labeled value $\LabeledValue{t}$, a type $\LabeledType{\ell}{t}$, and labeling and unlabeling functions $\xlabel{t}$ and $\aunlabel{t}$.
As an example of a labeled resource we introduce references as values $\RefValue{n}$ as well as means for allocating, reading, and writing to references.
\begin{figure}[htb!]
  \centering
  \begin{minipage}[t]{0.65\textwidth}
    \begin{tabular}{>{$}l<{$}>{$}r<{$}>{$}l<{$}>(l<)}
      \text{Terms, } t, \ell &\Coloneqq &\ldots \\
                             &| & \pure{t} & return operator \\
                             &| & \DIOValue{t} & computation \\
                             &| & \LabeledValue{t} & labeled value\\
                             &| & \RefValue{\ell}{n} & reference \\
                             &| & \bind{t}{t} & bind\\
                             &| & \ylabel{t} & labeling\\
                             &| & \aunlabel{t} & unlabeling\\
                             &| & \aplug{t} & DIO plugging\\
                             &| & \newRef{\ell}{t}& new reference\\
                             &| & \readRef{t} & read reference \\
                             &| & \writeRef{t}{t} & write reference\\
    \end{tabular}
  \end{minipage}
  \begin{minipage}[t]{0.34\textwidth}
    \begin{tabular}{>{$}l<{$}>{$}r<{$}>{$}l<{$}>(l<)}
      \text{Types, } \tau &\Coloneqq &\ldots \\
                          &| & \aDIO{\ell}{t} \\
                          &| & \LabeledType{\ell}{t}  \\
                          &| & \RefType{\ell}{t} \\
    \end{tabular}
  \end{minipage}
  \caption{Syntax of \TTsec.}
  \label{fig:TTsec}
\end{figure}

\subsection{Operational semantics}\label{sec:semantics}
\begin{definition}[Small-step pure semantics]
  Let $\Term$ be the set of terms in \TTsec\ and let $t_{1}, t_{2} \in \Term$.
  The relation
  $$t_{1} \leadsto t_{2} \subseteq \Term \times \Term$$
  denotes the small-step operational semantics of the \TTsec\ calculus. The relation $t_{1} \leadsto t_{2}$ denotes that $t_{1}$ reduces to $t_{2}$ in one reduction step according to the inference rules in Figure~\ref{fig:TTsec-semantics}.
\end{definition}
We explicitly distinguish pure-term evaluation from top-level monadic-term evaluation.
The extended semantics is represented as the relation $c_{1} \evalsto c_{2}$ introduced in Definition \ref{def:monad-semantic} which extends the pure semantics $\leadsto$ via \textsc{Lift}.
\begin{definition}[Store]
  Let $\Label$ be a set of security labels.
  The function
  $$
  \Sigma : \Label  \rightarrow \List\ \Term
  $$
  denotes a \emph{store} compartmentalized into isolated labeled segments, one for each label.
  We write $\Sigma(\ell)[n]$ to retrieve the $n$th cell in the $\ell$-memory and $\Sigma(\ell)[n] \Coloneqq t$ for the store obtained by performing the update $\Sigma(\ell)[n \mapsto t]$.
\end{definition}
\begin{definition}[Monadic-term semantics]\label{def:monad-semantic}
  Let $\conf{\Sigma}{t}$ be a \emph{configuration} consisting of a store $\Sigma$ and a term $t \in \Term$.
  Let $\Conf$ be the set of all such configurations.
  The relation
  $$ c_{1} \evalsto c_{2}  \subseteq \Conf \times \Conf
  $$
  denotes the monadic-term evaluation according to the inference rules of Figure~\ref{fig:TTsec-semantics}.
  $\conf{\Sigma}{t} \evalstostar \conf{\Sigma'}{t'}$ denotes the reflexive transitive closure of \evalsto, and we write $\conf{\Sigma}{t} \bigstepto \conf{\Sigma'}{v}$ if and only if $v$ is a value and $\conf{\Sigma}{t} \evalstostar \conf{\Sigma'}{v}$.
\end{definition}
Note that we consider all non-reducible terms to be values and that constructors \LabeledValue{\hspace{-0.4em}}, \DIOValue{\hspace{-0.4em}}, and \RefValue{\hspace{-0.2em}} are not available to the user but only introduced in the semantics to model the run-time value produced by e.g.
\xlabel{\hspace{-0.4em}} and \pure{\hspace{-0.2em}}.

\begin{center}
  \begin{mathparpagebreakable}
    \textbf{Core calculus}
    \\
    \infer[App$_1$]
    { t_1 \leadsto t_1' }
    { t_1 \ t_2 \leadsto t_1' \ t_2 }
    \and
    \infer[App$_2$]
    { t \leadsto t' }
    { v \ t \leadsto v \ t' }
    \and
    \infer[Beta]
    { \phantom{ } }
    { (\lambda x . t) \ v \leadsto t[v/ x] }
    \and
    \infer[If$_1$]
    { t_1 \leadsto t_1' }
    { \ifThenElse{t_{1}}{t_2}{t_3} \leadsto \ifThenElse{t_{1}'}{t_2}{t_3} }
    \and
    \infer[If$_2$]
    { \phantom{ } }
    { \ifThenElse{\True}{t_2}{t_3} \leadsto t_2 }
    \and
    \infer[If$_3$]
    { \phantom{ } }
    { \ifThenElse{\False}{t_2}{t_3} \leadsto t_3 }
    \\
    \textbf{\DepSec (pure)}
    \\
    \infer[Bind$_{1}$]
    { \phantom{ } }
    { \bind{\DIOValue{t_1}}{t_2} \leadsto \app{t_2}{t_1} }
    \and
    \infer[Pure$_{1}$]
    { t \leadsto t'  }
    { \pure t  \leadsto \pure{t'} }
    \and
    \infer[Pure$_{2}$]
    { \phantom{ } }
    { \pure v  \leadsto \DIOValue{v} }
    \and
    \infer[Label$_{1}$]
    { t \leadsto t' }
    { \ylabel{t} \leadsto  \ylabel{t'} }
    \and
    \infer[Label$_{2}$]
    { \phantom{ } }
    { \ylabel{v} \leadsto {(\LabeledValue{v})} }
    \and
    \infer[Unlabel$_{1}$]
    { t \leadsto t' }
    { \aunlabel{t} \leadsto \aunlabel{t'} }
    \and
    \infer[Unlabel$_{2}$]
    { \phantom{ } }
    { \aunlabel{(\LabeledValue{v})} \leadsto \pure{v} }
    \and
    \infer[New$_{1}$]
    { t \leadsto t' }
    { \newRef{\ell}{t} \leadsto \newRef{\ell}{t'} }
    \and
    \infer[Write$_{1}$]
    { t_1 \leadsto t_1' }
    { \writeRef{t_1}{t_2} \leadsto \writeRef{t_1'}{t_2} }
    \and
    \infer[Write$_{2}$]
    { t_2 \leadsto t_2' }
    { \writeRef{v}{t_2} \leadsto \writeRef{v}{t_2'} }
    \and
    \infer[Read$_{1}$]
    { t \leadsto t' }
    { \readRef{t} \leadsto \readRef{t'} }
    \and
    \infer[DIO$_{1}$]
    { t \leadsto t' }
    { \aDIO{t}{\tau} \leadsto \aDIO{t'}{\tau}}
    \and
    \infer[DIO$_{2}$]
    { \tau \leadsto \tau' }
    { \aDIO{v}{\tau} \leadsto \aDIO{v}{\tau'}}
    \and
    \infer[Labeled$_{1}$]
    { t \leadsto t' }
    { \LabeledType{t}{\tau} \leadsto \LabeledType{t'}{\tau}}
    \and
    \infer[Labeled$_{2}$]
    { \tau \leadsto \tau' }
    { \LabeledType{v}{\tau} \leadsto \LabeledType{v}{\tau'}}
    \and
    \infer[Ref$_{1}$]
    { t \leadsto t' }
    { \RefType{t}{\tau} \leadsto \RefType{t'}{\tau}}
    \and
    \infer[Ref$_{2}$]
    { \tau \leadsto \tau' }
    { \RefType{v}{\tau} \leadsto \RefType{v}{\tau'}}
    \and
    \infer[Forall$_{1}$]
    { \tau \leadsto \tau' }
    { \forall x : \tau . t \leadsto \forall x : \tau' . t}
    \and
    \infer[Forall$_{2}$]
    { t \leadsto t' }
    { \forall x : \tau . t \leadsto \forall x : \tau . t' }
    \and
    \\
    \textbf{\DepSec (monadic)}
    \\
    \infer[Lift]
    { t \leadsto t' }
    { \conf{\Sigma}{t} \evalsto \conf{\Sigma}{t'} }
    \and
    \infer[Bind$_{2}$]
    { \conf{\Sigma}{t_1} \evalsto \conf{\Sigma'}{t_1'} }
    { \conf{\Sigma}{\bind{t_1}{t_2}} \evalsto \conf{\Sigma'}{\bind{t_1'}{t_2}} }
    \and
    \and
    \infer[Plug]
    { \conf{\Sigma}{t} \bigstepto \conf{\Sigma'}{\DIOValue{t'}} }
    { \conf{\Sigma}{\aplug{t}} \evalsto \conf{\Sigma'}{\pure{(\LabeledValue{t'})}} }
    \\
    \infer[New$_{2}$]
    { |\Sigma(\ell)| = n }
    { \conf{\Sigma}{\newRef{\ell}{(\LabeledValue{v})}} \evalsto \conf{\Sigma(\ell)[n] \Coloneqq v}{\pure{(\RefValue{\ell}{n})}} }
    \and
    \infer[Write$_{3}$]
    { \phantom{ } }
    { \conf{\Sigma}{\writeRef{(\RefValue{\ell}{n})}{(\LabeledValue{v})}} \evalsto \conf{\Sigma(\ell)[n] \Coloneqq v}{\pure{()}} }
    \and
    \infer[Read$_{2}$]
    { \phantom{ } }
    { \conf{\Sigma}{\readRef{(\RefValue{\ell}{n})}} \evalsto \conf{\Sigma}{\pure{\big(\LabeledValue{\Sigma(\ell)[n]}\big)}} }
  \end{mathparpagebreakable}
\captionof{figure}[short caption]{Operational semantics of \TTsec.}
\label{fig:TTsec-semantics}
\end{center}

\subsection{Typing rules}\label{sec:typing}
Similar to \TT, type checking and the dynamic semantics are defined mutually since evaluation relies on terms to be well-typed, and type checking relies on evaluation as equivalence of terms or types is determined by comparing their normal forms.
Compile-time evaluation of \TTsec is defined by the pure reductions rules in Figure~\ref{fig:TTsec-semantics} relative to a context $\Gamma$.
\emph{Conversion} ($\simeq$) is the smallest equivalence relation closed under reduction, that is, if $\Gamma \vdash x \simeq y$ then $x$ and $y$ reduce to the same normal form.

The type inference rules for \TTsec is presented in Figure~\ref{fig:typing}.
These rules use the \emph{cumulativity} ($\preceq$) relation defined in Figure~\ref{fig:cumulativity}.
In \TT, the type of types, \Type, is parameterized by a universe level (constructing an infinite hierarchy of universes) to prevent Girard's paradox.
As universe levels are transparent to the user, this is not relevant for our noninterference proof and we ignore this matter in the following.
As for \TT, we also conjecture that \TTsec respects usual properties such as type preservation and uniqueness of typing at compile-time.

\begin{center}
  \begin{mathparpagebreakable}
    \textbf{Core calculus}
    \\
    \infer[T-Type]
    { \phantom{ } }
    { \Gamma \vdash \Type : \Type }

    \infer[T-Const$_{1}$]
    { \phantom{ } }
    { \Gamma \vdash i : \Int }

    \infer[T-Const$_{2}$]
    { \phantom{ } }
    { \Gamma \vdash bool : \Bool }

    \infer[T-Const$_{3}$]
    { \phantom{ } }
    { \Gamma \vdash \Unit : \Unit }

    \infer[T-Const$_{4}$]
    { \phantom{ } }
    { \Gamma \vdash \Int : \Type }

    \infer[T-Const$_{5}$]
    { \phantom{ } }
    { \Gamma \vdash \Bool : \Type }

    \infer[T-Var]
    { (s : S) \in \Gamma }
    { \Gamma \vdash s : S }

    \infer[T-App]
    { \Gamma \vdash f : \forall x : S . T \\ \Gamma \vdash s : S }
    { \Gamma \vdash \app{f}{s} : T[s/x] }

    \infer[T-Lam]
    { \Gamma;x:S \vdash e : T \\ \Gamma \vdash \forall x : S . T : \Type}
    { \Gamma \vdash \lambda x : S . e : \forall x : S . T }

    \infer[T-Forall]
    { \Gamma; x : S \vdash T : \Type \\ \Gamma \vdash S : \Type}
    { \Gamma \vdash \forall x : S . T : \Type }

    \infer[T-IfThenElse]
    { \Gamma \vdash t_1 : \Bool \\ \Gamma; t_1 \equiv \True \vdash t_2 : S \\ \Gamma; t_1 \equiv \False  \vdash t_3 : S }
    { \Gamma \vdash \ifThenElse{t_{1}}{t_{2}}{t_{3}} : S}

    \infer[T-Conv]
    { \Gamma \vdash x : A \\ \Gamma \vdash A' : \Type \\ \Gamma \vdash A \preceq A' }
    { \Gamma \vdash x : A' }
    \\
    \textbf{\DepSec}
    \\
    \infer[T-DIO]
    { \Gamma \vdash \ell : \Label \\ \Gamma \vdash t : \Type }
    { \Gamma \vdash \aDIO{\ell}{t} : \Type }

    \infer[T-Labeled]
    { \Gamma \vdash \ell : \Label \\ \Gamma \vdash t : \Type }
    { \Gamma \vdash \LabeledType{\ell}{t} : \Type }

    \infer[T-Ref]
    { \Gamma \vdash \ell : \Label \\ \Gamma \vdash t : \Type }
    { \Gamma \vdash \RefType{\ell}{t} : \Type }

    \infer[T-Label]
    { \Gamma \vdash \ell :\Label \\ \Gamma \vdash s : S }
    { \Gamma \vdash \ylabel{s}: \LabeledType{\ell}{S}}

    \infer[T-Unlabel]
    { \ell_L \flowsto \ell_H \\ \Gamma \vdash s : \LabeledType{\ell_L}{S} }
    { \Gamma \vdash \aunlabel{s}: \aDIO{\ell_H}{S} }

    \infer[T-Bind]
    { \Gamma \vdash s : \aDIO{\ell}{S} \\ \Gamma \vdash t : S \rightarrow \aDIO{\ell}{T}}
    { \Gamma \vdash \bind{s}{t} : \aDIO{\ell}{T} }

    \infer[T-Pure]
    { \Gamma \vdash s : S \\ \Gamma \vdash \ell : \Label }
    { \Gamma \vdash \pure{s} : \aDIO{\ell}{S} }

    \infer[T-Plug]
    { \ell_L \flowsto \ell_H \\ \Gamma \vdash s : \aDIO{\ell_H}{S} }
    { \Gamma \vdash \aplug{s} : \aDIO{\ell_L}{(\LabeledType{\ell_H}{S})} }

    \infer[T-NewRef]
    { \ell_L \sqsubseteq \ell_M \sqsubseteq \ell_H \\ \Gamma \vdash s : \LabeledType{\ell_M}{S} }
    { \Gamma \vdash \newRef{\ell_H}{s} : \aDIO{\ell_L}{(\RefType{\ell_H}{S})}}

    \infer[T-WriteRef]
    { \ell_L \sqsubseteq \ell_M \sqsubseteq \ell_H \\ \Gamma \vdash s : \RefType{\ell_H}{S} \\ \Gamma \vdash t : \LabeledType{\ell_M}{S} }
    { \Gamma \vdash \writeRef{s}{t}: \aDIO{\ell_L}{\Unit}}

    \infer[T-ReadRef]
    { \ell_L \sqsubseteq \ell_H \\ \Gamma \vdash s : \RefType{\ell_H}{S}}
    { \Gamma \vdash \readRef{s}: \aDIO{\ell_L}{(\LabeledType{\ell_H}{S})}}
  \end{mathparpagebreakable}
  \captionof{figure}[short caption]{Typing rules for \TTsec.}
  \label{fig:typing}
\end{center}
\begin{figure}[H]
  \centering
  \begin{mathparpagebreakable}
    \infer[C-Conv]
    {\Gamma \vdash S \simeq T }
    {\Gamma \vdash S \preceq T }

    \infer[C-Forall]
    {\Gamma \vdash S_1 \simeq S_2 \\ \Gamma; x : S_1 \vdash T_1 \preceq T_2}
    {\Gamma \vdash \forall x:S_1.T_1 \preceq \forall x: S_2 . T_2 }
  \end{mathparpagebreakable}
  \caption{Cumulativity.}
  \label{fig:cumulativity}
\end{figure}

\subsection{Example: Concatenating strings}\label{sec:ttsec-concat}
This example illustrates the adequacy of the \TTsec calculus.
The concrete example immitates the \Verb|\IdrisFunction{readTwoFiles}| function presented in Section \ref{sec:generic-functions}.
It takes two labeled strings as input and returns the concatenated result of the content of these, labeled with the join of the original labels.
We assume having a well-typed string concatenation function $\concat$ and a well-defined \texttt{join} function for which the following rules hold:
  \begin{align*}
    &\infer{\ell : \Label \\ \ell' : \Label}{\ell \flowsto \mathtt{join}\ \ell\ \ell'}
    &\infer{\phantom{ }}{\mathtt{join}\ : \forall x,y : \Label . \Label}
  \end{align*}
The implementation of a concatenation function for labeled strings \TTsec is presented in Figure \ref{fig:ttsec-concat}.
\begin{figure}[htb!]
  \centering
\begin{lstlisting}[language = TTsec]
concat : forall$\ell$,$\ell$': Label.forallx:Labeled $\ell$ String.
                      forally:Labeled $\ell$' String.DIO (join $\ell$ $\ell$') String
concat = lambda$\ell$,$\ell$':Label.lambdax:Labeled $\ell$ String.lambday:Labeled $\ell$' String.
         unlabel x $\binds$ (lambdaux:DIO (join $\ell$ $\ell$') String.
         unlabel y $\binds$ (lambdauy:DIO (join $\ell$ $\ell$') String.
         pure (ux $\concat$ uy)))
\end{lstlisting}
  \caption{Concatenation of secure strings in \TTsec.}
  \label{fig:ttsec-concat}
\end{figure}
\\
\texttt{concat} is typed through multiple applications of \TLam, which reduces the problem of typing \texttt{concat} to showing
\begin{align*}
  &\Gamma; (x : \LabeledType{\ell}{\String}) ; (y : \LabeledType{\ell'}{\String})\\
  &\vdash \aunlabel{x} \dots \pure{(ux \concat uy)}: \aDIO{(\texttt{join}\ \ell\ \ell')}{\String}
\end{align*}
and
\begin{align*}
\Gamma \vdash &\forall x : \LabeledType{\ell}{\String}\ .\ \\&\forall y : \LabeledType{\ell'}{\String}\ .\  \aDIO{(\texttt{join}\ \ell\ \ell')}{\String} : \Type.
\end{align*}
The typing of the actual expression with the return type, $\aDIO{(\mathtt{join}\ \ell\ \ell')}{\String}$, continues by \TBind.
This judgment requires that $$\Gamma ;(\ell, \ell' : \Label);(x : \LabeledType{\ell}{\String}) \vdash \aunlabel{x} : \aDIO{(\mathtt{join}\ \ell\ \ell')}{\String}$$ and that the rest of the expression in fact has a function type which takes such an input.
As this goes by rules already presented we continue with the typing of $\aunlabel{x}$.
This follows by \TUnlabel which can be used by the assumption on \texttt{join} and by \TVar.

Showing that the proposed type is a type follows by \TForall, \TVar, \TApp, and the assumption on \texttt{join}.

\subsection{Erasure}\label{sec:erasure}
\begin{definition}[Erasure function on terms]\label{def:term-erasure}
  Let $\ell_{A}$ be the attacker's security level. The function
  $$
  \erasure : \Term \rightarrow \Term
  $$
  denotes the \emph{erasure function on terms} where values and primitive types like \True, \Int, etc. is unaffected but otherwise defined by:
  \begin{align*}
    \eps{\bullet} &\triangleq \bullet \\
    \eps{\lambda x . t} &\triangleq \lambda x . \eps{t} \\
    \eps{\app{t_{1}}{t_{2}} : \tau} &\triangleq
                                      \begin{cases}
                                        \bullet & \text{if } \tau = \aDIO{\ell}{\tau'} \land \ell \not\sqsubseteq \ell_{A} \\
                                        \app{\eps{t_{1}}}{\eps{t_{2}}} & \text{otherwise} \\
                                      \end{cases} \\
    \eps{\ifThenElse{t_{1}}{t_{2}}{t_{3}} : \tau} &\triangleq
                                                    \begin{cases}
                                                      \bullet & \text{if } \tau = \aDIO{\ell}{\tau'} \land \ell \not\sqsubseteq \ell_{A} \\
                                                      \ifThenElse{\eps{t_{1}}}{\eps{t_{2}}}{\eps{t_{3}}} & \text{otherwise}
                                                    \end{cases} \\
    \eps{\pure{t} : \aDIO{\ell}{\tau}} &\triangleq
                                        \begin{cases}
                                          \bullet & \text{if } \ell \not\sqsubseteq \ell_{A} \\
                                          \pure{\eps{t}} & \text{otherwise} \\
                                        \end{cases} \\
    \eps{\DIOValue{t} : \aDIO{\ell}{\tau}} &\triangleq
                                            \begin{cases}
                                              \bullet & \text{if } \ell \not\sqsubseteq \ell_{A} \\
                                              \DIOValue{\eps{t}} & \text{otherwise} \\
                                            \end{cases} \\
    \eps{\bind{t_{1}}{t_{2}} : \aDIO{\ell}{\tau}} &\triangleq
                                                   \begin{cases}
                                                     \bullet & \text{if } \ell \not\sqsubseteq \ell_{A} \\
                                                     \bind{\eps{t_{1} }}{\eps{t_{2}}} & \text{otherwise}
                                                   \end{cases} \\
    \eps{\LabeledValue{t} : \LabeledType{\ell}{\tau}} &\triangleq
                                                        \begin{cases}
                                                          \LabeledValue{\bullet} & \text{if } \ell \not\sqsubseteq \ell_{A} \\
                                                          \LabeledValue{\eps{t}} & \text{otherwise}
                                                        \end{cases} \\
    \eps{\ylabel{t} : \LabeledType{\ell}{\tau}} &\triangleq
                                                  \begin{cases}
                                                    \ylabel{\bullet} & \text{if } \ell \not\sqsubseteq \ell_{A} \\
                                                    \ylabel{\eps{t}} & \text{otherwise}
                                                  \end{cases} \\
    \eps{\aunlabel{t} : \aDIO{\ell}{\tau}} &\triangleq
                                           \begin{cases}
                                             \bullet & \text{if } \ell \not\sqsubseteq \ell_{A} \\
                                             \aunlabel{\eps{t}}  & \text{otherwise}
                                           \end{cases}\\
    \eps{\RefValue{n} : \RefType{\ell}{\tau}} &\triangleq
                                                \begin{cases}
                                                  \RefValue{\bullet} & \text{if } \ell \not\sqsubseteq \ell_{A} \\
                                                  \RefValue{n} & \text{otherwise}
                                                \end{cases} \\
    \eps{\newRef{\ell'}{t} : \aDIO{\ell}{(\RefType{\ell'}{\tau}})} &\triangleq
                                                               \begin{cases}
                                                                 \bullet & \text{if } \ell \not\sqsubseteq \ell_{A} \\
                                                                 \newRef{\ell'}{\eps{t}} & \text{otherwise}
                                                               \end{cases} \\
    \eps{\writeRef{t_{1}}{t_{2}} : \aDIO{\ell}{\tau}} &\triangleq
                                                       \begin{cases}
                                                         \bullet  & \text{if } \ell \not\sqsubseteq \ell_{A} \\
                                                         \writeRef{\eps{t_{1}}}{\eps{t_{2}}} & \text{otherwise}
                                                       \end{cases} \\
    \eps{\readRef{t} : \aDIO{\ell}{(\LabeledType{\ell'}{\tau})}} &\triangleq
                                                                  \begin{cases}
                                                                    \bullet &  \text{if } \ell \not\sqsubseteq \ell_{A} \\
                                                                    \readRef{\bullet} & \text{if } \ell' \not\sqsubseteq \ell_{A} \\
                                                                    \readRef{\eps{t}} & \text{otherwise}
                                                                  \end{cases} \\
    \eps{\aplug{t} : \aDIO{\ell}{(\LabeledType{\ell'}{\tau}})} &\triangleq
                                                               \begin{cases}
                                                                 \bullet & \text{if } \ell \not\sqsubseteq \ell_{A} \\
                                                                 \plugHole{\eps{t}} & \text{if } \ell' \not\sqsubseteq \ell_{A} \\
                                                                 \aplug{\eps{t}} & \text{otherwise}
                                                               \end{cases} \\
    \eps{\plugHole{t} : \aDIO{\ell}{(\LabeledType{\ell'}{\tau}})} &\triangleq
                                                                   \begin{cases}
                                                                     \bullet & \text{if } \ell \not\sqsubseteq \ell_{A} \\
                                                                     \plugHole{\eps{t}} & \text{otherwise}
                                                                   \end{cases} \\
    \eps{\aDIO{\ell}{\tau}} &\triangleq \aDIO{\eps{\ell}}{\eps{\tau}} \\
    \eps{\LabeledType{\ell}{\tau}} &\triangleq \LabeledType{\eps{\ell}}{\eps{\tau}} \\
    \eps{\RefType{\ell}{\tau}} &\triangleq \RefType{\eps{\ell}}{\eps{\tau}}\\
    \eps{\forall x : \tau . t } &\triangleq \forall x : \eps{\tau} . \eps{t}
  \end{align*}
\end{definition}
In most cases the definition of the erasure function is straightforward as we simply collapse sensitive information and computations to $\bullet$ if they are above the security level of the attacker and otherwise apply the function homomorphically.
In one particular case this idea fails, namely the erasure of the term $\aplug{t}$.

Consider $\aplug{t} : \aDIO{\ell}{(\LabeledType{\ell'}{\tau})}$ for some $\ell, \ell'$, and $\tau$.
If the adversary is not allowed to see $\ell$, i.e.
$\ell\not\sqsubseteq \ell_{A}$, the computation should not be visible to the adversary and therefore it should be completely collapsed into $\bullet$.
Unfortunately, this approach of rewriting entire computations fails if $\ell \sqsubseteq \ell_{A}$ and $\ell'\not\sqsubseteq \ell_{A}$ as it would not be possible to show that $\conf{\eps{\Sigma}}{\aplug{\bullet}}\evalsto \conf{\eps{\Sigma'}}{\pure{\LabeledValue{\bullet}}}$ as $\conf{\Sigma}{\bullet} \not \bigstepto \conf{\Sigma'}{\bullet}$ since $\bullet \leadsto \bullet$ and it does therefore not have a normal form.
Hence, we need a context-sensitive erasure function as the idea about simply erasing computations above the level of an attacker is too simple.
To handle this case soundly we make use of \textit{two-steps erasure} that works by introducing an extra semantic step for $\mathtt{plug}_{\bullet}$ introduced by the erasure function.
The extension is presented in Figure~\ref{fig:ttsec-hole-sem}.
Note that this, from an attackers point of view, still looks exactly like one would expect when erasing data and this is therefore purely a technicality of the proof.
\begin{figure}[H]
  \centering
  \begin{mathparpagebreakable}
    \infer[Hole]
    { \phantom{ } }
    { \bullet \leadsto \bullet }

    \infer[Plug$_{\bullet}$]
    { \phantom{ } }
    { \plugHole{t} \leadsto \pure{(\LabeledValue{\bullet})} }
  \end{mathparpagebreakable}

  \caption{Operational semantics of \TTsecbullet: extensions to \TTsec.}
  \label{fig:ttsec-hole-sem}
\end{figure}

Extra typing rules for both $\bullet$ and $\mathtt{plug}_{\bullet}$ are also introduced as presented in Figure~\ref{fig:ttsec-hole-typ}.
\begin{figure}[H]
  \centering
  \begin{mathparpagebreakable}
    \infer[T-Hole]
    { \Gamma \vdash \tau : \Type }
    { \Gamma \vdash \bullet : \tau }

    \infer[T-Plug$_{\bullet}$]
    { \ell_L \flowsto \ell_H \\ \Gamma \vdash s : \aDIO{\ell_H}{S} }
    { \Gamma \vdash \plugHole{s} : \aDIO{\ell_L}{(\LabeledType{\ell_H}{S})} }
  \end{mathparpagebreakable}
  \caption{Typing rules for \TTsecbullet: extensions to \TTsec.}
  \label{fig:ttsec-hole-typ}
\end{figure}
The definition of \erasure on stores is straightforward as we have a compartmentalized memory.
If the a store is erased up to a security level $\ell_{A}$ then all levels above this should simply be collapsed entirely.
\begin{definition}[Erasure function on configurations]\label{def:conf-erasure}
  Let $\ell_{A}$ be the attacker's security level.
  The function
  $$
  \erasure : \Conf \rightarrow \Conf_{\bullet}
  $$
  denotes the erasure function for configurations defined by
  $$
  \eps{\conf{\Sigma}{t : \aDIO{\ell}{\tau}}} \triangleq
  \begin{cases}
    \conf{\eps{\Sigma}}{\bullet} & \text{if } \ell \not\sqsubseteq \ell_{A} \\
    \conf{\eps{\Sigma}}{\eps{t}} & \text{otherwise}
  \end{cases}
  $$
  where the store $\Sigma$ is erased pointwise at each security level and in every cell, i.e. $\eps{\Sigma} = \lambda\ell. \eps{\Sigma(\ell)}$, where
  $$
  \eps{\Sigma(\ell)} \triangleq
  \begin{cases}
    \bullet & \text{if } \ell \not\sqsubseteq \ell_{A} \\
    \textit{map}\ \erasure\ \Sigma(\ell) & \text{otherwise}
  \end{cases}
  $$
\end{definition}
Note that writing to an erased cell yields no update, i.e. $(\Sigma(\ell)[\bullet] \Coloneqq t) \triangleq \Sigma(\ell)$, and reading from an erased compartment yields $\bullet$, i.e.  $\bullet[n] \triangleq \bullet$.

\begin{definition}[$\ell_{A}$-equivalence]\label{def:erasure-eq}
  Let $c_{1}, c_{2} \in \Conf$.
  $c_{1}$ and $c_{2}$ are said to be \emph{indistinguishable} from security level $\ell_{A}$, written $c_{1} \approx_{\ell_{A}} c_{2} $, if and only if $\eps{c_{1}}$ and $\eps{c_{2}}$ are structurally equivalent, written $\eps{c_{1}} \equiv \eps{c_{2}}$.
\end{definition}

\section{Results}\label{sec:results}

\begin{lemma}[Erasure of substitution] \label{lem:eps-subst}
  Let $t, v \in \Term$. Then $\eps{t[v/x]} \equiv \eps{t}[\eps{v}/x]$.
\end{lemma}
\begin{proof}
  The statement follows by case splitting on $t$ and $v$ and the definition of \erasure.
\end{proof}
\begin{lemma}[Distributivity on pure term reduction]\label{lem:dist-pure}
  Let $t_{1}, t_{2} \in \Term$.
  If $t_{1} \leadsto t_{2}$ then $\eps{t_{1}} \leadsto \eps{t_{2}}$.
\end{lemma}
\begin{proof}
  The proof goes by structural induction in the derivation of $t_{1} \leadsto t_{2}$.
  \begin{enumerate}[align = left]
  \item[\textsc{App$_{1}$}:] Assume $\app{t_{1}}{t_{2}} \leadsto \app{t_{1}'}{t_{2}}$.
    If $\app{t_{1}}{t_{2}}$ has type $\aDIO{\ell}{\tau'}$ and $\ell \not\sqsubseteq \ell_{A}$, the statement follows from the definition of \erasure\ and \textsc{Hole}.
    Otherwise, $t_{1} \leadsto t_{1}'$ holds by $\textsc{App}_{1}$ and by the induction hypothesis $\eps{t_{1}} \leadsto \eps{t_{1}'}$.
    By $\textsc{App}_{1}$ and definition of $\erasure$ it holds that $\eps{\app{t_{1}}{t_{2}}} \leadsto \eps{\app{t_{1}'}{t_{2}}}$.
  \item[\textsc{App$_{2}$}:] The argument is identical to the $\textsc{App}_{1}$ case.
  \item[\textsc{Beta}:] Assume $(\lambda x .
    t) \ v \leadsto t[v/ x]$.
    By \textsc{Beta}, $\lambda x .
    \eps{t} \ \eps{v} \leadsto \eps{t}[\eps{v}/ x]$.
    By definition of \erasure and Lemma \ref{lem:eps-subst} then $\eps{(\lambda x .
      t)\ v} \leadsto \eps{t[v/x]}$.
  \item[\textsc{If}$_{1}$:] The argument is identical to the $\textsc{App}_{1}$ case.
  \item[\textsc{If}$_{2}$ and \textsc{If}$_{3}$:] If $t_{1}$ and $t_{2}$ have type $\aDIO{\ell}{\tau'}$ and $\ell \not\sqsubseteq \ell_{A}$, the statement follows from the definition of \erasure\ and \textsc{Hole}.
    Otherwise, the statement follows directly by the definition \erasure\ and \textsc{If$_{i}$}.
  \item[\textsc{Bind$_{1}$}:] Assume $\pure{\bind{t_{1}}{t_{2}}} \leadsto \app{t_{2}}{t_{1}}$.
    As we assume well-typed terms, $\pure{\bind{t_{1}}{t_{2}}}$ has type $\aDIO{\ell}{\tau}$ for some $\ell$ and $\tau$.
    If $\ell \not\sqsubseteq \ell_{A}$ the statement follows from the definition of \erasure\ and \textsc{Hole}.
    Otherwise, the statement follows from \textsc{Bind-Pure} and the definition of \erasure.
  \item[\textsc{Pure$_{1}$}:] Assume $\pure{t} \leadsto \pure{t'}$.
    As we assume well-typed terms, $\pure{t}$ and $\pure{t'}$ have type $\aDIO{\ell}{\tau}$ for some $\ell$ and $\tau$.
    If $\ell \sqsubseteq \ell_{A}$ then the statement follows from the induction hypothesis, the definition of \erasure\ and \textsc{Pure$_{1}$}.
    If $\ell \not\sqsubseteq \ell_{A}$ then the statement follows from the definition \erasure\ and $\textsc{Hole}$.
  \item[\textsc{Pure$_{2}$}:] Assume $\pure{v} \leadsto \DIOValue{v}$.
    As we assume well-typed terms, $\pure{v}$ and $\DIOValue{v}$ have type $\aDIO{\ell}{\tau}$ for some $\ell$ and $\tau$.
    If $\ell \sqsubseteq \ell_{A}$ then the statement follows from the definition of \erasure\ and \textsc{Pure$_{2}$}.
    If $\ell \not\sqsubseteq \ell_{A}$ then the statement follows from the definition \erasure\ and $\textsc{Hole}$.
  \item[\textsc{Label$_{1}$}:] Assume $\ylabel{t} \leadsto \ylabel{t'}$.
    As we assume well-typed terms, $\ylabel{t}$ and $\ylabel{t'}$ have type $\LabeledType{\ell}{\tau}$ for some $\ell$ and $\tau$.
    If $\ell \sqsubseteq \ell_{A}$ then the statement follows from the induction hypothesis, the definition of \erasure, and \textsc{Label$_{1}$}.
    If $\ell \not\sqsubseteq \ell_{A}$ then the statement follows from the definition \erasure\, $\textsc{Hole}$ and \textsc{Label$_{1}$}.
  \item[\textsc{Label$_{2}$}:] Assume $\ylabel{t} \leadsto {\LabeledValue{t}}$.
    As we assume well-typed terms, $\ylabel{t}$ and $\LabeledValue{t}$ have type $\LabeledType{\ell}{\tau}$ for some $\ell,\ell'$, and $\tau$.
    If $\ell \not\sqsubseteq \ell_{A}$ the statement follows from the definition of \erasure\, \textsc{Hole} and \textsc{Label$_{1}$}.
    Otherwise the statement follows by \textsc{Label$_{2}$}, the definition of \erasure and the induction hypothesis.

  \item[\textsc{Unlabel}$_1$:] Assume $\aunlabel{t} \leadsto \aunlabel{t'}$.
    As we assume well-typed terms, $\aunlabel{t}$ and $\aunlabel{t'}$ have type $\aDIO{\ell}{\tau}$ for some $\ell$ and $\tau$.
    If $\ell \not\sqsubseteq \ell_{A}$ the statement follows from the definition of \erasure\ and \textsc{Hole}.
    Otherwise, the statement follows from the induction hypothesis, $\textsc{Unlabel}_{1}$, and the definition of \erasure.
  \item[\textsc{Unlabel}$_2$:] Assume $\aunlabel{(\LabeledValue{t})} \leadsto \pure{t}$.
    As we assume well-typed terms, $\aunlabel{(\LabeledValue{t})}$ and $\pure{t}$ have type $\aDIO{\ell}{\tau}$ where $\LabeledValue{t}$ have type $\LabeledType{\ell'}{\tau}$ such that $\ell' \sqsubseteq \ell$.
    If $\ell \not\sqsubseteq \ell_{A}$ the statement follows from the definition of \erasure\ and \textsc{Hole}.
    If $\ell \sqsubseteq \ell_{A}$ then by transitivity of the partial ordering $\ell' \sqsubseteq \ell_{A}$ holds and the statement then follows by $\textsc{Unlabel}_{2}$ and the definition of \erasure.
  \item[\textsc{New}$_{1}$:] Assume $\newRef{\ell}{t} \leadsto \newRef{\ell}{t'}$.
    As we assume well-typed terms, $\newRef{\ell}{t}$ and $\newRef{\ell}{t'}$ have type $\aDIO{\ell}{\tau}$ for some $\ell$ and $\tau$.
    If $\ell \not\sqsubseteq \ell_{A}$ then the statements follows from the definition of \erasure\ and \textsc{Hole}.
    If $\ell \sqsubseteq \ell_{A}$ the statement follows from the induction hypothesis, the definition of \erasure, and $\textsc{New}_{1}$.
  \item[\textsc{Write}$_{1}$:] Assume $\writeRef{t_{1}}{t_{2}} \leadsto \writeRef{t_{1}'}{t_{2}}$.
    As we assume well-typed terms, $\writeRef{t_{1}}{t_{2}}$ and $\writeRef{t_{1}'}{t_{2}}$ have type $\aDIO{\ell}{\Unit}$ and $t_{2}$ has type $\LabeledType{\ell'}{\tau}$.
    If $\ell \not\sqsubseteq \ell_{A}$ the statement follows from the definition of \erasure\ and \textsc{Hole}.
    If $\ell \sqsubseteq \ell_{A}$ then the statement follows by the induction hypothesis, $\textsc{Write}_{1}$, and the definition of \erasure.
  \item[\textsc{Write}$_{2}$:] The argument is identical to the $\textsc{Write}_{1}$ case.
  \item[\textsc{Read}$_{1}$:] Assume $\readRef{t} \leadsto \readRef{t'}$.
    As we assume well-typed terms $\readRef{t}$ and $\readRef{t'}$ have type $\aDIO{\ell'}{(\LabeledType{\ell}{\tau})}$.
    If $\ell' \not\sqsubseteq \ell_{A}$ the statement follows from the definition of \erasure\ and \textsc{Hole}.
    Otherwise, if $\ell' \sqsubseteq \ell_{A}$ then consider whether $\ell \sqsubseteq \ell_{A}$ holds.
    If $\ell \not\sqsubseteq \ell_{A}$ the statement follows by $\textsc{Read}_{1}$ and \textsc{Hole}.
    If $\ell \sqsubseteq \ell_{A}$ then the statement follows by the induction hypothesis, $\textsc{Read}_{1}$, and the definition of \erasure.
  \item[\textsc{DIO}$_{i}$, \textsc{Labeled}$_{i}$,]
  \item[\textsc{Ref}$_{i}$, \textsc{Forall}$_{i}$:] In all cases, the statement follows directly from the induction hypothesis, the definition of \erasure, and the inference rules.
  \item[\textsc{Plug}$_{\bullet}$:] Assume $\plugHole{t} \leadsto \pure{(\LabeledValue{\bullet})}$.
    As we assume well-typed terms, $\plugHole{t}$ has type $\aDIO{\ell}{(\LabeledType{\ell'}{\tau})}$.
    Consider whether $\ell' \sqsubseteq \ell_{A}$.
    In both cases, the statement follows by the definition \erasure\ and $\textsc{Plug}_{\bullet}$.
  \item[\textsc{Hole}:] The statement follows directly from the definition of \erasure\ and \textsc{Hole}.
  \end{enumerate}
\end{proof}
\begin{lemma}[Erasure of a computation] \label{lem:dio-erasure} Let $t \in \Term$.
  If $t$ has type \aDIO{\ell}{\tau} and $\ell \not\sqsubseteq \ell_{A}$ then $\eps{t} \equiv \bullet$.
\end{lemma}
\begin{proof}
  The statement follows directly by case splitting on $t$ and the definition of \erasure.
\end{proof}
\begin{lemma}[Single step erased store equivalence]\label{lem:single-store-eq}
  Let $c_{1} = \conf{\Sigma_{1}}{t_{1}}$ and $c_{2} = \conf{\Sigma_{2}}{t_{2}}$.
  If $t_{1}$ and $t_{2}$ have type $\aDIO{\ell}{\tau}$, $\ell \not\sqsubseteq \ell_{A}$, and $c_{1} \evalsto c_{2}$ then $\eps{\Sigma_{1}} \equiv \eps{\Sigma_{2}}$.
\end{lemma}
\begin{proof}
  The proof goes by case splitting in the derivation of $c_{1}\evalsto c_{2}$.

  \begin{enumerate}[align = left]
  \item[\textsc{Plug}:] Assume $\conf{\Sigma}{\aplug{t}{}} \evalsto \conf{\Sigma'}{\pure{(\LabeledValue{t'})}}$.
    As we assume well-typed terms, $t$ has type $\aDIO{\ell'}{\tau'}$ for some $\ell'$ and $\tau'$ where $\ell \sqsubseteq \ell'$.
    By transitivity of $\sqsubseteq$ it follows that $\ell' \not\sqsubseteq \ell_{A}$ and then the statement follows by Lemma \ref{lem:multi-store-eq} as $t$ is structurally smaller than \aplug{t}.
  \item[\textsc{New$_{2}$}:] Assume $ \conf{\Sigma}{\newRef{\ell'}{(\LabeledValue{t})}} \evalsto \conf{\Sigma(\ell')[n] \Coloneqq t}{\pure{(\RefValue{\ell'}{n})}}$.
    As we assume well-typed terms, $\newRef{\ell'}{(\LabeledValue{t})}$ and $\pure{(\RefValue{\ell'}{n})}$ have type $\aDIO{\ell}{(\RefType{\ell'}{\tau})}$ where $\ell\sqsubseteq \ell'$.
    By transitivity of $\sqsubseteq$ it follows that $\ell'\not\sqsubseteq \ell_{A}$.
    Note that the only memory compartment changed is the one of $\ell'$ and as writing to an erased cell makes no update the statement follows.
  \item[\textsc{Write$_{3}$}:] The argument is identical to the \textsc{New} case.
  \end{enumerate}
  The remaining cases does not change the store and the statement follows immediately.
\end{proof}
\begin{lemma}[Multi-step erased store equivalence]\label{lem:multi-store-eq}
  Let $c_{1} = \conf{\Sigma_{1}}{t_{1}}$ and $c_{2} = \conf{\Sigma_{2}}{t_{2}}$.
  If $t_{1}$ and $t_{2}$ have type $\aDIO{\ell}{\tau}$, $\ell \not\sqsubseteq \ell_{A}$, and $c_{1} \evalstostar c_{2}$ then $\eps{\Sigma_{1}} \equiv \eps{\Sigma_{2}}$.
\end{lemma}
\begin{proof}
  The statement follows from repeated applications of Lemma \ref{lem:single-store-eq}.
\end{proof}

\begin{lemma}\label{lem:store-update-erasure}
  Let $\Sigma$ be a store, $n \in \mathbb{N} \cup \{ \bullet \} $, and $t \in \Term$.
  Then $\eps{\Sigma(\ell)[n] \Coloneqq t} \equiv \eps{\Sigma}(\ell)[n] \Coloneqq \eps{t}$.
\end{lemma}
\begin{proof}
  Note that the erasure of a store erases each compartment in both the old and the updated store.
  Only the $\Sigma(\ell)$ compartment is changed, hence the remaining compartments are preserved by definition.
  It remains to show that the updated compartments are equivalent.
  Let the updated store $\Sigma(\ell)[n] \Coloneqq t$ be denoted by $\Sigma'$.
  If $\ell \sqsubseteq \ell_{A}$ then $\eps{\Sigma'(\ell)} = \textit{map}\ \erasure\ \Sigma'(\ell)$ cf.
  the definition of \erasure and hence the statement follows from properties of \textit{map}.
  If $\ell \not\sqsubseteq \ell_{A}$ then $\eps{\Sigma'(\ell)} \equiv \bullet$ and the statement follows as updating an erased cell yields no update.
\end{proof}
\begin{lemma}\label{lem:store-read-erasure}
  Let $\Sigma$ be a store and $n \in \mathbb{N} $.
  If $\ell \sqsubseteq \ell_{A}$ then $\eps{\Sigma(\ell)[n]} \equiv \eps{\Sigma}(\ell)[n]$.
\end{lemma}
\begin{proof}
  The statements follows from the definition of \erasure\ and properties of \textit{map}.
\end{proof}
\begin{proposition}[Distributivity of \erasure\ over $\evalsto$]\label{prop:single-dist-erasure}
  Let $c_{1}, c_{2} \in \Conf$.
  If $c_{1} \evalsto c_{2}$ then $\eps{c_{1}} \evalsto \eps{c_{2}}$.
\end{proposition}
\begin{proof}
  Let $c_{1} = \conf{\Sigma}{t}$ and $c_{2} = \conf{\Sigma'}{t'}$.
  The proof goes by structural induction in the derivation of $c_{1} \evalsto c_{2}$.
  \begin{enumerate}[align = left]
  \item[\textsc{Lift}:] Note $\Sigma = \Sigma'$. Necessarily, $t \leadsto t'$ must hold. By Lemma \ref{lem:dist-pure} it holds $\eps{t} \leadsto \eps{t'}$ and by \textsc{Lift} and the definition of \erasure on configurations the statement follows.
  \end{enumerate}
  Note that in the remaining cases $t$ and $t'$ necessarily have type $\aDIO{\ell}{\tau}$ as we assume well-typed terms.
  If $\ell \not\sqsubseteq \ell_{A}$ the statement in all cases follows by the definition of \erasure, Lemma \ref{lem:dio-erasure}, Lemma \ref{lem:single-store-eq}, \textsc{Lift}, and \textsc{Hole}.
  Now, assume $\ell \sqsubseteq \ell_{A}$.
  \begin{enumerate}[align = right]
  \item[\textsc{Bind$_{2}$}:] Assume $\conf{\Sigma}{\bind{t_{1}}{t_{2}}} \evalsto \conf{\Sigma'}{\bind{t_{1}'}{t_{2}}}$.
    The statement follows by \textsc{Bind$_{2}$}, the induction hypothesis and the definition of \erasure.
  \item[\textsc{Plug}:] Assume $\conf{\Sigma}{\aplug{t}{}} \evalsto \conf{\Sigma'}{\pure{(\LabeledValue{t'})}}$.
    As we assume well-typed terms, $t$ has type $\aDIO{\ell'}{\tau'}$ for some $\ell'$ and $\tau'$.
    From \textsc{Plug} it holds that $\conf{\Sigma}{t} \bigstepto \conf{\Sigma'}{\DIOValue{t'}}$.
    If $\ell' \sqsubseteq \ell_{A}$, cf.
    Lemma \ref{lem:bigstep-dist-erasure} and that $t$ is structurally smaller than \aplug{t}, it follows $\eps{\conf{\Sigma}{t}} \bigstepto \eps{\conf{\Sigma'}{\DIOValue{t'}}}$, and from the definition of \erasure and \textsc{Plug} the statement holds.
    If $\ell' \not\sqsubseteq \ell_{A}$, it follows from Lemma \ref{lem:multi-store-eq} that $\eps{\Sigma} \equiv \eps{\Sigma'}$, and using \textsc{Lift}, the definition of \erasure and \textsc{Plug$_{\bullet}$} the statement follows.
  \item[\textsc{New$_{2}$}:] Assume $ \conf{\Sigma}{\newRef{\ell''}{(\LabeledValue{t})}} \evalsto \conf{\Sigma(\ell'')[n] \Coloneqq t}{\pure{(\RefValue{\ell''}{n})}}$.
    As we assume well-typed terms, terms $\newRef{\ell''}{(\LabeledValue{t})}$ and $\pure{\RefValue{\ell''}{n}}$ have type $\aDIO{\ell}{(\RefType{\ell''}{\tau})}$, and the term $\LabeledValue{t}$ has type $\LabeledType{\ell'}{\tau}$ where $\ell \sqsubseteq \ell' \sqsubseteq \ell''$.
    If $\ell'' \sqsubseteq \ell_{A}$ then by transitivity $\ell' \sqsubseteq \ell_{A}$.
    From the definition of \erasure\ it follows $\eps{\Sigma(\ell'')} = \textit{map}\ \erasure\ \Sigma(\ell'')$ and hence $|\eps{\Sigma(\ell'')}| = |\Sigma(\ell'')| = n$.
    From Lemma \ref{lem:store-update-erasure}, the definition of \erasure, and $\textsc{New}_{2}$ the statement follows.
    If $\ell'' \not\sqsubseteq \ell_{A}$ then then from the definition of \erasure\ it follows $\eps{\Sigma(\ell'')} \equiv \bullet$ and the size of an erased label segment is also erased, hence $|\eps{\Sigma(\ell'')}| = |\bullet| = \bullet$.
    Consider whether $\ell' \sqsubseteq \ell_{A}$.
    In both cases the statement follows from Lemma \ref{lem:store-update-erasure}, the definition of \erasure, and \textsc{New$_{2}$}.
  \item[\textsc{Write$_{3}$}:] Assume $$\conf{\Sigma}{\writeRef{(\RefValue{\ell'}{n})}{(\LabeledValue{t})}} \evalsto \conf{\Sigma(\ell')[n] \Coloneqq t}{\pure{()}}.$$
    As we assume well-typed terms, $\RefValue{\ell'}{n}$ has type $\RefType{\ell'}{\tau}$ and $\LabeledValue{t}$ has type $\LabeledType{\ell''}{\tau}$ for some $\ell'$, $\ell''$, and $\tau$ where $\ell'' \sqsubseteq \ell'$.
    If $\ell' \sqsubseteq \ell_{A}$ then by transitivity $\ell'' \sqsubseteq \ell_{A}$ and hence
    \begin{align*}
      &\writeRef{(\RefValue{\ell'}{n})}{(\LabeledValue{\eps{t}})} \\
      &\equiv \writeRef{\eps{\RefValue{\ell'}{n}}}{\eps{\LabeledValue{t}}}
    \end{align*}
    by definition of \erasure. The statement now follows from Lemma \ref{lem:store-update-erasure}, \textsc{Write$_{3}$}, and the definition of \erasure. If $\ell' \not\sqsubseteq \ell_{A}$ then consider whether $\ell'' \sqsubseteq \ell_{A}$. In either case, the statement follows from Lemma \ref{lem:store-update-erasure}, \textsc{Write$_{3}$}, and the definition of \erasure.
  \item[\textsc{Read$_{2}$}:] Assume $\conf{\Sigma}{\readRef{(\RefValue{\ell}{n})}} \evalsto \conf{\Sigma}{\pure{\big(\LabeledValue{\Sigma(\ell)[n]}\big)}}$.
    As we assume well-typed terms, $\RefValue{\ell}{n}$ has type $\RefType{\ell}{\tau}$ and $\pure{(\LabeledValue{\Sigma(\ell)[n]})}$ has type $\aDIO{\ell'}{(\LabeledType{\ell}{\tau})}$ for some $\ell, \ell'$ and $\tau$ where $\ell' \sqsubseteq \ell$.
    If $\ell \sqsubseteq \ell_{A}$ then by transitivity $\ell' \sqsubseteq \ell_{A}$ and the statement follows from Lemma \ref{lem:store-read-erasure}, $\textsc{Read}_{2}$, and the definition of \erasure.
    If $\ell \not\sqsubseteq \ell_{A}$ the statement follows by the fact that reading from an erased compartment yields $\bullet$, the definition of \erasure, and $\textsc{Read}_{2}$.
  \end{enumerate}
\end{proof}
\begin{lemma}[Distributivity of \erasure\ over $\evalstostar$]\label{lem:multi-dist-erasure}
  Let $c_{1}, c_{2} \in \Conf$.
  If $c_{1} \evalstostar c_{2}$ then $\eps{c_{1}} \evalstostar \eps{c_{2}}$.
\end{lemma}
\begin{proof}
  The statement follows from repeated applications of Proposition \ref{prop:single-dist-erasure}.
\end{proof}
\begin{lemma}[Distributivity of \erasure\ over $\bigstepto$]\label{lem:bigstep-dist-erasure}
  Let $c_{1}, c_{2} \in \Conf$.
  If $c_{1} \bigstepto c_{2}$ then $\eps{c_{1}} \bigstepto \eps{c_{2}}$.
\end{lemma}
\begin{proof}
  The statement follows directly from Lemma \ref{lem:multi-dist-erasure}.
\end{proof}

\begin{lemma}[Determinacy of pure reductions]\label{lem:pure-det}
  Let $t_{1}, t_{2}, t_{3} \in \Term$.
  If $t_{1} \leadsto t_{2}$ and $t_{1} \leadsto t_{3}$ then $t_{2} \equiv t_{3}$.
\end{lemma}
\begin{proof}
  The proof goes by structural induction in the derivation of $t_{1} \leadsto t_{2}$ and $t_{1} \leadsto t_{3}$.
\end{proof}

\begin{proposition}[Single step determinacy]\label{prop:single-det}
  Let $c_{1}, c_{2}, c_{3} \in \Conf$.
  If $c_{1} \evalsto c_{2}$ and $c_{1} \evalsto c_{3}$ then $c_{2} \equiv c_{3}$.
\end{proposition}
\begin{proof}
  The proof goes by structural induction in the derivation of $c_{1} \evalsto c_{2}$ and $c_{1} \evalsto c_{3}$.
  Note that both $c_{1} \evalsto c_{2}$ and $c_{1} \evalsto c_{3}$ have to have been derived from the same inference rule, syntactically decidable from $c_{1}$.

  \begin{enumerate}[align = left]
  \item[\textsc{Lift}:] The statement follows by Lemma \ref{lem:pure-det}.
  \item[\textsc{Plug}:] Assume $\conf{\Sigma_{1}}{\aplug{t_{1}}} \evalsto \conf{\Sigma_{2}}{\pure{(\LabeledValue{t_{2}})}}$ and $\conf{\Sigma_{1}}{\aplug{t_{1}}} \evalsto \conf{\Sigma_{3}}{\pure{(\LabeledValue{t_{3}})}}$.
    The statement follows from Lemma \ref{lem:bigstep-det} as $t_{1}$ is structurally smaller than $\aplug t_{1}$.
  \end{enumerate}
  In the remaining cases the statement follows by standard structural induction.
\end{proof}
\begin{lemma}[Big step determinacy]\label{lem:bigstep-det}
  Let $c_{1}, c_{2}, c_{3} \in \Conf$.
  If $c_{1} \bigstepto c_{2}$ and $c_{1} \bigstepto c_{3}$ then $c_{2} \equiv c_{3}$.
\end{lemma}
\begin{proof}
  The statement follows from repeated applications of Proposition \ref{prop:single-det}.
\end{proof}

\begin{proposition}[Single step $\approx_{\ell_{A}}$ preservation]\label{prop:single-noninterference}
  Let $c_{1}, c_{1}', c_{2}, c_{2}' \in \Conf$.
  If $c_{1} \approx_{\ell_{A}} c_{2}$, $c_{1} \evalsto c_{1}'$, and $c_{2} \evalsto c_{2}'$ then $c_{1}' \approx_{\ell_{A}} c_{2}'$.
\end{proposition}
\begin{proof}
  Proposition \ref{prop:single-dist-erasure} states that $\eps{c_{1}}\evalsto \eps{c_{1}'}$ and $\eps{c_{2}}\evalsto \eps{c_{2}'}$.
  From Definition \ref{def:erasure-eq} it is known that $\eps{c_{1}} \equiv \eps{c_{2}}$ and from Proposition \ref{prop:single-det} it follows that $\eps{c_{1}'} \equiv \eps{c_{2}'}$.
  Hence $c_{1}' \approx_{{\ell_{A}}} c_{2}'$.
\end{proof}

\begin{theorem}[Progress-insensitive noninterference]\label{thm:noninterference}
  Let $c_{1}, c_{1}', c_{2}, c_{2}' \in \Conf$.
  If $c_{1} \approx_{\ell_{A}} c_{2}$, $c_{1} \bigstepto c_{1}'$, and $c_{2} \bigstepto c_{2}'$ then $c_{1}' \approx_{\ell_{A}} c_{2}'$.
\end{theorem}
\begin{proof}
  The statement follows from repeated application of Proposition \ref{prop:single-noninterference}.
\end{proof}

\section{\DepSec} \label{app:core-lib}

\begin{Verbatim}[commandchars=\\\{\}, numbers = left, frame=single, label=DIO.idr,]
\IdrisKeyword{module} DepSec.DIO

% access public export

||| Security Monad
||| @ l security label of wrapped value
||| @ valueType type of wrapped value
\IdrisKeyword{data} \IdrisType{DIO} : \IdrisBound{\IdrisBound{l}}
        -> (\IdrisBound{valueType} : \IdrisType{Type})
        -> \IdrisType{Type} \IdrisKeyword{where}
  ||| TCB
  \IdrisData{MkDIO} : \IdrisFunction{\IdrisFunction{IO}} \IdrisBound{\IdrisBound{valueType}} -> \IdrisType{\IdrisType{DIO}} \IdrisBound{\IdrisBound{l}} \IdrisBound{\IdrisBound{valueType}}

||| Executes secure computation
||| TCB
||| @ dio secure computation
\IdrisFunction{run} : (\IdrisBound{dio} : \IdrisType{\IdrisType{DIO}} \IdrisBound{\IdrisBound{l}} \IdrisBound{\IdrisBound{a}}) -> \IdrisFunction{\IdrisFunction{IO}} \IdrisBound{\IdrisBound{a}}
\IdrisFunction{\IdrisFunction{run}} (\IdrisData{\IdrisData{MkDIO}} \IdrisBound{m}) = \IdrisBound{\IdrisBound{m}}

||| Lifts arbitrary IO monad into security monad
||| TCB
||| @ io computation
\IdrisFunction{lift} : (\IdrisBound{io} : \IdrisFunction{\IdrisFunction{IO}} \IdrisBound{\IdrisBound{a}}) -> \IdrisType{\IdrisType{DIO}} \IdrisBound{\IdrisBound{l}} \IdrisBound{\IdrisBound{a}}
\IdrisFunction{\IdrisFunction{lift}} = \IdrisData{\IdrisData{MkDIO}}

\IdrisType{\IdrisType{\IdrisFunction{\IdrisFunction{\IdrisFunction{\IdrisFunction{\IdrisData{\IdrisData{\IdrisBound{\IdrisBound{\IdrisBound{\IdrisBound{\IdrisBound{\IdrisBound{Functor (\IdrisType{\IdrisType{\IdrisType{\IdrisType{\IdrisType{\IdrisType{DIO}}}}}} \IdrisBound{\IdrisBound{\IdrisBound{\IdrisBound{\IdrisBound{\IdrisBound{l}}}}}})}}}}}}}}}}}}}} \IdrisKeyword{where}
  \IdrisFunction{\IdrisFunction{\IdrisBound{\IdrisBound{map \IdrisBound{f} (\IdrisData{\IdrisData{MkDIO}} \IdrisBound{io})}}}} = \IdrisData{\IdrisData{MkDIO}} (\IdrisFunction{\IdrisFunction{map}} \IdrisBound{\IdrisBound{f}} \IdrisBound{\IdrisBound{io}})

\IdrisType{\IdrisType{\IdrisFunction{\IdrisFunction{\IdrisFunction{\IdrisFunction{\IdrisFunction{\IdrisFunction{\IdrisData{\IdrisData{\IdrisBound{\IdrisBound{\IdrisBound{\IdrisBound{\IdrisBound{\IdrisBound{\IdrisBound{\IdrisBound{\IdrisBound{Applicative (\IdrisType{\IdrisType{\IdrisType{\IdrisType{\IdrisType{\IdrisType{\IdrisType{\IdrisType{\IdrisType{\IdrisType{DIO}}}}}}}}}} \IdrisBound{\IdrisBound{\IdrisBound{\IdrisBound{\IdrisBound{\IdrisBound{\IdrisBound{\IdrisBound{\IdrisBound{\IdrisBound{l}}}}}}}}}})}}}}}}}}}}}}}}}}}}} \IdrisKeyword{where}
  \IdrisFunction{\IdrisFunction{\IdrisBound{\IdrisBound{pure}}}} = \IdrisData{\IdrisData{MkDIO}} \IdrisFunction{\IdrisFunction{.}} \IdrisFunction{\IdrisFunction{pure}}
  \IdrisFunction{\IdrisFunction{\IdrisBound{\IdrisBound{(<*>) (\IdrisData{\IdrisData{MkDIO}} \IdrisBound{f}) (\IdrisData{\IdrisData{MkDIO}} \IdrisBound{a})}}}} = \IdrisData{\IdrisData{MkDIO}} (\IdrisBound{\IdrisBound{f}} \IdrisFunction{\IdrisFunction{<*>}} \IdrisBound{\IdrisBound{a}})

\IdrisType{\IdrisType{\IdrisFunction{\IdrisFunction{\IdrisFunction{\IdrisFunction{\IdrisFunction{\IdrisFunction{\IdrisData{\IdrisData{\IdrisBound{\IdrisBound{\IdrisBound{\IdrisBound{\IdrisBound{\IdrisBound{\IdrisBound{\IdrisBound{\IdrisBound{Monad (\IdrisType{\IdrisType{\IdrisType{\IdrisType{\IdrisType{\IdrisType{\IdrisType{\IdrisType{\IdrisType{\IdrisType{\IdrisType{\IdrisType{\IdrisType{\IdrisType{DIO}}}}}}}}}}}}}} \IdrisBound{\IdrisBound{\IdrisBound{\IdrisBound{\IdrisBound{\IdrisBound{\IdrisBound{\IdrisBound{\IdrisBound{\IdrisBound{\IdrisBound{\IdrisBound{\IdrisBound{\IdrisBound{l}}}}}}}}}}}}}})}}}}}}}}}}}}}}}}}}} \IdrisKeyword{where}
  \IdrisFunction{\IdrisFunction{\IdrisBound{\IdrisBound{(>>=) (\IdrisData{\IdrisData{MkDIO}} \IdrisBound{a}) \IdrisBound{f}}}}} = \IdrisData{\IdrisData{MkDIO}} (\IdrisBound{\IdrisBound{a}} \IdrisFunction{\IdrisFunction{>>=}} \IdrisFunction{\IdrisFunction{run}} \IdrisFunction{\IdrisFunction{.}} \IdrisBound{\IdrisBound{f}})
\end{Verbatim}

\begin{Verbatim}[commandchars=\\\{\}, numbers = left, frame=single, label=Labeled.idr]
\IdrisKeyword{module} DepSec.Labeled

\IdrisKeyword{import} \IdrisKeyword{public} DepSec.DIO
\IdrisKeyword{import} \IdrisKeyword{public} DepSec.Poset

% access public export

||| Labeled value
||| @ label label
||| @ valueType type of labeled value
\IdrisKeyword{data} \IdrisType{Labeled} : (\IdrisBound{label} : \IdrisBound{\IdrisBound{labelType}})
            -> (\IdrisBound{valueType} : \IdrisType{Type})
            -> \IdrisType{Type} \IdrisKeyword{where}
  ||| TCB
  \IdrisData{MkLabeled} : \IdrisBound{\IdrisBound{valueType}} -> \IdrisType{\IdrisType{Labeled}} \IdrisBound{\IdrisBound{label}} \IdrisBound{\IdrisBound{valueType}}

||| Label values
||| @ value value to label
\IdrisFunction{label} : \IdrisType{\IdrisType{Poset}} \IdrisBound{\IdrisBound{labelType}}
      => \{\IdrisBound{l} : \IdrisBound{\IdrisBound{labelType}}\}
      -> (\IdrisBound{value} : \IdrisBound{\IdrisBound{a}})
      -> \IdrisType{\IdrisType{Labeled}} \IdrisBound{\IdrisBound{l}} \IdrisBound{\IdrisBound{a}}
\IdrisFunction{\IdrisFunction{label}} = \IdrisData{\IdrisData{MkLabeled}}

||| Unlabel values
||| @ flow evidence that l may flow to l'
||| @ labeled labeled value to unlabel
\IdrisFunction{unlabel} : \IdrisType{Poset} \IdrisBound{labelType}
       => \{\IdrisBound{l},\IdrisBound{l'} : \IdrisBound{labelType}\}
       -> \{\IdrisKeyword{auto} \IdrisBound{flow} : \IdrisBound{l} \IdrisFunction{`leq`} \IdrisBound{l'}\}
       -> (\IdrisBound{labeled} : \IdrisType{Labeled} \IdrisBound{l} \IdrisBound{a})
       -> \IdrisType{DIO} \IdrisBound{l'} \IdrisBound{a}
\IdrisFunction{unlabel} (\IdrisData{MkLabeled} \IdrisBound{val}) = \IdrisFunction{pure} \IdrisBound{val}

||| Upgrade the security level of a labeled value
||| @ flow evidence that l may flow to l'
||| @ labeled labeled value to relabel
\IdrisFunction{relabel} : \IdrisType{Poset} \IdrisBound{labelType}
       => \{\IdrisBound{l}, \IdrisBound{l'} : \IdrisBound{labelType}\}
       -> \{\IdrisKeyword{auto} \IdrisBound{flow} : \IdrisBound{l} \IdrisFunction{`leq`} \IdrisBound{l'}\}
       -> (\IdrisBound{labeled} : \IdrisType{Labeled} \IdrisBound{l} \IdrisBound{a})
       -> \IdrisType{Labeled} \IdrisBound{l'} \IdrisBound{a}
\IdrisFunction{relabel} (\IdrisData{MkLabeled} \IdrisBound{x}) = \IdrisData{MkLabeled} \IdrisBound{x}

\IdrisFunction{unlabel'} : \IdrisType{Poset} \IdrisBound{labelType}
        => \{\IdrisBound{l},\IdrisBound{l'} : \IdrisBound{labelType}\}
        -> \{\IdrisKeyword{auto} \IdrisBound{flow} : \IdrisBound{l} \IdrisFunction{`leq`} \IdrisBound{l'}\}
        -> (\IdrisBound{labeled} : \IdrisType{Labeled} \IdrisBound{l} \IdrisBound{a})
        -> \IdrisType{DIO} \IdrisBound{l'} \IdrisType{(\IdrisBound{c}} \IdrisType{:} \IdrisBound{a} \IdrisType{**} \IdrisFunction{label} \IdrisBound{c} \IdrisType{=} \IdrisBound{labeled\IdrisType{)}}
\IdrisFunction{unlabel'} (\IdrisData{MkLabeled} \IdrisBound{x}) = \IdrisFunction{pure} \IdrisData{(\IdrisBound{x}} \IdrisData{**} \IdrisData{Refl\IdrisData{)}}

||| Plug a secure computation into a less secure computation
||| @ flow evidence that l may flow to l'
||| @ dio secure computation to plug into insecure computation
\IdrisFunction{plug} : \IdrisType{\IdrisType{Poset}} \IdrisBound{\IdrisBound{labelType}}
    => \{\IdrisBound{l},\IdrisBound{l'} : \IdrisBound{\IdrisBound{\IdrisBound{\IdrisBound{labelType}}}}\}
    -> (\IdrisBound{dio} : \IdrisType{\IdrisType{DIO}} \IdrisBound{\IdrisBound{l'}} \IdrisBound{\IdrisBound{a}})
    -> \{\IdrisKeyword{auto} \IdrisBound{flow} : \IdrisBound{\IdrisBound{l}} \IdrisFunction{\IdrisFunction{`leq`}} \IdrisBound{\IdrisBound{l'}}\}
    -> \IdrisType{\IdrisType{DIO}} \IdrisBound{\IdrisBound{l}} (\IdrisType{\IdrisType{Labeled}} \IdrisBound{\IdrisBound{l'}} \IdrisBound{\IdrisBound{a}})
\IdrisFunction{\IdrisFunction{plug}} \IdrisBound{dio} = \IdrisFunction{\IdrisFunction{lift}} \IdrisFunction{\IdrisFunction{.}} \IdrisFunction{\IdrisFunction{run}} $ \IdrisBound{\IdrisBound{dio}} \IdrisFunction{\IdrisFunction{>>=}} \IdrisFunction{\IdrisFunction{pure}} \IdrisFunction{\IdrisFunction{.}} \IdrisData{\IdrisData{MkLabeled}}
\end{Verbatim}

\begin{Verbatim}[commandchars=\\\{\}, numbers = left, frame=single, label=Poset.idr]
\IdrisKeyword{module} DepSec.Poset

% access public export
% default \IdrisKeyword{total}
% hide Prelude.Monad.join

||| Verified partial ordering
\IdrisKeyword{interface} \IdrisType{\IdrisFunction{\IdrisFunction{\IdrisFunction{\IdrisFunction{\IdrisFunction{\IdrisFunction{\IdrisFunction{\IdrisFunction{\IdrisData{\IdrisData{\IdrisData{\IdrisData{\IdrisData{\IdrisData{\IdrisData{\IdrisData{\IdrisType{\IdrisType{\IdrisType{\IdrisType{\IdrisType{\IdrisType{\IdrisType{\IdrisType{\IdrisType{\IdrisType{\IdrisType{\IdrisType{\IdrisBound{\IdrisBound{\IdrisBound{\IdrisBound{\IdrisBound{\IdrisBound{\IdrisBound{\IdrisBound{\IdrisBound{\IdrisBound{\IdrisBound{\IdrisBound{\IdrisBound{\IdrisBound{\IdrisBound{\IdrisBound{\IdrisBound{\IdrisBound{\IdrisBound{\IdrisBound{\IdrisBound{\IdrisBound{\IdrisBound{\IdrisBound{\IdrisBound{Poset \IdrisType{\IdrisType{\IdrisType{\IdrisType{\IdrisType{\IdrisType{\IdrisType{\IdrisBound{a}}}}}}}}}}}}}}}}}}}}}}}}}}}}}}}}}}}}}}}}}}}}}}}}}}}}}}}}}}}}}} \IdrisKeyword{where}
  \IdrisFunction{leq} : \IdrisBound{\IdrisBound{a}} -> \IdrisBound{\IdrisBound{a}} -> \IdrisType{Type}
  \IdrisFunction{reflexive} :  (\IdrisBound{x} : \IdrisBound{\IdrisBound{a}}) -> \IdrisBound{\IdrisBound{x}} \IdrisFunction{\IdrisFunction{\IdrisBound{`leq`}}} \IdrisBound{\IdrisBound{x}}
  \IdrisFunction{antisymmetric} : (\IdrisBound{x}, \IdrisBound{y} : \IdrisBound{\IdrisBound{\IdrisBound{\IdrisBound{a}}}}) -> \IdrisBound{\IdrisBound{x}} \IdrisFunction{\IdrisFunction{\IdrisBound{`leq`}}} \IdrisBound{\IdrisBound{y}} -> \IdrisBound{\IdrisBound{y}} \IdrisFunction{\IdrisFunction{\IdrisBound{`leq`}}} \IdrisBound{\IdrisBound{x}} -> \IdrisBound{\IdrisBound{x}} \IdrisType{=} \IdrisBound{\IdrisBound{y}}
  \IdrisFunction{transitive} : (\IdrisBound{x}, \IdrisBound{y}, \IdrisBound{z} : \IdrisBound{\IdrisBound{\IdrisBound{\IdrisBound{\IdrisBound{\IdrisBound{a}}}}}}) -> \IdrisBound{\IdrisBound{x}} \IdrisFunction{\IdrisFunction{\IdrisBound{`leq`}}} \IdrisBound{\IdrisBound{y}} -> \IdrisBound{\IdrisBound{y}} \IdrisFunction{\IdrisFunction{\IdrisBound{`leq`}}} \IdrisBound{\IdrisBound{z}} -> \IdrisBound{\IdrisBound{x}} \IdrisFunction{\IdrisFunction{\IdrisBound{`leq`}}} \IdrisBound{\IdrisBound{z}}
\end{Verbatim}

\begin{Verbatim}[commandchars=\\\{\}, numbers = left, frame=single, label=Lattice.idr]
\IdrisKeyword{module} DepSec.Lattice

\IdrisKeyword{import} \IdrisKeyword{public} DepSec.Poset

% access public export
% hide Prelude.Monad.join
% default \IdrisKeyword{total}

||| Verified join semilattice
\IdrisKeyword{interface} \IdrisType{\IdrisFunction{\IdrisFunction{\IdrisFunction{\IdrisFunction{\IdrisFunction{\IdrisFunction{\IdrisFunction{\IdrisFunction{\IdrisData{\IdrisData{\IdrisData{\IdrisData{\IdrisData{\IdrisData{\IdrisData{\IdrisData{\IdrisType{\IdrisType{\IdrisType{\IdrisType{\IdrisType{\IdrisType{\IdrisType{\IdrisType{\IdrisType{\IdrisType{\IdrisType{\IdrisType{\IdrisBound{\IdrisBound{\IdrisBound{\IdrisBound{\IdrisBound{\IdrisBound{\IdrisBound{\IdrisBound{\IdrisBound{\IdrisBound{\IdrisBound{\IdrisBound{\IdrisBound{\IdrisBound{\IdrisBound{\IdrisBound{\IdrisBound{\IdrisBound{\IdrisBound{\IdrisBound{\IdrisBound{\IdrisBound{\IdrisBound{\IdrisBound{\IdrisBound{JoinSemilattice \IdrisType{\IdrisType{\IdrisType{\IdrisType{\IdrisType{\IdrisType{\IdrisType{\IdrisBound{a}}}}}}}}}}}}}}}}}}}}}}}}}}}}}}}}}}}}}}}}}}}}}}}}}}}}}}}}}}}}}} \IdrisKeyword{where}
  \IdrisFunction{join} : \IdrisBound{\IdrisBound{a}} -> \IdrisBound{\IdrisBound{a}} -> \IdrisBound{\IdrisBound{a}}
  \IdrisFunction{joinAssociative} : (\IdrisBound{x}, \IdrisBound{y}, \IdrisBound{z} : \IdrisBound{\IdrisBound{\IdrisBound{\IdrisBound{\IdrisBound{\IdrisBound{a}}}}}})
                 -> \IdrisBound{\IdrisBound{x}} \IdrisFunction{\IdrisFunction{\IdrisBound{`join`}}} (\IdrisBound{\IdrisBound{y}} \IdrisFunction{\IdrisFunction{\IdrisBound{`join`}}} \IdrisBound{\IdrisBound{z}}) \IdrisType{=} (\IdrisBound{\IdrisBound{x}} \IdrisFunction{\IdrisFunction{\IdrisBound{`join`}}} \IdrisBound{\IdrisBound{y}}) \IdrisFunction{\IdrisFunction{\IdrisBound{`join`}}} \IdrisBound{\IdrisBound{z}}
  \IdrisFunction{joinCommutative} : (\IdrisBound{x}, \IdrisBound{y} : \IdrisBound{\IdrisBound{\IdrisBound{\IdrisBound{a}}}})    -> \IdrisBound{\IdrisBound{x}} \IdrisFunction{\IdrisFunction{\IdrisBound{`join`}}} \IdrisBound{\IdrisBound{y}} \IdrisType{=} \IdrisBound{\IdrisBound{y}} \IdrisFunction{\IdrisFunction{\IdrisBound{`join`}}} \IdrisBound{\IdrisBound{x}}
  \IdrisFunction{joinIdempotent}  : (\IdrisBound{x} : \IdrisBound{\IdrisBound{a}})       -> \IdrisBound{\IdrisBound{x}} \IdrisFunction{\IdrisFunction{\IdrisBound{`join`}}} \IdrisBound{\IdrisBound{x}} \IdrisType{=} \IdrisBound{\IdrisBound{x}}

||| A well defined join induces a partial ordering.
\IdrisKeyword{implementation} \IdrisType{\IdrisType{\IdrisType{\IdrisType{\IdrisType{\IdrisType{\IdrisType{\IdrisType{\IdrisType{\IdrisType{\IdrisFunction{\IdrisFunction{\IdrisFunction{\IdrisFunction{\IdrisFunction{\IdrisFunction{\IdrisFunction{\IdrisFunction{\IdrisFunction{\IdrisFunction{\IdrisFunction{\IdrisFunction{\IdrisFunction{\IdrisFunction{\IdrisFunction{\IdrisFunction{\IdrisFunction{\IdrisFunction{\IdrisFunction{\IdrisFunction{\IdrisFunction{\IdrisFunction{\IdrisData{\IdrisData{\IdrisBound{\IdrisBound{\IdrisBound{\IdrisBound{\IdrisBound{\IdrisBound{\IdrisBound{\IdrisBound{\IdrisBound{\IdrisBound{\IdrisBound{\IdrisBound{\IdrisBound{\IdrisBound{\IdrisBound{\IdrisBound{\IdrisBound{\IdrisBound{\IdrisBound{\IdrisBound{\IdrisBound{\IdrisBound{\IdrisBound{JoinSemilattice \IdrisBound{\IdrisBound{\IdrisBound{\IdrisBound{\IdrisBound{\IdrisBound{\IdrisBound{\IdrisBound{\IdrisBound{\IdrisBound{a}}}}}}}}}} => \IdrisType{\IdrisType{Poset}} \IdrisBound{\IdrisBound{\IdrisBound{\IdrisBound{\IdrisBound{\IdrisBound{\IdrisBound{\IdrisBound{\IdrisBound{\IdrisBound{\IdrisBound{\IdrisBound{\IdrisBound{\IdrisBound{\IdrisBound{\IdrisBound{\IdrisBound{\IdrisBound{\IdrisBound{\IdrisBound{\IdrisBound{\IdrisBound{a}}}}}}}}}}}}}}}}}}}}}}}}}}}}}}}}}}}}}}}}}}}}}}}}}}}}}}}}}}}}}}}}}}}}}}}}}}}}}}} \IdrisKeyword{where}
  \IdrisFunction{\IdrisFunction{\IdrisBound{leq \IdrisBound{x} \IdrisBound{y}}}} = (\IdrisBound{\IdrisBound{x}} \IdrisFunction{\IdrisFunction{`join`}} \IdrisBound{\IdrisBound{y}} \IdrisType{=} \IdrisBound{\IdrisBound{y}})
  \IdrisFunction{\IdrisFunction{\IdrisBound{reflexive}}} = \IdrisFunction{\IdrisFunction{joinIdempotent}}
  \IdrisFunction{\IdrisFunction{\IdrisBound{antisymmetric \IdrisBound{x} \IdrisBound{y} \IdrisBound{lexy} \IdrisBound{leyx}}}} =
    \IdrisKeyword{rewrite} \IdrisFunction{\IdrisFunction{\IdrisFunction{\IdrisFunction{\IdrisFunction{\IdrisFunction{\IdrisBound{sym $ \IdrisBound{\IdrisBound{\IdrisBound{\IdrisBound{lexy}}}}}}}}}}} \IdrisKeyword{in}
    \IdrisKeyword{rewrite} \IdrisFunction{\IdrisFunction{\IdrisFunction{\IdrisFunction{\IdrisFunction{\IdrisFunction{\IdrisBound{joinCommutative \IdrisBound{\IdrisBound{\IdrisBound{\IdrisBound{x}}}} \IdrisBound{\IdrisBound{\IdrisBound{\IdrisBound{y}}}}}}}}}}} \IdrisKeyword{in}
    \IdrisKeyword{rewrite} \IdrisFunction{\IdrisFunction{\IdrisFunction{\IdrisFunction{\IdrisFunction{\IdrisFunction{\IdrisBound{sym $ \IdrisBound{\IdrisBound{\IdrisBound{\IdrisBound{leyx}}}}}}}}}}} \IdrisKeyword{in} \IdrisData{\IdrisData{Refl}}
  \IdrisFunction{\IdrisFunction{\IdrisBound{transitive \IdrisBound{x} \IdrisBound{y} \IdrisBound{z} \IdrisBound{lexy} \IdrisBound{leyx}}}} =
    \IdrisKeyword{rewrite} \IdrisFunction{\IdrisFunction{\IdrisFunction{\IdrisFunction{\IdrisFunction{\IdrisFunction{\IdrisBound{sym $ \IdrisBound{\IdrisBound{\IdrisBound{\IdrisBound{leyx}}}}}}}}}}} \IdrisKeyword{in}
    \IdrisKeyword{rewrite} \IdrisFunction{\IdrisFunction{\IdrisFunction{\IdrisFunction{\IdrisFunction{\IdrisFunction{\IdrisBound{joinAssociative \IdrisBound{\IdrisBound{\IdrisBound{\IdrisBound{x}}}} \IdrisBound{\IdrisBound{\IdrisBound{\IdrisBound{y}}}} \IdrisBound{\IdrisBound{\IdrisBound{\IdrisBound{z}}}}}}}}}}} \IdrisKeyword{in}
    \IdrisKeyword{rewrite} \IdrisFunction{\IdrisFunction{\IdrisFunction{\IdrisFunction{\IdrisFunction{\IdrisFunction{\IdrisBound{sym $ \IdrisBound{\IdrisBound{\IdrisBound{\IdrisBound{lexy}}}}}}}}}}} \IdrisKeyword{in} \IdrisData{\IdrisData{Refl}}

||| A join-semilattice with an identity element (the lattices' bottom)
||| of the join operation.
\IdrisKeyword{interface} \IdrisType{\IdrisType{\IdrisType{\IdrisType{\IdrisFunction{\IdrisFunction{\IdrisFunction{\IdrisFunction{\IdrisFunction{\IdrisFunction{\IdrisData{\IdrisData{\IdrisData{\IdrisData{\IdrisData{\IdrisData{\IdrisType{\IdrisType{\IdrisType{\IdrisType{\IdrisType{\IdrisType{\IdrisType{\IdrisType{\IdrisType{\IdrisType{\IdrisBound{\IdrisBound{\IdrisBound{\IdrisBound{\IdrisBound{\IdrisBound{\IdrisBound{\IdrisBound{\IdrisBound{\IdrisBound{\IdrisBound{\IdrisBound{JoinSemilattice \IdrisBound{\IdrisBound{\IdrisBound{\IdrisBound{a}}}} => \IdrisType{BoundedJoinSemilattice} \IdrisType{\IdrisType{\IdrisType{\IdrisType{\IdrisType{\IdrisType{\IdrisBound{a}}}}}}}}}}}}}}}}}}}}}}}}}}}}}}}}}}}}}}}}}}}}} \IdrisKeyword{where}
  \IdrisFunction{Bottom} : \IdrisBound{\IdrisBound{a}}
  \IdrisFunction{bottomUnitaryElement} : (\IdrisBound{e} : \IdrisBound{\IdrisBound{a}}) -> \IdrisBound{\IdrisBound{e}} \IdrisFunction{\IdrisFunction{`join`}} \IdrisFunction{\IdrisFunction{\IdrisBound{Bottom}}} \IdrisType{=} \IdrisBound{\IdrisBound{e}}
\end{Verbatim}

\begin{Verbatim}[commandchars=\\\{\}, numbers = left, frame=single, label=Ref.idr]
\IdrisKeyword{module} DepSec.Ref

\IdrisKeyword{import} \IdrisKeyword{public} DepSec.Labeled
\IdrisKeyword{import} \IdrisKeyword{public} DepSec.DIO
\IdrisKeyword{import} Data.IORef

% access export

||| Data type for secure references.
\IdrisKeyword{data} \IdrisType{SecRef} : (\IdrisBound{l} : \IdrisBound{\IdrisBound{labelType}}) -> (\IdrisBound{valueType} : \IdrisType{Type}) -> \IdrisType{Type} \IdrisKeyword{where}
  |||TCB
  \IdrisData{MkSecRef} : (\IdrisBound{ref} : \IdrisType{\IdrisType{IORef}} \IdrisBound{\IdrisBound{a}}) -> \IdrisType{\IdrisType{SecRef}} \IdrisBound{\IdrisBound{l}} \IdrisBound{\IdrisBound{a}}

||| Creating a reference to a labeled value.
||| @ flow evidence that l may flow to l'
||| @ flow' evidence that l' may flow to l''
||| @ value The initial value for the reference.
\IdrisFunction{newRef} : \IdrisType{\IdrisType{Poset}} \IdrisBound{\IdrisBound{labelType}}
      => \{\IdrisBound{l}, \IdrisBound{l'}, \IdrisBound{l''} : \IdrisBound{\IdrisBound{\IdrisBound{\IdrisBound{\IdrisBound{\IdrisBound{labelType}}}}}}\}
      -> \{\IdrisKeyword{auto} \IdrisBound{flow} : \IdrisBound{\IdrisBound{l}} \IdrisFunction{\IdrisFunction{`leq`}} \IdrisBound{\IdrisBound{l'}}\}
      -> \{\IdrisKeyword{auto} \IdrisBound{flow'} : \IdrisBound{\IdrisBound{l'}} \IdrisFunction{\IdrisFunction{`leq`}} \IdrisBound{\IdrisBound{l''}}\}
      -> (\IdrisBound{value} : \IdrisType{\IdrisType{Labeled}} \IdrisBound{\IdrisBound{l'}} \IdrisBound{\IdrisBound{a}})
      -> \IdrisType{\IdrisType{DIO}} \IdrisBound{\IdrisBound{l}} (\IdrisType{\IdrisType{SecRef}} \IdrisBound{\IdrisBound{l''}} \IdrisBound{\IdrisBound{a}})
\IdrisFunction{\IdrisFunction{newRef}} (\IdrisData{\IdrisData{MkLabeled}} \IdrisBound{v})
  = \IdrisFunction{\IdrisFunction{lift}} $ \IdrisFunction{\IdrisFunction{newIORef}} \IdrisBound{\IdrisBound{v}} \IdrisFunction{\IdrisFunction{>>=}} \IdrisFunction{\IdrisFunction{pure}} \IdrisFunction{\IdrisFunction{.}} \IdrisData{\IdrisData{MkSecRef}}

||| Reading a secure reference.
||| @ flow evidence that l may flow to l'
||| @ ref The reference which we wish to read.
\IdrisFunction{readRef} : \IdrisType{\IdrisType{Poset}} \IdrisBound{\IdrisBound{labelType}}
       => \{\IdrisBound{l}, \IdrisBound{l'} : \IdrisBound{\IdrisBound{\IdrisBound{\IdrisBound{labelType}}}}\}
       -> \{\IdrisKeyword{auto} \IdrisBound{flow} : \IdrisBound{\IdrisBound{l}} \IdrisFunction{\IdrisFunction{`leq`}} \IdrisBound{\IdrisBound{l'}}\}
       -> (\IdrisBound{ref}  : \IdrisType{\IdrisType{SecRef}} \IdrisBound{\IdrisBound{l'}} \IdrisBound{\IdrisBound{a}})
       -> \IdrisType{\IdrisType{DIO}} \IdrisBound{\IdrisBound{l}} (\IdrisType{\IdrisType{Labeled}} \IdrisBound{\IdrisBound{l'}} \IdrisBound{\IdrisBound{a}})
\IdrisFunction{\IdrisFunction{readRef}} (\IdrisData{\IdrisData{MkSecRef}} \IdrisBound{ioRef})
  = \IdrisFunction{\IdrisFunction{lift}} $ \IdrisFunction{\IdrisFunction{map}} \IdrisData{\IdrisData{MkLabeled}} $ \IdrisFunction{\IdrisFunction{readIORef}} \IdrisBound{\IdrisBound{ioRef}}

||| Wrting a labeled value to a secure reference.
||| @ flow evidence that l may flow to l'
||| @ flow' evidence that l' may flow to l''
||| @ ref The reference which we wish to write to.
||| @ content The content which we wish too read.
\IdrisFunction{writeRef} : \IdrisType{\IdrisType{Poset}} \IdrisBound{\IdrisBound{labelType}}
        => \{\IdrisBound{l}, \IdrisBound{l'}, \IdrisBound{l''} : \IdrisBound{\IdrisBound{\IdrisBound{\IdrisBound{\IdrisBound{\IdrisBound{labelType}}}}}}\}
        -> \{\IdrisKeyword{auto} \IdrisBound{flow} : \IdrisBound{\IdrisBound{l}} \IdrisFunction{\IdrisFunction{`leq`}} \IdrisBound{\IdrisBound{l'}}\}
        -> \{\IdrisKeyword{auto} \IdrisBound{flow'} : \IdrisBound{\IdrisBound{l'}} \IdrisFunction{\IdrisFunction{`leq`}} \IdrisBound{\IdrisBound{l''}}\}
        -> (\IdrisBound{ref} : \IdrisType{\IdrisType{SecRef}} \IdrisBound{\IdrisBound{l''}} \IdrisBound{\IdrisBound{a}})
        -> (\IdrisBound{content} : \IdrisType{\IdrisType{Labeled}} \IdrisBound{\IdrisBound{l'}} \IdrisBound{\IdrisBound{a}})
        -> \IdrisType{\IdrisType{DIO}} \IdrisBound{\IdrisBound{l}} \IdrisType{()}
\IdrisFunction{\IdrisFunction{writeRef}} (\IdrisData{\IdrisData{MkSecRef}} \IdrisBound{ioRef}) (\IdrisData{\IdrisData{MkLabeled}} \IdrisBound{content})
  = \IdrisFunction{\IdrisFunction{lift}} $ \IdrisFunction{\IdrisFunction{writeIORef}} \IdrisBound{\IdrisBound{ioRef}} \IdrisBound{\IdrisBound{content}}
\end{Verbatim}

\begin{Verbatim}[commandchars=\\\{\}, numbers = left, frame=single, label=File.idr]
\IdrisKeyword{module} DepSec.File

\IdrisKeyword{import} \IdrisKeyword{public} DepSec.DIO
\IdrisKeyword{import} \IdrisKeyword{public} DepSec.Labeled

%access export

||| Secure file
\IdrisKeyword{data} \IdrisType{SecFile} : \{\IdrisBound{label} : \IdrisType{Type}\} -> (\IdrisBound{l} : \IdrisBound{label}) -> \IdrisType{Type} \IdrisKeyword{where}
  ||| TCB
  \IdrisData{MkSecFile} : (\IdrisBound{path} : \IdrisType{String}) -> \IdrisType{SecFile} \IdrisBound{l}

||| Make a secure file from string
||| TCB
||| @ path path to file
\IdrisFunction{makeFile} : (\IdrisBound{path} : \IdrisType{String}) -> \IdrisType{SecFile} \IdrisBound{l}
\IdrisFunction{makeFile} = \IdrisData{MkSecFile}

||| Read a secure file
||| @ flow evidence that l may flow to l'
||| @ file secure file to read
\IdrisFunction{readFile} : \IdrisType{Poset} \IdrisBound{labelType}
        => \{\IdrisBound{l},\IdrisBound{l'} : \IdrisBound{labelType}\}
        -> \{\IdrisKeyword{auto} \IdrisBound{flow} : \IdrisBound{l} \IdrisFunction{`leq`} \IdrisBound{l'}\}
        -> (\IdrisBound{file} : \IdrisType{SecFile} \IdrisBound{l'})
        -> \IdrisType{DIO} \IdrisBound{l} (\IdrisType{Labeled} \IdrisBound{l'} (\IdrisType{Either} \IdrisType{FileError} \IdrisType{String}))
\IdrisFunction{readFile} (\IdrisData{MkSecFile} \IdrisBound{path}) = \IdrisFunction{lift} $ \IdrisFunction{map} \IdrisData{MkLabeled} $ \IdrisFunction{readFile} \IdrisBound{path}

||| Write to a secure file
||| @ file secure file to write to
||| @ flow evidence that l may flow to l'
||| @ flow' evidence that l' may flow to l''
||| @ content labeled content to write
\IdrisFunction{writeFile} : \IdrisType{Poset} \IdrisBound{labelType}
         => \{\IdrisBound{l},\IdrisBound{l'},\IdrisBound{l''} : \IdrisBound{labelType}\}
         -> \{\IdrisKeyword{auto} \IdrisBound{flow} : \IdrisBound{l} \IdrisFunction{`leq`} \IdrisBound{l'}\}
         -> (\IdrisBound{file} : \IdrisType{SecFile} \IdrisBound{l''})
         -> \{\IdrisKeyword{auto} \IdrisBound{flow'} : \IdrisBound{l'} \IdrisFunction{`leq`} \IdrisBound{l''}\}
         -> (\IdrisBound{content} : \IdrisType{Labeled} \IdrisBound{l'} \IdrisType{String})
         -> \IdrisType{DIO} \IdrisBound{l} (\IdrisType{Labeled} \IdrisBound{l''} (\IdrisType{Either} \IdrisType{FileError} \IdrisType{()}))
\IdrisFunction{writeFile} (\IdrisData{MkSecFile} \IdrisBound{path}) (\IdrisData{MkLabeled} \IdrisBound{content})
  = \IdrisFunction{lift} $ \IdrisFunction{map} \IdrisData{MkLabeled} $ \IdrisFunction{writeFile} \IdrisBound{path} \IdrisBound{content}
\end{Verbatim}


}{}

\end{document}